\documentclass[acmsmall]{acmart}

\newcommand*{\Draft}{}%

\acmConference{Technical Report}{2022}{arXiv}
\acmYear{2022}
\acmISBN{}
\acmDOI{}
\startPage{1}

\setcopyright{none}

\usepackage{booktabs}   %
\usepackage{subcaption} %

\newcommand\shorttitle{\relax}
\usepackage{import}
\usepackage[T1]{fontenc}
\usepackage[utf8]{inputenc}

\usepackage{amsmath}
\usepackage{amssymb}

\usepackage{thmtools}
\usepackage{thm-restate}

\usepackage{xcolor}
\usepackage{hyperref}
\hypersetup{
  hidelinks,
  pdfcreator={LaTeX via pandoc}}
\newtheorem{thm}{Theorem}[section]

\newtheorem{lemma}[thm]{Lemma}

\usepackage{listings}
\lstdefinelanguage{kotlinish}{
  basicstyle=\small\ttfamily,
  keywords={val,if, then, handle, return, def, match, case, new, type,
    effect, pretype, class, interface, open, sealed, final, in, out, extends, infix, else, when, fun, is},
  keywordstyle=\bfseries,
  sensitive=false,
  comment=[l]{//},
  commentstyle=\color{gray},
  stringstyle=\color{gray}, %
  morestring=[b]',
  morestring=[b]",
  moredelim=**[is][\color{red}]{@}{@},
}
\lstdefinelanguage{scalaish}{
  basicstyle=\small\ttfamily,
  keywords={
    val, var, def, type, effect, pretype,
    if, then, else, in, handle, return, match, case, new,
    trait, class, extends, sealed, final, this,
    infix,
    instanceof,in,out},
  keywordstyle=\bfseries,
  sensitive=false,
  comment=[l]{//},
  commentstyle=\color{gray},
  stringstyle=\color{gray}, %
  morestring=[b]',
  morestring=[b]"
}

\makeatletter
\lst@AddToHook{TextStyle}{\let\lst@basicstyle\ttfamily}
\makeatother

\usepackage{bcprules}
\usepackage[shorthand]{semantic}
\usepackage{mathpartir}
\usepackage{multicol}
\usepackage{placeins}
\usepackage{cprotect}

\usepackage{bussproofs}
\usepackage{tikz}
\usetikzlibrary{cd}
\usepackage{xspace}
\usepackage{commands}
\usepackage{nicefrac}

\renewcommand{\backrefalt}[4]{\textcolor{gray}{$\hookrightarrow$~{\ifcase #1 \or page\else pages\fi}~#2}} %

\ifdefined \Draft
\newcommand{\lionel}[1]{\textcolor{red!70!black}{[Lionel:] #1}}
\newcommand{\alex}[1]{\textcolor{red!70!black}{[Alex:] #1}}
\newcommand{\radek}[1]{\textcolor{red!70!black}{[Radek:] #1}}
\newcommand{\yichen}[1]{\textcolor{red!70!black}{[Yichen:] #1}}

\newcommand{\todo}[1]{\textcolor{red}{[TODO]: #1}}

\else
\newcommand{\lionel}[1]{}
\newcommand{\alex}[1]{}
\newcommand{\radek}[1]{}
\newcommand{\yichen}[1]{}

\newcommand{\todo}[1]{}

\fi

\newcommand\calculus{cDOT\xspace}
\newcommand\Xis{$\lambda_{2,G\mu}$}

\newcommand{\E}{\ensuremath{\mathcal{E}}}
\newcommand{\T}{\ensuremath{\mathcal{T}}}
\newcommand{\letp}[3]{\ensuremath{\textbf{let } #1 \, = \, #2 \textbf{ in } #3 }}

\newcommand\kLet{\operatorname{\mathbf{let}}}
\newcommand\kIn{\operatorname{\mathbf{in}}}
\newcommand\kCase{\operatorname{\mathbf{case}}}
\newcommand\kOf{\operatorname{\mathbf{of}}}

\newcommand\nLet{\kLet}
\newcommand\nIn{\kIn}

\newcommand\ident[1]{\mathrm{#1}}
\newcommand\constant[1]{\mathrm{#1}}

\newcommand\Int{\constant{Int}}

\newcommand\env{\ident{env}}
\newcommand\lib{\ident{lib}}
\newcommand\pmatch{\ident{pmatch}}
\newcommand\data{\ident{data}}

\begin{document}

\title[A case for DOT]{A {\tt case} for DOT: Theoretical Foundations for Objects with Pattern Matching and GADT-Style Reasoning}

\newcommand{\atsign}{@}

\lstMakeShortInline[columns=fixed,language=scalaish]@

\author{Aleksander Boruch-Gruszecki}
\email{aleksander.boruch-gruszecki\atsign epfl.ch}
\orcid{0000-0001-5769-6684}
\affiliation{%
  \institution{EPFL}
  \country{Switzerland}
}
\author{Radosław Waśko}
\authornote{Both authors contributed equally to this research.}
\email{wasko.radek\atsign gmail.com}
\orcid{0000-0001-5539-8245}
\affiliation{%
  \institution{University of Warsaw}
  \country{Poland}
}
\author{Yichen Xu}
\authornotemark[1]
\email{yichen.xu\atsign epfl.ch}
\orcid{0000-0003-2089-6767}
\affiliation{%
  \institution{Beijing University of Posts and Telecommunications}
  \country{China}
}
\author{Lionel Parreaux}
\authornote{Corresponding author.}
\email{parreaux\atsign ust.hk}
\orcid{0000-0002-8805-0728}
\affiliation{%
  \institution{HKUST}
  \country{Hong Kong, China}
}

\begin{abstract}
  Many programming languages in the OO tradition now support pattern matching in some form.
  Historical examples include Scala and Ceylon,
  with the more recent additions of Java, Kotlin, TypeScript, and Flow.
  But pattern matching on generic class hierarchies currently results in puzzling type errors in most of these languages.
  Yet this combination of features occurs naturally in many scenarios,
  such as when manipulating typed ASTs.
  To support it properly, compilers need to implement a form of
  \emph{subtyping reconstruction}:
  the ability to reconstruct subtyping information uncovered at runtime during pattern matching.
  We introduce cDOT, a new calculus in the family of Dependent Object Types (DOT)
  intended to serve as a formal foundation for subtyping reconstruction.
  Being descended from pDOT, itself a formal foundation for Scala,
  cDOT can be used to encode
  advanced object-oriented features such as generic inheritance, type constructor variance,
  F-bounded polymorphism,
  and first-class recursive modules.
  We demonstrate that subtyping reconstruction subsumes GADTs %
  by encoding $\lambda_{2,G\mu}$, a classical constraint-based GADT calculus,
  into cDOT.
\end{abstract}

\ccsdesc[500]{Software and its engineering~Patterns}
\ccsdesc[500]{Software and its engineering~Classes and objects}
\ccsdesc[300]{Software and its engineering~Object oriented languages}
\ccsdesc[500]{Theory of computation~Type structures}

\keywords{DOT, pattern matching, GADT, classes, type systems}  %

\maketitle

\section{Introduction}

In recent years, many programming languages in the object-oriented (OO) tradition
started incorporating support for functional programming idioms,
such as lambda expressions and pattern matching.
Notably, Java gained a simple form of pattern matching
  (limited to single `@instanceof@' type patterns) in JDK 16,
and work on extending this feature to more full-fledged pattern matching support is ongoing,
with an implementation available as a feature preview in JDK 17.
Version 7.0 of C$^\sharp$ introduced support for type patterns in @switch@ statements,
and version 8.0 introduced @switch@ expressions.
Kotlin similarly supports a restricted form of pattern matching
through its `\lstinline[language=kotlinish]!when! block' feature.
Other widely-used statically-typed languages such as Typescript and Flow
support a variant of pattern matching
called \emph{flow-sensitive} or \emph{occurrence} typing.
Finally, languages like Scala and Ceylon,
which merge the object-oriented and functional paradigms in a single language,
have supported pattern matching from the start.
\begin{figure}[h!]
\begin{lstlisting}[language=kotlinish]
  sealed class Expr<A>()
  class IntLit  (val value: Int) : Expr<Int>()
  class ... // other subclasses
\end{lstlisting}

\begin{lstlisting}[language=kotlinish]
  fun <T>eval(expr: Expr<T>): T =
    when (expr) {
      is IntLit -> expr.value
      is ...    -> // other cases
    }
\end{lstlisting}
  \caption{Example: evaluating simple conditional expressions in Kotlin}
  \label{fig:motiv-example}
\end{figure}

These new pattern matching approaches all rely on the idea
to statically \emph{refine} the type of a ``scrutinee'' (the object which is being pattern-matched)
based on its runtime shape.
Typically, they allow one to test whether a given value is an instance of some class @C@,
and in the branch where it is known to be,
use the value as such, without requiring an unsafe cast to @C@.

But combining generic classes with this form of refining type tests
leads to puzzling cases where code that should seemingly compile is rejected.
Consider the example in \autoref{fig:motiv-example}, where
we use Kotlin to define @Expr@, a class that represents simple arithmetic expressions,
and @eval@, a function which evaluates said expressions.
@eval@ takes an argument of type @Expr<T>@ and returns a @T@.
Intuitively, the single case shown in \autoref{fig:motiv-example} is correct.
Yet, the latest version of the Kotlin compiler to date (as of April 2022) rejects it,
saying that the type of @expr.value@ is @Int@ and not @T@.
But our intuition is not wrong:
we have checked that @expr@ is an instance of @IntLit@,
which means that it is simultaneously an @Expr<Int>@ and an @Expr<T>@,
which in turn means that @T@ is @Int@.\footnote{
  This reasoning is based on the fact that in languages like Kotlin, Java, Scala, C$^\sharp$, etc.
  a derived class may inherit from a given base class at most once, with specific type arguments.
  It would not work in C++,
  where a class could inherit from two instantiations of the same class, such as
  {\tt Expr<int>} and {\tt Expr<bool>}, which C++ considers to be two unrelated classes.
}
In other words, pattern matching on a value of type @Expr<T>@
lets us discover that @T@ is the same type as @Int@ in the @IntLit@ branch,
which should let us conclude that @eval@ is well-typed.

Generic class hierarchies \cite{Bracha1998}
have become pervasive in real-world code bases in Java \cite{parnin2013adoption}
and many other OO languages.
Therefore, it seems inevitable that
the adoption of new pattern-matching capabilities in these languages
will lead to problems like the one outlined above.
Yet, until now,
this important topic saw surprisingly little formal or informal investigation.
To the best of our knowledge, Scala is the only object-oriented language
whose implementation currently attempts to allow such pattern matches,
but formal investigations into how it does this have remained preliminary
\cite{parreaux2019,xu2021},
and previous work on the topic 
\cite{Burak06:variance-gen-constraints-csharp,Burak07:matching-objects-patterns,Kennedy05:gadts-oo}
was restricted to
simple type systems.
Furthermore, none of the previous approaches were mechanically verified.
This is a problematic state of affairs because
such reasoning
turns out to be
very subtle and hard to implement correctly \cite{giarrusso2013,parreaux2019}.

The specific example of \autoref{fig:motiv-example} will appear familiar to readers already aware of
\emph{generalized algebraic data types} (GADTs),
available in languages such as Haskell and OCaml,
which also allow discovering hidden types and type relationships by matching on values.
However, traditional GADTs operate in very different type systems than the ones of object-oriented languages.
Traditional GADT constructors do not introduce a type of their own, and are only attached to a single parent type,
forming a flat and closed hierarchy, contrasting with OOP's nested and possibly open class hierarchies.
Moreover, traditional GADT implementations only reason about 
type \emph{equations},\footnote{
  Interestingly, only reasoning about type \emph{equations} leads to some paradoxes in languages like OCaml which do support (an explicit form of) subtyping,
  which was previously investigated by \citet{scherer2013}.
}
whereas inheritance hierarchies,
particularly when they involve type constructor variance,
require us to reason about type \emph{inequations} (i.e., \emph{subtyping} relationships).

In this paper, we lay out formal foundations for {GADT-style reasoning}
in object-oriented programming languages.
We introduce \textbf{\emph{subtyping reconstruction}}, a technique that allows
pattern-matching functions like @eval@ to be recognized as well-typed.
We base our approach on a new and improved version of DOT (the Dependant Object Types calculus)
extended with a @case@ construct for pattern matching,
and thus named \emph{\calculus}.
The expressiveness of this foundational calculus allows us to establish our reasoning principles
in the presence of advanced object-oriented type system features,
such as generic inheritance, type constructor variance, F-bounded polymorphism,
and first-class recursive module types.
This is important because such features are often present
in practical programming languages like Java, C$^\sharp$, and Scala,
and they significantly complicate the task at hand.
We also rigorously explore the connection between our GADT-style reasoning
and traditional GADTs, formally demonstrating that our approach
encompasses traditional GADT reasoning,
since the latter can be encoded in the former.

We make the following specific contributions:

\begin{itemize}
  \item We introduce \calculus,
    a formal calculus to serve as a foundation for subtyping reconstruction
    (\autoref{sec:presentation-calculus} and \autoref{sec:formalization}).
    \calculus is based on \pdot \cite{rapoportPathDOTFormalizing2019},
    which is itself an evolution of the original DOT \cite{amin2016},
    which has been proposed as a foundation for statically-typed object-oriented programming
    \cite{martres2022}.
    We provide a mechanized proof of soundness for \calculus. %
  \item We propose a variant of the \Xis{} calculus formalizing traditional GADTs \cite{xi2003}
    (\autoref{sec:variant-of-xi}).
    Our variant removes some idiosyncrasies present in the original formulation of \Xis{},
    notably by making its operational semantics deterministic and
    by making progress hold.
    We also provide a mechanized soundness proof for this modified calculus.
  \item We develop an encoding of our variant of \Xis{} into \calculus{}
    (\autoref{sec:encoding})
    and show that this encoding preserves typing (\autoref{sec:typing-preservation}),
    which demonstrates that
    \calculus{} can express traditional GADT reasoning.
    This goes to show that subtyping reconstruction
    \emph{subsumes}
    (i.e., is at least as powerful as) traditional GADT reasoning.
\end{itemize}

\section{Static Typing of Class-Based Pattern Matching}
\label{sec:problem-solution}

In this section, we present the problem of objects with pattern-matching and GADT-style reasoning,
as well as our proposed solution to it, from a high-level and intuition-focused point of view.

\subsection{A Core View on Pattern Matching}

We define pattern-matching loosely as any conditional expression form
which allows learning about the runtime shape of a value
and consequently extracting information from this shape \emph{in a statically type-safe way}.
For instance, Java's original @instanceof@ construct does not constitute proper pattern matching
because making use of the information it returns (a boolean value) requires the use of an \emph{unsafe} cast,
as in `@if@ @(x@ @instanceof@ @IntLint)@ @return@ @((IntLint)x).value@'.
The problem is that the consistency between the runtime class instance test and
the cast to @IntLit@ is not statically checked, and programmers could (and often do) get it wrong,
for example if they had written
`@if@ @(some_condition@ @||@ @x@ @instanceof@ @IntLint)@ @return@ @((IntLint)x).value@' instead,
possibly resulting in runtime @ClassCastException@ crashes.
On the other hand,
Java's newer construct `@if@ @(x@ @instanceof@ @IntLint@ @y)@ @return@ @y.value@' is a proper pattern matching implementation,
statically known to be safe.

Many pattern matching implementations allow complex nested patterns,
but these can usually be treated as ``syntactic sugar'' over a core representation of pattern matching
that proceeds one level at a time. In an object-oriented language,
this core construct would be matching on the class of a given instance,\footnote{
  Indeed, %
  the Scala compiler essentially desugars its complex pattern matching syntax
  into a succession of instance tests.
}
similar to Java's new `@instanceof@' construct and to Kotlin's `\lstinline[language=kotlinish]!when!'.

\subsection{Uncovering Types and Subtype Relationships}
\label{sec:uncovering-tps-rels}

Consider the @Expr@/@eval@ example from the introduction again,
with more cases filled in:\footnote{
  In Kotlin, the star symbol {\tt *}
  is used as a placeholder when a type argument is unknown.
}

\begin{lstlisting}[language=kotlinish,mathescape]
sealed class Expr<A>()
class IntLit     (val value: Int)                    $\,$: Expr<Int>()
class MkPair<B,C>(val lhs: Expr<B>, val rhs: Expr<C>)$\,$: Expr<Pair<B,C>>()
class First <B,C>(val pair: Expr<Pair<B,C>>)         $\,$: Expr<B>()
class Second<B,C>(val pair: Expr<Pair<B,C>>)         $\,$: Expr<C>()

fun <T>eval(expr: Expr<T>): T =
  when (expr) {
    is IntLit      -> expr.value
    is MkPair<*,*> -> Pair(eval(expr.lhs), eval(expr.rhs))
    is First <T,*> -> eval(expr.pair).first
    is Second<*,T> -> eval(expr.pair).second
  }
\end{lstlisting}

\noindent
As we saw, the @IntLit@ branch of this example should compile
even though @expr.value@ returns an @Int@ where a value of type @T@ is expected:
when we enter the @IntLit@ branch, we discover that @expr@ is simultaneously an @Expr<T>@ and an @Expr<Int>@,
and since @Expr@ is invariant, this can only hold if @T@ and @Int@ are the same type.
While this particular case appears simple, in general this form of reasoning is non-trivial.

Consider the other cases of @eval@,
and notice the use of @*@ to fill in those type parameters we locally do not know.
For instance, having an @Expr<T>@ that is also a @First<B,C>@
tells us that @B@ is @T@, because @First<B,C>@ extends @Expr<B>@;
however, it tells us \emph{nothing} about @C@,
which could be @Int@, @Boolean@, or any other type.
Accordingly, @C@ needs to be treated as an unknown.
Just like for traditional GADTs,
matching on generic classes may uncover {unknown types},
which are thus said to be \emph{existentially quantified}.
To handle these types adequately,
the compiler must treat them as unspecified abstract types,
locally considered to be distinct from all other types in the program
(sometimes referred to as a ``skolem'').
Here again the Kotlin compiler fails us:
it disregards these types, widening the corresponding pattern class types
to @Any@ and @Nothing@.
Typing @eval@ also necessitates
relating types that are %
only \emph{partially} known,
containing references to locally-uncovered type arguments.
This is the case in the @MkPair@ branch: %
the result of the branch should have type @Pair<?B,@ @?C>@,
where @?B@ and @?C@ are the locally-uncovered unknown types associated with @MkPair<?B,@ @?C>@,
and a compiler implementing GADT-style reasoning should recognize that in this branch,
@T@ is the same as @Pair<?B,@ @?C>@, allowing the expression to type check.

The addition of bounded polymorphism \cite{Cardelli1985},
and F-bounded polymorphism \cite{canning1989f} in particular,
which are supported by many object-oriented programming languages (including Java)
and whereby every type parameter may be associated with possibly-cyclic bounds,\footnote{
  A standard example is a polymorphic function working over any type {\tt T}
  where {\tt T} is assumed to be a subtype of {\tt Comparable<T>},
  which means that {\tt T} has a cyclic bound
  -- i.e., {\tt T}'s bound refers to {\tt T} itself.
}
complicates the picture further.
Essentially, it mean that we should not treat locally-uncovered types
as complete unknowns,
but rather as \emph{bounded}, possibly-recursive abstract types.
This is already a strong justification to reach
for the expressive power
of the Dependent Object Types calculus, which features abstract types with recursive bounds.

\subsection{Type Constructor Variance}
\label{sec:ty-ctor-variance}

Languages like Kotlin, Scala, and C$^\sharp$ support a concept known as
\emph{(declaration-site) type constructor variance},
which allows more flexibility in the way %
generic types can be related with each other.
Declaring a class such as @Expr@ to be \emph{covariant},
which is done in Kotlin with the syntax @class@ @Expr<out@ @T>@,
means that if @Apple@ is a subtype of @Fruit@
(written @Apple@ $<:$ @Fruit@),
then @Expr<Apple>@ is itself considered a subtype of @Expr<Fruit>@
(i.e., @Expr<Apple>@ $<:$ @Expr<Fruit>@).
How does defining @Expr@ as covariant impact @eval@?

It turns out @eval@ should still type check even when @Expr@ is made {covariant},
but the reasoning becomes more subtle.
For example, in the @IntLit@ branch, @T@ is no longer the same as @Int@.
Indeed, the @Expr<Int>@ we got in parameter may have been \emph{widened}
to any supertype, such as @Expr<Any>@ (since @Int@ is a subtype of @Any@)
at the time it was passed to @eval@, so @T@ @=@ @Any@ is anther possibility we now have to account for.
So all we can assume here is that @Int@ is a \emph{subtype} of @T@.
Fortunately, this is enough for this branch to type check.
Moreover, similar \emph{subtype}-based reasoning can be applied to the other branches,
allowing @eval@ to type check as a whole.

As another example,
if we discover through pattern matching that some value @v@ is simultaneously a @Comparable<T>@ and a @Comparable<Int>@,
we cannot actually conclude anything about @T@ and @Int@.
Since @Comparable@ is \emph{contra}variant
(meaning @Comparable<S>@ $<:$ @Comparable<T>@ if @T@ $<:$ @S@),
@v@ could for example very well be a @Comparable<Any>@,
since we have both @T@ $<:$ @Any@ and @Int@ $<:$ @Any@, without @T@ and @Int@ being related in any way.
This case differs from the @IntLit@ case of the @eval@ function described above
because in the latter we \emph{know} what is the precise type @A@ that @IntLit@ inherits @Expr<A>@ with.
Assuming a covariant @Expr@,
we would be similarly stuck if all we had was the information that the scrutinee is both of type @Expr<T>@
and of type @Expr<Int>@ --- we would then be unable to related @T@ and @Int@.

In general, type constructor variance makes GADT-style reasoning
for pattern matching in object-oriented languages a lot more difficult.
Indeed, when a pattern class is covariant or contravariant,
we can learn \emph{strictly less} about the way its type parameters
relate with the type of the scrutinee
than in the invariant case.
To make matters worse, some languages like Scala (but unlike Kotlin)
support a feature called \emph{variant inheritance},
whereby a class may inherit from a variant base several times and at \emph{different} type arguments.
For example, the following Kotlin code is illegal, but its equivalent in Scala is legal
and is treated as @Derived2@ inheriting from both @Base<Any>@ \emph{and} @Base<String>@:

\begin{lstlisting}[language=kotlinish]
  interface Base<out T>
  open class Derived1: Base<Any>()
  class Derived2: Derived1(), Base<String>
\end{lstlisting}

\noindent
This has led to paradoxes related to pattern matching
\cite{giarrusso2013,parreaux2019}.

These considerations show that the study of
subtyping reconstruction diverges significantly from that of traditional GADT reasoning.

\subsection{A Guiding Reasoning Principle}
\label{sec:guiding-reasoning-principle}

The previous subsection seems to
point us in an important direction:
our reasoning should somehow be based on the existence of a ``most precise'' type argument
used when inheriting from the scrutinee's class.
When we create a new instance of a class,
we need to pick specific type arguments for every type parameter of every inherited class.
But the type arguments we later ascribe to this instance may become less precise due to variance,
and it is crucial to take that into account.

\label{sec:sub-evidence-example}
Let us inspect another example.
We define a @SUB<S,@ @T>@ class that is
contravariant in @S@ and covariant in @T@ and
that works as runtime {\it evidence} that @S@ is a subtype of @T@ \cite{yallop2019}:
\begin{lstlisting}[language=kotlinish]
  sealed class SUB<in S, out T>()
  class Refl<U>() : SUB<U, U>()
\end{lstlisting}

\noindent
To illustrate its use, we define a function to @convert@ between two seemingly unrelated types:

\begin{lstlisting}[language=kotlinish]
  fun <T>convert(t: T, ev: SUB<T, Int>): Int =
    when (ev) { is Refl<*> -> t }
\end{lstlisting}

\noindent
When we check that @ev@ is an instance the @Refl@ class,
we discover is that there was some type @U@
which was used as a type argument to @Refl@ when creating @ev@.
This type is both a supertype of @T@ and a subtype of @Int@,
so it can only exist if @T@ $<:$ @Int@
(otherwise the bounds of @U@ would be inconsistent),
which is what we need to type check @convert@.

Let us inspect one final example:

\begin{lstlisting}[language=kotlinish]
  fun <T>convert2(t: T, ev: SUB<Expr<T>, Expr<Int>>): Int =
    when (ev) { is Refl<*> -> t }
\end{lstlisting}

\noindent
By analogous reasoning, %
we discover that @Expr<T>@ $<:$ @Expr<Int>@.
However, by itself this is not enough to type check @convert2@.
Since @Expr@ is covariant,
we need to {\it infer}
that @T@ must necessarily be a subtype of @Int@
for this relationship to hold,
which exemplifies that
subtyping reconstruction
needs to ``invert'' subtyping relationships,
inferring things about the arguments of related types.

Given a value @x@ of some covariant class type @Expr<T>@ (declared as @class@ @Expr<out@ @A>@),
let us denote by @x.A@ the precise argument type used to construct this instance,
regardless of any widening that may have happened afterwards.
So we must have that @x.A@ $<:$ @T@.
The fact that this notation looks like a reference to a \emph{member} @A@ in @x@ is not accidental:
our main idea is to represent these ``most precise types''
as \emph{type members}. Type members live inside class instances,
which thus in a sense behave like first-class modules.
This notation allows us to conveniently make explicit how a function like @convert@ above could type check,
by rephrasing it as follows:

\begin{lstlisting}[language=kotlinish,mathescape=true]
  fun <T>convert(t: T, ev: SUB<T, Int>): Int =
    when (ev) { is Refl<*> -> t @: ev.U@ }
\end{lstlisting}

\noindent
In the above, we added a type annotation to @t@ which shows that it has type @ev.U@,
since @T@ $<:$ @ev.U@ by contravariance,
which allows the term to type check since @ev.U@ $<:$ @Int@ by covariance.

Reasoning about most-precise type arguments in terms of type members
allows us to clarify and resolve subtleties and paradoxes related to pattern matching and inheritance.
Going back to the @Derived2@ example from \autoref{sec:ty-ctor-variance},
we can now explain how Kotlin and Scala have two different notion of inheritance:
when seeing an inheritance clause like `@:@ @Base<Int>@', Kotlin implicitly assumes @this.T@ @=@ @Int@,
whereas Scala
only assumes @this.T@ $<:$ @Int@ because @T@ is covariant in @Base@.\footnote{
  In fact, in Scala makes an exception for final and {\tt case} classes,
  which use \emph{invariant} inheritance.
  This is because {\tt case} classes are usually intended for pattern matching,
  so it is better to recover more information from them in such use cases.
}
In fact, in Scala, it is even possible to inherit (indirectly) from both @Base<Int>@ and @Base<String>@,
which will result in @this.T@ $<:$ @Int@ @&@ @String@ (where @&@ denotes an \emph{intersection type}).
Both choices are sound and have pros and cons.
They also lead to different
subtyping reconstruction capabilities.

\subsection{Type Parameters as Type Members, Classes as Runtime Tags}

As we shall see in the next section,
we go even further and completely do away with type parameters.
Indeed, it turns out that type parameters themselves,
along with type constructor variance, can be encoded using type members,
intersections, and structural types, which are all part of our DOT-based calculus.
Doing so allows us to work on a unified representation of all these concepts,
greatly simplifying the specification and soundness proof of our
subtyping reconstruction principles.

Similarly, classes themselves are not a ``core'' concept of the various DOT calculi,
and are normally encoded using DOT's expressive recursive object type system.
Indeed, it is possible to encode nominal classes by using bounded type members
and by defining libraries as abstract modules that hide their implementations.
This encoding of the \emph{static typing} aspect of classes is not new \cite{martres2022}.
What is new in this paper is %
that we associate runtime ``tags'' to such encoded classes,
which can be matched against through a @case@ construct to recover some static type information
about the corresponding instance,
as we shall explained in detail in the next section.

Many of the ideas presented in this paper are %
already well-known.
In fact, they %
are essentially how
the problem of type checking pattern matching
was reasoned about %
while implementing
the Scala 3 compiler,
after a long history of unsound and limited GADT support in Scala 2.\footnote{
  The Scala versions of all examples we have shown so far type check correctly in Scala 3.
}
But this paper is the first to rigorously formalize the system and derive (mechanized) proofs of its soundness.

\subsection{Real-World Justification}

Our experience with Scala is that GADT-style reasoning is pervasive in code bases that make advanced
use of the type system to enforce compile-time guarantee about their programs.
Classes like @Expr@ naturally occur in %
programs that manipulate typed abstract syntax trees,
for instance in database libraries, which need to manipulate query representations.
Indeed, one of our motivating examples has been the design of a query compiler
called dbStage\footnote{Early prototypes of dbStage can be found at \url{https://github.com/epfldata/dbstage}.}
using the Squid type-safe metaprogramming framework \cite{parreauxThesis2020}
that relies very heavily on GADT-style reasoning.
As another example,
\citet{eisenberg2020} recently described Stitch,
an interpreter and type checker similarly making heavy use of GADTs.
Being able to discover type relationships by pattern matching also allows design patterns that were not possible or convenient before,
for example using subtyping evidence \cite{yallop2019} as described in \autoref{sec:sub-evidence-example};
the Scala standard library itself makes heavy use of similar subtyping evidence types.\cprotect\footnote{%
  For example,
  the @unzip@ method on @List<T>@
  requires an implicit parameter evidence that @T@ $<:$ @(A,@ @B)@
  for some @A@, @B@.
}

\section{Informal Introduction to \calculus}
\label{sec:presentation-calculus}
In this section, we informally present \calculus{} and discuss how it allows pattern matching
to uncover hidden types and subtyping relationships.

\lstset{
  language=scalaish,
  mathescape=true,
}

\calculus{} belongs to the DOT family of systems, which were originally intended as formal foundations for Scala.
The base DOT calculus was designed to explain surface type features of Scala using as few core features as possible \cite{amin2016}.
Importantly for our purposes, there are no classes nor type parameters in DOT,
as both can be encoded through type members \cite{amin2016,rapoportPathDOTFormalizing2019,martres2022}.
This incidentally demonstrates that a language does not need to feature nominality
in order to support subtyping reconstruction.

We now gradually introduce the concepts of \calculus{} by using Scala as the surface syntax.
Our first goal is to encode @IntLit@ from \autoref{fig:motiv-example}.
In \calculus, objects are created via object literals, which are typed with structural types,
as in
@val@ @e@ @=@ @new@ @{@ @val@ @x@ @=@ @0@ @}@ @:@ @{@ @val@ @x:@ @Int@ @}@.

Types are compared based on their \emph{structure},
so to encode @IntLit@, we need to first recover nominality.
To do this, we rely on the fact that
in \calculus, as well as in Scala, an object instance can contain associated types --- called \emph{type members}. For instance:
\begin{lstlisting}
type Animal = { type Food; def eat(food: this.Food): Unit }
def feed(ani: Animal, food: ani.Food) = ani.eat(food)
\end{lstlisting}
In the above example, @Food@ is a type member of the object type @Animal@.
Given a value of type @Animal@, for instance @ani@, we can reference the type member that lives in that instance using a \emph{path-dependent type}, in this case @ani.Food@.
This is analogous to defining @Animal@ as a generic class and adding a type parameter to @feed@.
The distinction between the above type and one like @Animal<Food>@ is that a type member is existential by default.
For instance, we can refer to a @List<Animal>@, a list in which every animal may have its own distinct type of food.

Like type parameters, type members can have (upper and lower) bounds, and in particular they can also be equal to some type (which simply means that their lower and upper bounds are the same).
Much like instantiating a generic class %
requires specific type arguments for all type parameters of the class (they are often inferred,
but are there nonetheless),
when we create an object in \calculus{}, all of its type members must be %
defined to be equal to some existing types.
This idea allows \calculus{} to capture our intuition from before about the ``most precise'' type arguments
of object instances.
To illustrate,
we can create %
a value of type @Animal@ as follows:

\begin{lstlisting}
val goat: Animal = new { type Food = Any; def eat(food: Any) = () }
\end{lstlisting}

\noindent
Given the above @val@, @goat.Food@ refers to the precise food type that this specific animal eats,
even though it is an abstract type and we know nothing about it due to the @Animal@ type ascription.

Using bounded type members %
and structural types, we can recover nominal classes.
To encode @IntLit@, we create a ``package'' object @g@ with a type member @IntLit@ that defines the members of the class.
We also define a method @newIntLit@ which allows us to create instances of said type member
(i.e., a class constructor).
We annotate the package object
with a type that only leaves an upper bound on @g.IntLit@.
Taking everything together, our @g@ looks as follows:
\begin{lstlisting}
val g: {
  type IntLit <: { val value: Int }
  def newIntLit(i: Int): IntLit
} = new {
  type IntLit = { val value: Int }
  def newIntLit(i: Int): IntLit = new { val value = i }
}
\end{lstlisting}
Outside of @g@, we only have an upper bound on @g.IntLit@.
As a result, with the above definition the only way to create an object of type @g.IntLit@ is by calling @g.newIntLit@.
In particular, a value that has the same \emph{shape} as @IntLit@ but does not statically have @g.IntLit@ as part of its type
will \emph{not} be usable as a @g.IntLit@.
In this way, by hiding the exact type @g.IntLit@ stands for, we have encoded nominality.

We now take a separate look at encoding generic classes based on the example of @Expr<A>@.
The class itself is encoded as a type member: @type@ @Expr@ $<:$ @{@ @type@ @A@ @}@;
notice that its type parameter turns into a type member.
An applied type like @Expr<Int>@ is translated to an {\it intersection type}
@Expr@ @&@ @{@ @type@ @A@ @=@ @Int@ @}@.
This translation is variance-sensitive;
if
@Expr@ were covariant, its type argument would be translated using an \emph{upper bound} instead, as
@Expr@ @&@ @{@ @type@ @A@ $<:$ @Int@~@}@;
likewise, contravariant type arguments correspond to \emph{lower bounds} in applied types.

Let us now come back to our package object @g@,
which we extend with @Expr<A>@,
with which @IntExpr@ is now related through an intersection type:
\begin{lstlisting}
val g: {
  type Expr <: { type A }
  type IntLit <: Expr & { type A = Int; val value: Int }
  def newIntLit(i: Int): IntLit
} = new {
  type Expr = { type A }
  type IntLit = Expr & { type A = Int; val value: Int }
  def newIntLit(i: Int): IntLit = new { type A = Int; val value = i }
}
\end{lstlisting}

\noindent
Again, @g@'s annotation leaves only upper bounds for @g.Expr@ and @g.IntLit@, to encode nominality.
In \calculus, structural types with multiple members are represented
via single-member structural type intersections.
For instance,
@{@ @type@ @A@ @=@ @Int;@ @val@ @value:@ @Int@ @}@
is represented as
@{@ @type@ @A@ @=@ @Int@ @}@ @&@ @{@ @val@ @value:@ @Int@ @}@.
We use the former as a \emph{shorthand} for the latter.

Function type parameters are also encoded using type members.
A generic function in \calculus{} takes an additional argument %
with one type member per original type parameter.
To illustrate, we can now take a look at an encoded version of @eval@ from \autoref{fig:motiv-example}:
\begin{lstlisting}
def eval(tp: { type T }, e: g.Expr & { type A = tp.T }): tp.T =
  e match { case e1: g.IntLit => e1.value }
\end{lstlisting}
\noindent
As expected, this example is well-typed in \calculus; let us inspect it in detail.
The type of @e1@ is @e.type@ @&@ @g.IntLit@, with @e.type@ being the {\it singleton type} of @e@.
Singleton types are, conceptually,
only inhabited by a single object instance (here @e@).
Subtyping interacts with intersection types as expected: we can type @e1@ both as @e.type@ and as @g.IntLit@.
Since we have @e1@ @:@ @e.type@, we know that @e1@ is the same object as @e@, i.e. it is an alias.
However, because of the other part of the intersection type, the type of @e1@ is more precise than
@g.Expr@ @&@ @{@ @type@ @A@ @=@ @tp.T@ @}@, the type of @e@.
In this sense, @e1@ makes the type of @e@ more precise,
which is what we need to type the example.

Since we have both @e1@ @:@ @g.IntLit@ and @e1@ @:@ @e.type@,
we respectively know that @Int@ $<:$ @e1.A@ and @e1.A@ $<:$ @tp.T@.
We conclude that @Int@ @<:@ @tp.T@ by transitivity of subtyping.
So we can return @e1.value@ (whose type is @Int@) from @eval@ (whose result type is @tp.T@).

One problem with @eval@ so far is that the @match@ expression cannot succeed at runtime,
since our encoding of @IntLit@ so far lacks a \emph{runtime tag}.
Implementations of OO languages typically tag every object value with its \emph{runtime class},
which enables dynamic type checks and casts.
To account for this, we actually encode each class @C@ as
\emph{both} a type member (as seen before)
\emph{and} a tag value:
\begin{lstlisting}
  type C <: { val tag: g.C_tag.type; $\ldots$ }; val C_tag = g.freshTag()
\end{lstlisting}
\noindent
where each call to @g.freshTag()@ creates a unique tag value.\footnote{
  This concept of unique tag is akin to TypeScript's {\tt unique} {\tt symbol} feature
  (\url{https://www.typescriptlang.org/docs/handbook/release-notes/typescript-2-7.html}),
  although the latter only works with static paths,
  while we allow path-dependent tags.}
The representation of tag values at runtime does not matter
as long as we can compare them. E.g., we could use new empty objects,
compared by object identity.
Accordingly, to match on an object we now inspect its tag:
\begin{lstlisting}
  x match { case e1: g.C => $\ldots$ }
  // is encoded as:
  x.tag match { case g.C_tag => $\ldots$ }
\end{lstlisting}
\noindent
By construction of the encoding, we know that only objects created through the constructor of @g.C@
will be associated with @g.C_tag@ at runtime,
which lets us type @x@ in the branch body as @x.type@ @&@ @g.C@.
(Note that this reasoning is not supported by Scala itself.)
Finally, to encode class hierarchies more than one level deep, we would need to add multiple tags to each object.

This concludes our informal explanation of \calculus's main ideas.
As we shall see, in the actual formal calculus we strive for simplicity:
objects may only have \emph{one} tag
and \emph{all} type members are associated with a runtime tag,
even those not meant to represent classes.
Extending the system so that objects can have multiple tags is straightforward and so is left out of \calculus.
Associating all type members with a tag is not a problem because
when translating a program to \calculus{}
all tags used in objects correspond to actual classes,
so non-class type member tags could be erased.

\section{Formal Presentation of \calculus{}}
\label{sec:formalization}

\FloatBarrier

\begin{wide-rules}\noindent
	{\footnotesize\begin{multicols}{2}\noindent
			\begin{flalign}
			x,\,y,\,z            \tag*{\textbf{Variable}}\\
			a,\,b,\,c            \tag*{\textbf{Term member}}\\
			A,\,B,\,C            \tag*{\textbf{Type member}}\\
			p,\,q,\,r\coloneqq\ &           \tag*{\textbf{Path}}\\
			&x                                      \tag*{variable}\\
			&p.a                              \tag*{field selection}\\
			t,\,u\coloneqq\ &                     \tag*{\textbf{Term}}\\
			&\assignTrm                             \tag*{stable term}\\
			&p\,q                             \tag*{application}\\
			&\tLet x t u                           \tag*{let binding}\\
      &\new{\tCase p {q.A} y {t_1} {t_2}}          \tag*{case matching}\\
      v\coloneqq\ &                           \tag*{\textbf{Value}}\\
      &\new{\tNew p A x T d}                            \tag*{tagged object}\\
      &\tLambda x T t                         \tag*{lambda}\\[1em]
			\assignTrm \coloneqq\  &                  \tag*{\textbf{Stable Term}}\\
      &p \;|\; v                                 \tag*{path or value}\\
			d\coloneqq\ &                            \tag*{\textbf{Definition}}\\
			&\set{a=\assignTrm}              \tag*{field definition}\\
			&\set{A=T}                              \tag*{type definition}\\
			&\tAnd d d                           \tag*{aggregate definition}\\
			S,\,T,\,U,\,V\coloneqq &                     \tag*{\textbf{Type}}\\
			&\top                                   \tag*{top type}\\
			&\bot                                   \tag*{bottom type}\\
			&\tFldDec a T                           \tag*{field declaration}\\
			&\tTypeDec A S T                        \tag*{type declaration}\\
			&\tAnd S T                              \tag*{intersection}\\
			&\tRec x T                              \tag*{recursive type}\\
			&\tForall x S T                         \tag*{dependent function}\\
			&p.A                               \tag*{type projection}\\
			&\single p                        \tag*{singleton type}
			\end{flalign}
	\end{multicols}}
	\caption{Abstract syntax of \calculus{}}
	\label{fig:syntax}
\end{wide-rules}

We now present the \calculus calculus.
\calculus extends pDOT \cite{rapoportPathDOTFormalizing2019}, which itself is a generalization of DOT \cite{amin2016} that allows arbitrary path lengths in path-dependent types:
both $\single{x}$ and $\single{x.y}$ are permitted. Importantly for us, it also formalizes singleton types.

All \calculus{} terms are in a variant of A-normal form, or ANF \cite{sabry1993}.
For instance, to apply one term to another, we must first bind the expressions to variables:
$$ \nLet x = t \nIn \nLet y = s \nIn x\,y  $$

\noindent
This does not lead to any loss of expressivity, since a simple syntactic translation can change regular terms to ANF.
Essentially, ANF gives us a name for every value, which is important when we allow function results to depend on their arguments with path-dependent types.
Other DOT systems did not require ANF \cite{rompf2016}, but at the cost of having two rules for typing application: one for value arguments and another for variable arguments.

Since function results can depend on their arguments, in \calculus{} they are typed with {\it path-dependent function types}.
The identity function is typed as follows:
$$ \vdash \lambda(x : \top)\,x : \forall(x : \top)\,\single{x} $$

\noindent
If the result of a function type $\forall(x:S)\,T$ does not depend on its argument (i.e. if we have $x \notin \mathrm{fv}\, T$),
we spell it as $S \rightarrow T$ as a shorthand.

So far, we have been simplifying one aspect of the syntax.
Whereas one would typically use variables to reference let-bound values, in \calculus{} we can use {\it paths}.
A path $x.a_1.\cdots.a_n$ is a chain of field selections $.a_i$ starting from a variable $x$.
Like in pDOT, a \calculus{} path also represents the \emph{identity} of an object, an important notion for a path-dependent type system.
In other DOT systems, objects were bound in a store and referenced by simple variables,
whereas
in pDOT and \calculus{} objects may also be nested inside other objects,
justifying the use of paths for object identity.
The typing judgment $\typDft p {\single{q}}$ means that $p$ and $q$ have the same identity.
Like in pDOT, a path is never directly substituted for a value in \calculus{}, since doing so would strip an object of its identity.

In \calculus{}, objects are \emph{tagged} in order to allow pattern matching. An object literal in \calculus{} looks as follows, where $\nu$ is a binder that stands for the usual @new@ in Scala:
$$
{\tNew p A x T d}
$$
Let us inspect it from the right.
$d$ is the body of the object;
it can contain fields and type members,
whose names are case-sensitive: we use $A,B,C$ for type members and $a,b,c$ for term members.
Object definitions in \calculus{} can be circular: their fields can freely reference other fields, regardless of the definition order.
In order to ensure this does not lead to initialization problems, all fields must be initialized with \emph{stable} terms $s$ which are either values or paths.
Methods are represented with fields bound to lambda abstractions.
The ``tag'' $p.A$ is the part of the object value that enables pattern matching. %
In \calculus{}, every object is tagged with a \emph{type member} --- recall that type members are how we formally represent classes.
The \emph{self reference} $x$ (usually known as @this@) is explicitly bound in the syntax and annotated with a type $T$,
which can be used to specify the type that this object instance should have.

As an illustration, the body of the @newIntLit@ function from before would look as follows:
$$
\lambda(i : \Int)\,
  \nu(x : \set{ A = \Int } \wedge \set{ \ident {value} : \Int }) [g.\ident{IntLit}]
  \set{ A = \Int } \wedge \set{ \ident{value} = i }
$$

Finally, the case form $\ttCase p {q.A} y {t_1} {t_2}$ allows pattern-matching on arbitrary values.
There are two branches in this form: if $p$ resolves to an object whose tag conforms to $q.A$, we bind $y$ to $p$ and enter $t_1$, otherwise we enter $t_2$.

We will now inspect the typing and subtyping rules of \calculus.
Most of the rules are the same as in pDOT; we have highlighted all the changes in grey.

\subsection{Typing}

\begin{wide-rules}\footnotesize
\begin{flalign*}
  \G&\Coloneqq \varnothing\ |\ \G,\,x\colon T
                       \tag*{Type environment}
\end{flalign*}

\setlength{\fboxsep}{0.5em}
\textbf{Term typing \quad \fbox{$\typDft t T$}}
\setlength{\fboxsep}{1pt}
\begin{multicols}{3}

\infrule[Var]
  {\G(x)=T}
  {\typDft x T}

\infrule[All-I]
{\typ {\extendG x T} t U
	\andalso
	x\notin\fv T}
{\typDft{\tLambda x T t}{\tForall x T U}}

\infrule[All-E]
{\typDft {p} {\tForall z S T}
	\andalso
	\typDft {q} S}
{\typDft {{p\, q}} {\tSubst z {q} T}}
  
\infrule[\{\}-I]
{\dtyp x {\extendG x T} d T\\
\new{\typ {\extendG x T} {x} {q.A}}}
{\typDft {\tNew q A x T d} {\tRec x T}}

\infrule[Fld-E]
{\typDft {p} {\tFldDec a T}}
{\typDft {p.a} T}
  
\infrule[Fld-I]
{\typDft {p.a} T}
{\typDft {p} {\tFldDec a T}}

\infrule[Let]
{\typDft t T
	\\
	\typ {\extendG x T} u U
	\andalso
	x\notin\fv U}
{\typDft {\ttLet x t u} U}

\newruletrue
\infrule[Case]
{\typeable {p} \andalso \typeable q \\
 \typ {\extendG y {\tAnd{\single{p}}{q.A}}} {t_1} U \\
 \typDft {t_2} U \andalso
 y\notin \fv U}
{\typDft {\ttCase p {q.A} y {t_1} {t_2}} U}
\newrulefalse

\infrule[Sngl-Trans]
	{\typDft p {\single q} \andalso \typDft q T}
	{\typDft p T}

\newruletrue
\infrule[Sngl-Self]
	{\typeable p}
	{\typDft p {\single p}}
\newrulefalse

\infrule[Sngl-E]
  {\typDft p {\single q} \andalso
   \typeable {q.a}}
  {\typDft {p.a} {\single {q.a}}}

\infrule[Rec-I]
  {\typDft {p} {\tSubst x {p} T}}
  {\typDft {p} {\tRec x T}}

\infrule[Rec-E]
  {\typDft {p} {\tRec x T}}
  {\typDft {p} {\tSubst x {p} T}}

\infrule[\&-I]
  {\typDft {p} T
    \andalso
    \typDft {p} U}
  {\typDft {p} {\tAnd T U}}

\infrule[Sub]
  {\typDft t T
    \andalso
    \subDft T U}
  {\typDft t U}

\end{multicols}

\setlength{\fboxsep}{0.5em}
\textbf{Definition typing \quad \fbox{$\dtypDft d T$}}
\setlength{\fboxsep}{1pt}
\begin{multicols}{2}

 \infax[Def-Typ]
 {\dtypDft {\set{A=T}} {\tTypeDec A T T}}
 
\infrule[Def-All]
  {\typDft {\tLambda x T t} {\tForall x U V}}
  {\dtypDft {\set{a=\tLambda x T t}} {\tFldDec a {\tForall x U V}}}

\infrule[Def-New]
{\dtyp {p.a} \G {d\theta} {T\theta} \andalso \new{\typ \G {p.a} {q.A\theta}} \\
  \theta = \!{\tSubst y {p.a} {}} \andalso \tightbounds T
}{
  \dtypDft {\set{a = \nu(y\colon T)\new{[q.A]}d}} {\tFldDec a {\tRec y T}}}

\infrule[Def-Path]
  {\typeable q}
  {\dtypDft {\set{a = q}} {\tFldDec a {\single q}}}
  
\infrule[AndDef-I]
  {\dtypDft {d_1} {T_1}
    \andalso
    \dtypDft {d_1} {T_2}
    \\
    \dom{d_1} \cap \dom{d_2} = \varnothing
   }{\dtypDft {\tAnd {d_1} {d_2}} {\tAnd {T_1} {T_2}}}
\end{multicols}
\begin{minipage}{0.35\linewidth}
\noindent
\setlength{\fboxsep}{0.5em}
\textbf{Typeable paths \quad \fbox{$\typeable p$}}
\setlength{\fboxsep}{1pt}
\vspace{1.2em}
\infrule[Wf]
    {\typDft p T}
    {\typeable p}
\vspace{1.6em}
\noindent
\setlength{\fboxsep}{0.5em}
\end{minipage}%
\begin{minipage}{0.1\linewidth}
\
\end{minipage}%
\begin{minipage}{0.55\linewidth}
\setlength{\fboxsep}{0.5em}
\textbf{Tight bounds \quad \fbox{$\tightboundseq T$}}
\setlength{\fboxsep}{1pt}
\[
{\tightboundseq T} =
\begin{cases}
U = V							&\text{if } T = \tTypeDec{A}{U}{V}\\
\tightboundseq U         			&\text{if } T = \tRec{x}{U} \text{ or } \tFldDec{a}{U}\\
\tightboundseq U \text{ and }
    \tightboundseq V           	&\text{if } T = U \wedge V\\
\text{true}           			&\text{otherwise}
\end{cases}
\]
\end{minipage}
\caption{\calculus{} typing rules}
\label{fig:typing-rules}
\end{wide-rules}

We first go through the typing and subtyping rules of \calculus{} one by one.
\autoref{fig:typing-rules} presents the typing rules of \calculus{}.
We start with term typing rules, followed by definition typing, and finally explain the subtyping rules.

\paragraph{Term typing rules}
The \rn{Var}, \rn{Let} and \rn{Sub} rules for typing variables, \code{let} forms and subsumption are as expected.
Term abstractions are typed with dependent function types by \rn{All-I}.
Rule \rn{All-E} assigns term applications the function result type, with the term argument replacing the function parameter.
For example, if $\typDft{p}{\tForall{x}{\top}{x.type}}$ and $q$ is a typeable path, then $p\, q$ is typed as $\single{q}$.
The \rn{\{\}-I} rule types the tagged objects.
The body is typed with definition typing, explained later.
The object type also must conform to its self-reference type and its tag.

Rules \rn{Sngl-Trans}, \rn{Sngl-Self} and \rn{Sngl-E} deal with singleton types.
If $p$ is typeable, it can be typed as $\single{p}$ with \rn{Sngl-Self}.
This rule not only lets \calculus{} model Scala's type system more closely than pDOT (which surprisingly lacks it),
but also turns out to seriously impact both soundness (\autoref{sec:metatheory}) and subtyping ({\autoref{sec:subtyping}}).
If $\typDft{p}{\single{q}}$, we say that $p$ \emph{aliases} $q$.
In such a case, we can always assign $p$ the same type as $q$ using \rn{Sngl-Trans}.
Aliasing is propagated through field selection with \rn{Sngl-E}:
$p.a$ aliases $q.a$ if $q$ has a field $a$ and $p$ aliases $q$.

Field selections are typed with \rn{Fld-E}.
The \rn{Fld-I} rule introduces field types when the field selection is typeable.
This rule increases the expressiveness of the type system~\cite{rapoportPathDOTFormalizing2019}.
Given $\typDft{p}{\tRec{x}{\tFldDec{a}{\single q}}}$ and $\typDft{q}{U}$, we cannot derive that $\subDft{\single{q}}{U}$.
However, with \rn{Fld-I} we still can type $p$ as $\tFldDec{a}{U}$.

The \rn{Rec-I} and \rn{Rec-E} rules respectively introduce and eliminate recursive types.
Recursive types do not participate in \calculus{} subtyping;
instead, we unwrap recursive types with \rn{Rec-E} and re-introduce them with \rn{Rec-I}.
For example, the following derivation tree shows how we can derive $\typDft{p}{\tRec{x}{\tFldDec{a}{T}}}$
when $\typDft{p}{\tRec{x}{\tFldDec{a}{S}}}$ and $\subDft{S}{T}$ (where $x$ is absent from $S$ and $T$):
\begin{mathpar}
	\inferrule*[right=\rn{Rec-I}]
  {
    \inferrule*[right=\rn{Sub}]{
      \inferrule*[right=\rn{Rec-E}]{
        \typDft{p}{\tRec{x}{\tFldDec{a}{S}}}
      }{\typDft{p}{\tFldDec{a}{S}}} \\
      \inferrule*[right=\rn{Fld-<:-Fld}]{
        \subDft{S}{T}
      }{
        \subDft{\tFldDec{a}{S}}{\tFldDec{a}{T}}
      }
    }{\typDft{p}{\tFldDec{a}{T}}}}
  {\typDft{p}{\tRec{x}{\tFldDec{a}{T}}}}
\end{mathpar}
Although leaving out subtyping relationships between recursive types results in a small loss of expressiveness,
it significantly simplifies the metatheory \cite{amin2016}.

The \rn{\&-I} rule is the introduction rule for intersection types.
One may think that this rule is redundant because we have the \rn{Sub} and \rn{<:-And} rules.
However, this rule does add expressiveness when interacting with recursive types \cite{amin2016}.

The \rn{Case} rule types case forms ${\ttCase p {q.A} y {t_1} {t_2}}$.
Typing the else branch $t_2$ is straightforward.
When typing $t_1$, we extend the environment with $y : {\tAnd{\single{p}}{q.A}}$,
which essentially introduces a witness that the path $p$ can be typed as $q.A$.
We see in \autoref{sec:subtyping} that together with \emph{subtyping inversion},
such a witness allows reconstruction of subtype relationships.

\paragraph{Definition typing}
The definition typing rules in \autoref{fig:typing-rules} are used to type object bodies.
In \pdot{} and \calculus{}, only function values, object values and paths can be used to initialize a field.
The \rn{Def-All}, \rn{Def-New} and \rn{Def-Path} rules type object fields as their precise type.
The $\tightbounds T$ condition in \rn{Def-New} ensures all type members in $T$ have equal bounds.
In addition to type checking the definition, the \rn{Def-New} rule also checks the conformance between the object type and the type tag.
The \rn{Def-Typ} rule types the type member definition.
The \rn{AndDef-I} rule type multiple member definitions as an intersection type.
The prefix $p$ in definition typing keeps track of the identity of the currently typing object,
since the definition typing rules can be used to type nested object definitions.

\paragraph{Tagged objects}
An object form $\tNew x A y T d$ must itself be typeable as its tag $x.A$;
this condition is checked by the \rn{\{\}-I} and \rn{Def-New} rules.
It may appear problematic to require \emph{all} object to refer to another object's type member in its tag:
after all we need a tag
to type check the very first object in a \calculus program,
but tags themselves come from objects!
This is not actually a problem because in \calculus objects can be typed recursively:
we can type check a self-tagged object $\kLet x = \tNew{x}{A}{x}{\set{A = \top}}{\set{A = \top}}$,
whose type member $x.A$ may later be used as a ``top tag'' that can tag any object.
In later examples we will sometimes skip the tag of an object;
we understand such objects to be tagged with this ``top tag''.
Notice that our tag system does not support inheritance:
in order to encode class hierarchies more than one level deep,
we would need to tag objects with \emph{several} class tags.
For simplicity, we leave this straightforward extension out of \calculus{}.

\subsection{Subtyping}\label{sec:subtyping}
\begin{wide-rules}
\setlength{\fboxsep}{0.5em}
\textbf{Subtyping \quad \fbox{$\subDft S T$}}
\setlength{\fboxsep}{1pt}
  \begin{multicols}{3}
\infax[Top]
  {\subDft T \top}

\infax[Bot]
  {\subDft \bot T}

\infax[Refl]
  {\subDft T T}

\infrule[Trans]
  {\subDft S T
    \andalso
    \subDft T U}
  {\subDft S U}

\infax[And$_1$-$<:$]
  {\subDft {\tAnd T U} T}

\infax[And$_2$-$<:$]
  {\subDft {\tAnd T U} U}

\infrule[$<:$-And]
  {\subDft S T
    \andalso
    \subDft S U}
  {\subDft S {\tAnd T U}}
  
\infrule[Fld-$<:$-Fld]
{\subDft T U}
{\subDft {\tFldDec a T} {\tFldDec a U}}

\newruletrue
\infrule[Fld-$<:$-Fld-Inv]
{\uniqueFlow{U}{\tFldDec a {T_1}} \\
  \subDft {U} {\tFldDec a {T_2}}}
{\subDft {T_1} {T_2}}
\newrulefalse

\infrule[Typ-$<:$-Typ]
{\subDft {S_2} {S_1}
	\andalso
	\subDft {T_1} {T_2}}
{\subDft {\tTypeDec A {S_1} {T_1}} {\tTypeDec A {S_2} {T_2}}}

\newruletrue
\infrule[Typ-$<:$-Typ-Inv$_1$]
{\uniqueFlow{U}{\tTypeDec A {S_1} {T_1}} \\
  \subDft {U} {\tTypeDec A {S_2} {T_2}}}
{\subDft {S_2} {S_1}}

\infrule[Typ-$<:$-Typ-Inv$_2$]
{\uniqueFlow{U}{\tTypeDec A {S_1} {T_1}} \\
  \subDft {U} {\tTypeDec A {S_2} {T_2}}}
{\subDft {T_1} {T_2}}
\newrulefalse

\infrule[$<:$-Sel]
{\typDft {p} {\tTypeDec A S T}}
{\subDft S {{p}.A}}

\infrule[Sngl$_{pq}$-$<:$]
{\typDft p {\single q} \andalso \typeable q}
{\subDft T {\repl p q T}}

\infrule[Sngl$_{qp}$-$<:$]
{\typDft p {\single q} \andalso \typeable q}
{\subDft T {\repl q p T}}

\infrule[Sel-$<:$]
  {\typDft {p} {\tTypeDec A S T}}
  {\subDft {p.A} T}

\infrule[All-$<:$-All]
{\subDft {S_2} {S_1}
	\\
	\sub {\extendG x {S_2}} {T_1} {T_2}}
{\subDft {\tForall x {S_1} {T_1}} {\tForall x {S_2} {T_2}}}
\end{multicols}

\newruletrue
\setlength{\fboxsep}{0.5em}
\textbf{Unique membership \quad \fbox{$\uniqueFlow T U$}}
\setlength{\fboxsep}{1pt}
\infrule[Unique-Label]
{\uniqueFlowLabel T U {\mathcal L}}
{\uniqueFlow T U}

\setlength{\fboxsep}{0.5em}
\textbf{Unique membership with label \quad \fbox{$\uniqueFlowLabel T U {\mathcal L}$}}
\setlength{\fboxsep}{1pt}

\begin{multicols}{2}
  \infax[One-Typ]
  {\uniqueFlowLabel {\tTypeDec A S T} {\tTypeDec A S T} {\set{A}}}

  \infax[One-Fld]
  {\uniqueFlowLabel {\tFldDec a T} {\tFldDec a T} {\set{a}}}

  \infax[One-Rec]
  {\uniqueFlowLabel {\tRec x T} {\tRec x T} {\emptysetAlt}}

  \infrule[And-Left]
  {\uniqueFlowLabel {U_1} {T_1} {\mathcal {L}_1} \andalso
   \uniqueFlowLabel {U_2} {T_2} {\mathcal {L}_2} \\
   \mathcal {L}_1 \cap \mathcal {L}_2 = \emptysetAlt
  }
  {\uniqueFlowLabel {\tAnd{U_1}{U_2}} {T_1} {\mathcal L_1 \cup \mathcal L_2}}

  \infrule[And-Right]
  {\uniqueFlowLabel {U_1} {T_1} {\mathcal {L}_1} \andalso
   \uniqueFlowLabel {U_2} {T_2} {\mathcal {L}_2} \\
   \mathcal {L}_1 \cap \mathcal {L}_2 = \emptysetAlt
  }
  {\uniqueFlowLabel {\tAnd{U_1}{U_2}} {T_2} {\mathcal L_1 \cup \mathcal L_2}}
\newrulefalse
\end{multicols}

\normalsize
\caption{\calculus{} subtyping rules}
\label{fig:subtyping-rules}
\end{wide-rules}
We now inspect the subtyping rules defined in \autoref{fig:subtyping-rules}.
Subtyping is made reflexive and transitive with \rn{Refl} and \rn{Trans}.
Rules for top, bottom and intersection types are typical (respectively \rn{Top}, \rn{Bot} and \rn{And$_1$-<:}, \rn{And$_2$-<:}, \rn{<:-And}).
Rules \rn{Fld-<:-Fld}, \rn{Typ-<:-Typ} and \rn{All-<:-All} derive the standard subtyping relationships between object fields, type members and functions.
Rules \rn{<:-Sel} and \rn{Sel-<:} relate a path-dependent type and its bounds.
Importantly, these rules also allow path-dependent types to introduce new subtyping relationships.
For example, if we have $\typDft{p}{\tTypeDec{A}{\Int}{q.T}}$,
then by \rn{Trans} and relating $q.T$ with its bounds we also can derive $\subDft{\Int}{q.T}$.

\paragraph{Aliasing paths}
Rules \rn{Sngl$_{pq}$-<:} and \rn{Sngl$_{qp}$-<:} establish equivalence of aliasing paths:
if we can derive $\Gamma |- p : \single{q}$, then with subtyping we can freely substitute $p$ for $q$ in types and vice versa.
The expressivity of these rules is augmented by \rn{Sngl-Self}.
Consider the following example:
\begin{equation*}
	\ttCase{p}{g.\ident{IntLit}}{y}{p}{\cdots}
\end{equation*}
In the matched branch we should be able to type $p$ as $g.\ident{LitInt}$.
However, this is not possible without the \rn{Sngl-Self} rule.
When typing the type case, the environment is extended with $y : \single p \wedge g.\ident{LitInt}$.
In \pdot{}, we can only type $\typDft y {\single p}$, but not the other direction due to the asymmetric nature of \pdot{} path aliasing:
the \rn{Sngl-Trans} rule only allows the propagation of typing from $p$ to $y$, but not the reverse.
This means that, although $y$ should witness that $p$ can be typed as $g.\ident{LitInt}$, there is no way to assign the witnessed type to $p$.
The \rn{Sngl-Self} rule mitigates the limitation by making path aliasing symmetric (i.e. we can derive $\typDft p {\single q}$ from $\typDft q {\single p}$).
The \rn{Sngl-Trans} can type $p$ as $g.\ident{LitInt}$.
The following derivation tree illustrates how \rn{Sngl-Self} enables us to type path aliasing symmetrically
Here $\G^\prime$ denotes $\extendG{y}{\single p \wedge g.\ident{IntLit}}$.
\begin{mathpar}
	\inferrule*[right=\rn{Sub}]{
    \inferrule*[right=\rn{Sngl-Self}]{
      \G^\prime \vdash p
    }{
      \typ {\G^\prime} p {\single p}
    } \\
    \inferrule*[right=\rn{Sngl$_{qp}$-<:}]{
      \typ {\G^\prime} {y} {\single p} \\
      \G^\prime \vdash p
    }{
      \sub {\G^\prime} {\single p} {\single y}
    }
  }{
    \typ {\G^\prime} p {\single y}
  }
\end{mathpar}

\subsubsection*{Subtyping inversion}
\label{sec:inversion-subtyping}
Finally, we can reconstruct premises of other subtyping rules
with \rn{Fld-$<:$-Fld-Inv}, \rn{Typ-$<:$-Typ-Inv$_1$} and \rn{Typ-$<:$-Typ-Inv$_2$}, the \emph{inversion rules}.
For instance, using \rn{Fld-$<:$-Fld-Inv}, we can derive $\Gamma \vdash \Int <: x.T$ from $\Gamma \vdash \set{ a: \Int } <: \set{ a: x.T }$,
respectively the premise and conclusion of \rn{Fld-$<:$-Fld}.
The inversion rules are only useful if the bindings in the context introduce some subtyping relationships as assumptions.
Consider the following example:
\begin{align*}
  &\tLambda x {\tTypeDec T \bot \top }\,
    \tLambda f {x.T \rightarrow \Int}\,
    \tLambda y {\set{ A : {\tFldDec a {\Int}}..{\tFldDec a {x.T}} }} \;\; f\,0
\end{align*}
\noindent
To type it, we need to derive that $\Int <: x.T$.
Based on the bounds of $y.A$, we can derive that $\tFldDec{a}{\Int} <: \tFldDec{a}{x.T}$; in other words, this relationship is an assumption introduced by $y$.
Intuitively, in any context in which we can construct a value for $y$, we can also derive that $\Int <: x.T$;
using bounds of $y$ and subtyping inversion, we can also derive that if we only have $y$:
\begin{prooftree}
  \eqfontsize
  \AxiomC{\ldots}
  \RightLabel{\rn{<:-Sel}}
  \UnaryInfC{$\subDft {\tFldDec a {\Int}} {y.A}$}
  \AxiomC{\ldots}
  \RightLabel{\rn{Sel-<:}}
  \UnaryInfC{$\subDft {y.A} {\tFldDec a {x.T}}$}
  \RightLabel{\rn{Trans}}
  \BinaryInfC{$\subDft {\tFldDec a {\Int}} {\tFldDec a {x.T}}$}
  \AxiomC{$\uniqueFlow {\tFldDec a {\Int}} {\tFldDec a {\Int}}$}
  \RightLabel{\rn{Fld-<:<-Fld-Inv}}
  \BinaryInfC{$\subDft {\Int} {x.T}$}
\end{prooftree}
\noindent
The combination of bindings that introduce subtyping assumptions and inversion rules that allow deriving premises of assumptions is \emph{subtyping reconstruction}:
the ability to use a binding to derive subtyping relationships necessary to construct a value which could stand for the binding.

The inversion rules use the unique membership relation $\uniqueFlow T U$, defined in \autoref{fig:subtyping-rules}.
It allows extracting a component $U$ out of an intersection type $T$,
which allows inverting relationships involving intersection types, for instance:
\begin{prooftree}
  \eqfontsize
  \AxiomC{$\subDft {\set{a : S} \wedge \set{b : U_2}} {\set{a : T}}$}
  \AxiomC{$\uniqueFlow {\set{a : S} \wedge \set{b : U}} {\set{a : S}}$}
  \RightLabel{\rn{Fld-<:-Fld-Inv}}
  \BinaryInfC{$\subDft S T$}
\end{prooftree}
\noindent
In order to invert relationships that involve an intersection type on the RHS, we can use \rn{Trans}.
For instance, in the above derivation, the inverted subtyping relation could be derived like this:
\begin{prooftree}
  \AxiomC{$\subDft {\set{a : S} \wedge \set{b : U_1}} {\set{a : T} \wedge \set{c : U_2}}$}
  \AxiomC{$\subDft {\set{a : T} \wedge \set{c : U_2}} {\set{a : T}}$}
  \RightLabel{\rn{Trans}}
  \BinaryInfC{$\subDft {\set{a : S} \wedge \set{b : U_2}} {\set{a : T}}$}
\end{prooftree}

In order for subtyping inversion to be sound, all the labels of $T$ in $\uniqueFlow T U$ must be unique.
To see why, consider this relationship:
$$
\subDft {\tFldDec{a}{T_1} \wedge \tFldDec{a}{T_2}} {\tFldDec{a}{U}},
$$
\noindent
The original premise used to derive this relationship may have been either $\subDft{T_1}{U}$ or $\subDft{T_2}{U}$.
Since we cannot know which one it was, we can derive neither with the inversion rules.

\subsection{Evaluation}

\setlength{\multicolsep}{9pt}
\newcommand\hole{[\ ]}
\begin{wide-rules}

\begin{minipage}{0.5\textwidth}
  \begin{align*}
    \sta\Coloneqq& \ \varnothing\ |\ \sta\co x\mapsto v
                   \tag*{\textbf{Store}}\\[10pt]
    \resolvedDft\Coloneqq& \ p \quad \text{(if $\exists v.\, \lookupStepDft p v$)}
                           \tag*{\textbf{Resolved Path}}\\[10pt]
  \end{align*}
\end{minipage}%
\begin{minipage}{0.5\textwidth}
  \begin{align*}
    e^\gamma\Coloneqq& \hole \ | \ \hole\, q \ |\ \resolvedDft\, \hole
                       \tag*{\textbf{Evaluation Context}} \\
    |& \ \ttLet x {e^\gamma} u \\
    |& \ \ttCase {\hole} {q.A} y {t_1} {t_2} \\
    |& \ \ttCase {\resolvedDft} {\hole.A} y {t_1} {t_2} \\
  \end{align*}
\end{minipage}

\begin{minipage}{0.45\textwidth}

\infrule[\redCtx]
{\reduction {\sta} t {\sta^\prime} {t^\prime}}
{\reduction 
    {\sta}
    {e^\gamma[t]}
    {\sta^\prime}
    {e^\gamma[t^\prime]}}
\vspace{0.5em}
\infrule[\redResolve]
{{\lookupStepDft p q}}
{\reductionDft
    {p}
    {q}}
\vspace{0.5em}
\infrule[\redApply]
{{\lookupStepDft q {\tLambda z T t}}}
{\reductionDft
    {{q\,\resolvedDft}}
    {\tSubst z {\resolvedDft} t}}
\vspace{0.5em}
\infax[\redLetpath]
{\reductionDft
    {\ttLet x {\resolvedDft} t}
    {\tSubst x {\resolvedDft} t}}
\vspace{0.5em}
\infrule[\redLetvalue]
{x \notin \dom\sta}
{\reduction
    \sta
    {\ttLet x v t}
    {\extendSta x v} 
    t}
\end{minipage}%
\begin{minipage}{0.55\textwidth}
\vspace{-0.95em}
\infrule[\redCaseThen]
{\lookupStepDft q {\tNew r A x T d} \andalso
 \lookupDft r {\resolvedDft}}
{\reductionDft {\ttCase q {\resolvedDft\!.\,A} y {t_1} {t_2}} {\tSubst y p {t_1}}}
\vspace{1.8em}
\infrule[\redCaseElse]
{\lookupStepDft q {\tNew r A x T d} \quad
 \lookupDft r {\resolvedDft_1} \quad
 {\resolvedDft_1}\!.\,A \neq {\resolvedDft_2}.B
 }
{\reductionDft {\ttCase q {{\resolvedDft_2}.B} y {t_1} {t_2}} {t_2}}
\vspace{1.2em}
\infrule[\redCaseLambda]
{\lookupStepDft p {\tLambda x T U}}
{\reductionDft {\ttCase p {\resolvedDft\!.\,A} y {t_1} {t_2}} {t_2}}
\end{minipage}
\vspace{0.5em}
\begin{multicols}{3}
	\infrule[Lookup-Step-Var]
		{\sta(x) = v}
		{\lookupStepDft x v}

	\infrule[Lookup-Step-Val]
		{\lookupStepDft p {\tNew r A x T {\dots \set{a=\assignTrm} \dots}}}
		{\lookupStepDft {p.a} {\tSubst x p \assignTrm}}		
		
	\infrule[Lookup-Step-Path]
		{\lookupStepDft p q}
		{\lookupStepDft{p.a}{q.a}}
\end{multicols}

\caption{Operational semantics of \calculus{}}
\label{fig:reduction-rules}

\end{wide-rules}

\setlength{\multicolsep}{2pt}

The operational semantics of \calculus{} is presented in \autoref{fig:reduction-rules}.
We define $\lookupSymbol^{*}$ as the transitive, reflexive closure of $\lookupSymbol$.
Since our reduction rules are significantly different from \pdot, we do not highlight the differences.

\subsubsection*{Store.}
Similarly to pDOT, \calculus reduction is store-based.
During evaluation, we still need to keep track of object identities; to do so, the store $\sta$ binds \emph{variables} to values.
Note that objects still contains nested objects during evaluation: in general, the identity of an object may be a path.

\subsubsection*{Path resolution.}
Compared to pDOT, the most significant change in \calculus is that the \rn{Resolve} rule allows reducing aliasing paths:
given $\lookupStepDft p q$, we can reduce $\reductionDft p q$.
We call a path $p$ that directly resolves to a value in $\gamma$ a \emph{resolved path} $\rho^{\gamma}$ (see \autoref{fig:reduction-rules}).
Resolved paths cannot be reduced any further;
much like singleton types capture object identities on the type level, on the term level a resolved path directly corresponds to the runtime identity of an object.
Importantly, as a consequence of modelling recursive objects and modules, it is possible to create objects with circular fields.
References to such fields results in infinite loops.
A circular object may be defined like this:
\[\nu(x : \set{a:\single{x.b}} \wedge \set{b:\single{x.a}}) \set{a=x.b} \wedge \set{b=x.a}\]

In \pdot, it is impossible to reduce paths, and even resolving path aliases breaks preservation.
For example, if $p$ can be typed as $q.type$ in $\G$ and it looks up to $q$ in $\sta$,
there is still no guarantee that $q$ can be typed as $q.type$ in $\G$ and so allowing $\reductionDft p q$ would violate type preservation.
In \calculus, we can always type $q$ as $q.type$ thanks to the \rn{Sngl-Self} rule and so reducing paths preserves their types.

\begin{figure}[h]
\begin{multicols}{2}
\noindent
  \infax[]
  {\reductionDft
    {\ttLet x {p} t}
    {\tSubst x {p} t}}

  \infax[]
  {\reductionDft
    {\ttLet x {\new \resolvedDft} t}
    {\tSubst x {\resolvedDft} t}}
\end{multicols}
\vspace{0.5em}
\begin{multicols}{2}
  \noindent
  \infrule[]
  {{\lookupDft q {\tLambda z T t}}}
  {\reductionDft
    {{q\,p}}
    {\tSubst z {p} t}}

  \infrule[]
  {{\lookupStepDft q {\tLambda z T t}}}
  {\reductionDft
    {{q\,\new{\resolvedDft}}}
    {\tSubst z {\new \resolvedDft} t}}
\end{multicols}

  \caption{The \rn{Let-Path} (top) and \rn{Apply} (bottom) rules of \pdot{} (left) and \calculus{} (right).}
  \label{fig:comparing-let-path-reduction}
\end{figure}

The reduction of paths also brings significant changes to the reduction of let bindings and lambda applications,
as we can see on \autoref{fig:comparing-let-path-reduction}.
In \pdot, a let-bound path is reduced by substituting the path into the let body;
in \calculus{}, we first reduce the path until it is fully resolved.
\subsubsection*{Case term reduction.}
Both pDOT and \calculus{} ensure that paths typeable with function and object types always resolve to a value, which is necessary for soundness \cite{rapoportPathDOTFormalizing2019}.
\autoref{fig:comparing-let-path-reduction} shows that pDOT transitively looks up paths when reducing apply forms $q\,p$: being unable to look up $q$ would result in a stuck term.
However, \calculus{} case forms allow scrutinees of arbitrary types and such scrutinees may refer to a circularly-defined field.
To avoid transitive lookups of such fields, \calculus{} allows reducing paths, letting case forms with such scrutinees loop endlessly.
This mirrors how looking up lazily defined, circular fields in Scala also results in endless loops.

\subsection{Metatheory}
\label{sec:metatheory}

We provide a mechanized soundness proof for \calculus, basing on soundness proof for \pdot{} \cite{rapoportPathDOTFormalizing2019}, at its core formulated in the standard progress-and-preservation style \cite{wright1994}.
Our Coq proof scripts have $\approx$ 13100 lines of code, compared to $\approx$ 9800 lines in the \pdot{} soundness proof.
We attach the proof as an artifact.

Evaluating any well-typed \calculus{} program either diverges or results in a normal form.
\begin{theorem}[Type Safety]
If $\typ {} t T$, then the evaluation of $t$ either diverges, 
or $\reduction {\emptysetAlt} t {\sta} {u}$ such that
(a) $u$ is a normal form under $\sta$; (b) $\typDft u T$ for some $\G$ and (c) \tnew{$\wf \G \sta$}.
\end{theorem}
\noindent
\tnew{$\wf \G \sta$} means that $\sta$ conforms to $\G$, i.e. for any variable $x_i \in \mathrm{dom}\, \G$, if $\G(x_i) = T_i$ and $\sta(x_i) = v_i$, then $\typDft {v_i} {T_i}$.
Normal forms are values and resolved paths.

\subsubsection{Soundness Proof: From \pdot{} to \calculus{}}
Our soundness proof extensively reuses the infrastructure and toolkits provided by \pdot{}.
The strategy central to the proof of \pdot{} is to stratify typing rules into seven levels \cite{rapoportPathDOTFormalizing2019}, as illustrated in the following diagram:
\begin{align*}
  {\tiny
  \text{General } (\vdash) \rightarrow
  \text{Tight } (\turnstileTight)\rightarrow
  {\text{{Introduction-$qp$ }} (\turnstileRepl)} \rightarrow
  \text{Introduction-$pq$ }(\turnstileInvertible) \rightarrow
  {\text{Elim-III }(\turnstilePrecThree)} \rightarrow
  {\text{Elim-II }(\turnstilePrecTwo)} \rightarrow
  \text{Elim-I }(\turnstilePrecOne)}
\end{align*}
The stratified typing level hierarchy tackles the bad bounds problem \cite{amin2016} and eliminates cycles in type derivation \cite{rapoportPathDOTFormalizing2019,simpleSoundnessProof}.
Innermost typing levels assign precise typing information to paths directly followed from the environment and path aliasing.
It is easy to do induction and reason about typing judgments on these levels.
With a series of transformation lemmas converting judgments between levels,
the general recipe of \pdot{} soundness proof is to transform the surface typing level (general typing $\vdash$) to inner typing levels and then reason about the typing judgments.
(See Appendix \ref{sec:soundness-proof-of-pdot} for a detailed description.)

The soundness proof of \calculus{} follows the stratified typing levels
of the \pdot{} soundness proof.
We describe our extensions in detail in \autoref{sec:invertible-subtyping}.
We state Preservation and Progress as follows:
\begin{theorem}[Preservation] \label{thm:preservation}
  Let $\G$ be an \tnew{inert}, \tnew{well-formed} typing environment and let $\wf \G \sta$.
  If $\reduction \sta t {\sta^\prime} {t^\prime}$ and $\typDft t T$,
  then there exists an \tnew{inert} and \tnew{well-formed} typing environment $\G^\prime$
  such that $\wf {\sta^\prime} {\G^\prime}$ and $\typ {\G^\prime} {t^\prime} T$.
\end{theorem}
\begin{theorem}[Progress] \label{thm:progress}
  Let $\G$ be an \tnew{inert}, \tnew{well-formed} typing environment and let $\wf \G \sta$.
  If $\typDft t T$,
  then either $t$ is a normal form,
  or there exists $\sta^\prime$ and $t^\prime$, such that $\reduction \sta t
  {\sta^\prime} {t^\prime}$.
\end{theorem}
An environment is \tnew{inert} if it only contains inert types, i.e. types that cannot introduce subtyping assumptions.
A type is inert if it is a lambda type or an object types with tight type members (see Appendix \ref{sec:soundness-proof-of-pdot}).
An environment $\G$ is \tnew{well-formed} if every path mentioned in $\G$ is typeable in $\G$.
Notably, \rn{Sngl-Self} is necessary to show Preservation.
For example, consider the following term:
\begin{equation*}
	\ttCase p {q.A} y {y} {t_2}
\end{equation*}
We assume that both $p$ and $q$ are resolved, the pattern is matched and $t_2$ is a term that can be typed as $\single p$.
We can type the case term $t$ as $\single p$ with the \rn{Case} rule.
Then, when reducing the case term, we apply the \rn{Case-Then} rule and reduce the term to $\tSubst y p y = p$.
Here, to prove preservation, $\typDft p {\single p}$ must be derivable, which is only made possible by the \rn{Sngl-Self} rule.

The Progress proof in the case for terms of the form $\ttCase{p}{{\resolvedDft}.A}{y}{t_1}{t_2}$
relies on \autoref{thm:tag-resolution}, which shows if $p$ resolves to an object $\tNew r A x T d$ then $r$ also can be resolved.
Otherwise, neither \rn{\redCaseThen} nor \rn{\redCaseElse} would apply, making the term stuck.
\begin{lemma}[Tag Resolution] \label{thm:tag-resolution}
  Let $\G$ be inert and well formed. If $p$ is typeable in $\G$, $\wf \G \sta$ and $\lookupStepDft{p}{\tNew{r}{A}{x}{T}{d}}$,
  then there exists $\resolvedDft$ such that $\lookupDft {r} {\resolvedDft}$.
\end{lemma}

\begin{wide-rules}
  \begin{multicols}{3}

\infrule[Typ-$<:$-Typ$_{\#\#}$]
{\new{\subTightDft {S_2} {S_1}}
	\andalso
	\new{\subTightDft {T_1} {T_2}}
}{
  \subInvDft {\tTypeDec A {S_1} {T_1}} {\tTypeDec A {S_2} {T_2}}}

\infrule[$<:$-Sel$_{\#\#}$]
{\typPrecThreeDft {p} {\tTypeDec A T T} \\
  \new{\subInvDft S T}}
{\subInvDft {\new S} {{p}.A}}

\infrule[Sel-$<:$$_{\#\#}$]
{\typPrecThreeDft {p} {\tTypeDec A T T} \\
  \new{\subInvDft T U}}
{\subInvDft {p.A} {\new U}}

\infrule[Sngl$_{pq2}$-$<:$$_{\#\#}$]
{\typPrecThreeDft p {\single q} \andalso \typeablePrecTwo q \\
 \subInvDft {\new S} T}
{\subInvDft {\new S} {\repl p q T}}

\newruletrue
\infrule[Sngl$_{pq1}$-$<:$$_{\#\#}$]
{\typPrecThreeDft p {\single q} \andalso \typeablePrecTwo q \\
 \subInvDft S T}
{\subInvDft {\repl p q S} T}
\newrulefalse

\infrule[Sngl$_{qp2}$-$<:$$_{\#\#}$]
{\typPrecThreeDft p {\single q} \andalso \typeablePrecTwo q \\
 \subInvDft {\new S} T}
{\subInvDft {\new S} {\repl q p T}}

\newruletrue
\infrule[Sngl$_{qp1}$-$<:$$_{\#\#}$]
{\typPrecThreeDft p {\single q} \andalso \typeablePrecTwo q \\
 \subInvDft S T}
{\subInvDft {\repl q p S} T}
\newrulefalse

  \end{multicols}

\normalsize
\caption{Invertible subtyping rules (interesting cases, with changes highlighted).}
\label{fig:invertible-subtyping-rules}
\end{wide-rules}

\subsubsection{Proving Inversion Rules with Invertible Subtyping}
\label{sec:invertible-subtyping}
When implementing the soundness proof, we add inversion rules only to the surface typing level, and keep the inner levels unchanged.
Then, we only need to re-prove the transformation lemma from general typing to the tight level.
This is to prove that the inversion rules are \emph{admissible} in tight typing when the environment is inert,
which is stated in the following two lemmas.
\begin{lemma}[Field Inversion] \label{thm:field-inversion}
	In an \tnew{inert} environment $\G$, if $\subTightDft U {\tFldDec a {T_2}}$ and
  $\uniqueFlow U {\tFldDec a {T_1}}$, then $\subTightDft {T_1} {T_2}$.
\end{lemma}

\begin{lemma}[Type Member Inversion] \label{thm:type-member-inversion}
	In an \tnew{inert} environment $\G$, if $\subTightDft U {\tTypeDec A {S_2} {T_2}}$ and
  $\uniqueFlow U {\tTypeDec A {S_1} {T_1}}$, then $\subTightDft {S_2} {S_1}$ and
  $\subTightDft {T_1} {T_2}$.
\end{lemma}
$\subTightDft S T$ denotes tight subtyping, a notion inherited from \pdot{} (see Appendix \ref{sec:soundness-proof-of-pdot}).
Tight typing and inertness forbid the derivation of subtyping relations from type member bounds.
They can rule out absurd subtyping relations, such as $\tFldDec{a}{\forall(x: S)T} <: \tFldDec{a}{\bot}$, which cannot be inverted.

The major obstacle when proving subtyping inversion is \rn{Trans}.
For instance, if we try to invert the tight subtyping relation $\subTightDft {\tFldDec a S}
{\tFldDec a T}$ to get $\subTightDft S T$ by induction on its derivation,
we become stuck at the \rn{Trans} case.
There, we have $\subTightDft {\tFldDec a S} U$ and $\subTightDft U {\tFldDec a T}$ as premises
and the inductive hypothesis is not applicable since the shape of $U$ is unknown.

To tackle this problem, we add invertible subtyping as an additional layer (see \autoref{fig:invertible-subtyping-rules}).
In this layer, we remove \rn{Trans} and instead inline it in the path replacement and type member subtyping rules.
We have four rules for path replacement in invertible subtyping, compared to two in general and tight subtyping.
This is because general subtyping relies on transitivity to replace aliasing paths on the left hand side.
If $\subTightDft S T$ and $p$ aliases $q$, to derive $\subTightDft{\repl p q S} T$ we first need to derive $\subTightDft {\repl p q S} S$ by \rn{Sngl$_{qp}$-<:$_{\#}$} and then derive $\subTightDft {\repl p q S} T$ by \rn{Trans}.

Invertible subtyping can easily be inverted and inductively reasoned about.
Similarly to typing, we prove that tight subtyping can be transformed to invertible subtyping.
\begin{lemma}[$<:_{\#}$ to $<:_{\#\#}$] \label{thm:tight-to-invertible}
  If $\G$ is an inert typing environment and $\subTightDft S T$, then
  $\subInvDft S T$.
\end{lemma}
\noindent
And it is easy to prove that we can transform invertible subtyping to tight subtyping.
\begin{lemma}[$<:_{\#\#}$ to $<:_{\#}$] \label{thm:invertible-to-tight}
	For any environment $\G$, if $\subInvDft S T$ then $\subTightDft S T$
\end{lemma}

\begin{lemma}[Field Inversion in $<:_{\#\#}$]
  If $\subInvDft {\tFldDec a S} {\tFldDec a T}$, then $\subInvDft S T$.
\end{lemma}
\emph{Proof sketch.}
We can prove the theorem by induction on the derivation of the judgment.
The \rn{Refl$_{\#\#}$} and \rn{Fld-<:-Fld$_{\#\#}$} cases are trivial.
    In the first case we get $S = T$.
    The goal can be proven by applying the reflexivity rule.
    In the second case the goal is already proven.
The other four cases from path replacement can be proven with the inductive
hypothesis and the typing rules.

\section{Translating Constraint-Based GADTs into \texorpdfstring{\calculus{}}{cDOT}}
\label{sec:validation-encoding}
In this section, we demonstrate the expressiveness of subtyping reconstruction by proving that it subsumes generalised algebraic data types (GADTs).

\newcommand\ts{\vdash}
\newcommand\entl{\vDash}

\subsection{Introducing GADTs}
Algebraic Data Types (ADTs) are a feature of ML-like languages which allow representing data types that can be created with a set of \emph{constructors} known in advance.
They are similar to class hierarchies one level deep, like @Expr@ from \autoref{fig:motiv-example}; ADT constructors correspond to subclasses of @Expr@.
However, every constructor of a polymorphic ADT must have the exact same type parameters as the ADT itself.
To correspond to an ADT, each subclass of @Expr<A>@ would need to be generic in @A@ and would need to extend @Expr<A>@, a restriction which is violated by @IntLit@.
GADTs lift this restriction: GADT constructors may have arbitrary type parameters, and may apply the GADT type constructor to arbitrary types.
In particular, @Expr<A>@ can be faithfully represented as a GADT.

\citeauthor{xi2003} studied GADTs in the ML setting and presents the GADT calculus \Xis{}.
Given an appropriate definition of {\tt Expr}, we can express the {\tt eval} example in \Xis{} as follows:
\begin{align*}
  \ident{eval} &: \ident{Expr}\, \alpha \rightarrow \Int \\
  \ident{eval} &= \Lambda\, \alpha.\, \lambda x : \ident{Expr}\, \alpha.\, \kCase x \kOf \ident{IntLit}(i) => i \;|\; \dots
\end{align*}

To support GADTs, the typing judgment in \Xis{} keeps track of a {\it constraint} $\Delta$;
the judgment itself has the form $\Delta;\Pi |- t : \tau$, where $\Pi$ stands for a typing context in \Xis{}.\footnote{
  \citeauthor{xi2003} used $\Gamma$ for their environments as expected. We use $\Pi$ instead, to avoid reusing the $\Gamma$ metavariable.}
Pattern matching on GADTs extends $\Delta$ with additional information.
As long as the constraint {\it entails} equivalence between some types (denoted as $\Delta \entl \tau_1 \equiv \tau_2$),
they are interchangeable to the typing judgment.
For instance, in the above $\ident{eval}$ example, the constraint $\Delta$ used to type the $\ident{IntLit}$ branch would contain $\alpha \equiv \Int$,
which we could use to type $i$ as $\alpha$ instead of $\Int$.
The entailment relationship in \Xis{} is based on substitution:
we have $\Delta \entl \tau_1 \equiv \tau_2$ if for any $\theta$ that solves $\Delta$, we also have $\theta\tau_1 \equiv \theta\tau_2$.
To illustrate, given $\ident{List}\, \alpha \equiv \ident{List}\, \Int \in \Delta$, we would also have $\Delta \entl \alpha \equiv Int$.

Our notion of subtyping reconstruction subsumes GADTs from \Xis{}: it supports subtyping and class hierarchies of arbitrary depth.
\calculus{} does not need an additional context $\Delta$, since its bindings can introduce subtyping assumptions which later can be inverted (\autoref{sec:inversion-subtyping}).
In order to make this claim precise and validate it, we show an encoding of a variant of \Xis{} into \calculus.

\subsection{A Variant of \texorpdfstring{\Xis}{the GADT System}}
\label{sec:variant-of-xi}

We have adjusted \Xis{} by removing those features which are not relevant for the purpose of demonstrating that \calculus{} can encode GADTs.
In order of significance, the changes are:

\subsubsection*{Adjusted constraint reasoning.}
Besides the constraint-based entailment, \Xis{} also features an equivalent relation $\Delta |- \tau_1 \equiv \tau_2$ defined syntactically \cite[Fig. 7]{xi2003}.
Our variant of \Xis{} uses a restricted version of this relation, which we denote as $|-'$.
Our changes are limited to removing rules that allow ex-falso reasoning, as well as reasoning about function and universal types.
We give the exact rules in the appendix.

Ex-falso rules allow deriving arbitrary type equivalences from a contradictory constraint.
Branches which introduce contradictory constraints would never be entered at runtime, since they correspond to patterns that cannot match.
Such cases could in principle be handled while preserving their semantics by replacing the branch bodies with a loop, for instance.
We excluded these rules, as programming languages that support GADTs (for instance OCaml, Haskell, Agda) typically do not allow typing arbitrary code if a branch cannot be entered.

Similarly, we have excluded the rules that allow reasoning based on equality of function and universal types.
Again, such reasoning is not core to GADTs and not all implementations of GADTs allow it
(for instance, Agda and Idris don't invert function types
and Haskell does not invert polymorphic types).
We still allow inverting equalities of GADTs and tuples, for example given $\alpha * \beta \equiv \gamma * \delta \in \Delta$ we can infer $\Delta |-' \alpha \equiv \gamma$.

To show what the rules defining $\vdash'$ look like, we present two of them below.
One allows replacing a type with a variable equal to it, another allows inverting GADT equality.

\vspace{0.5em}
\begin{minipage}{0.5\linewidth}
  \infrule{
    \alpha \text{ has no free occurrences in } \tau \\
    \vec{\alpha}, \Delta[\alpha \mapsto \tau] \vdash' \tau_1[\alpha \mapsto \tau] \equiv \tau_2[\alpha \mapsto \tau]
  }{
    \vec{\alpha}, \tau \equiv \alpha, \Delta \vdash' \tau_1 \equiv \tau_2
  }
\end{minipage}%
\begin{minipage}{0.5\linewidth}
\begin{prooftree}
  \AxiomC{$\vec{\alpha}, \vec{\tau_A} \equiv \vec{\tau_B} \vdash' \tau_1 \equiv \tau_2$}
  \UnaryInfC{$\vec{\alpha}, (\vec{\tau_A}) T \equiv (\vec{\tau_B}) T \vdash' \tau_1 \equiv \tau_2$}
\end{prooftree}
\end{minipage}

\subsubsection*{Deterministic reduction.}
\Xis{} reduces $\kCase$ forms non-deterministically if multiple patterns %
match, which is a strange choice of evaluation semantics.
Handling non-determinism in \calculus{} would be quite involved, as we would either need to modify the evaluation rules of pDOT to be non-deterministic or encode non-determinism.
While either choice is an interesting problem, non-determinism is orthogonal to GADTs and reduction is deterministic in our variant of \Xis{}.

\subsubsection*{Exhaustive and not-nested patterns.}
In contrast to \citeauthor{xi2003}, our variant of \Xis{} disallows nested patterns and requires $\kCase$ forms to match on all GADT constructors.
Nested patterns can be desugared to remove the nesting.
Many languages, like Coq for example \cite{coqpatterns}, desugar such matches to make compilation into primitive operations easier.
See \cite{DecisionTrees} for an efficient algorithm converting nested pattern matches into decision trees.

Original \Xis{} allowed inexhaustive pattern matches, which were allowed to get stuck.
We instead require patterns to be exhaustive; verifying this is simple without nested patterns.

\subsection{Metatheory}
In order to validate our variant of \Xis, we provide a mechanised proof of soundness, expressed as standard progress and preservation theorems.
Thanks to the additional exhaustivity requirement we added to pattern matching, well-typed terms never get stuck:
we can prove progress, which did not hold in the original formulation.
The mechanized proof was originally described by \citet{RadekThesis} and is attached as one of the artifacts of this paper.

\begin{restatable}[Preservation]{thm}{metaPreserv}
  For any well typed term $\Delta; \Pi \vdash e : \tau$,
  if $e \longrightarrow e'$,
  then the other term is also well-typed and has the same type: $\Delta; \Pi \vdash e' : \tau$.
\end{restatable}
\begin{restatable}[Progress]{thm}{metaProg}
  For any well typed term $\Delta; \Pi \vdash e : \tau$,
  the term is either a value or it can be reduced further: there exists a term $e'$ such that $e \longrightarrow e'$.
\end{restatable}

\subsection{Encoding}
\label{sec:encoding}

To encode \Xis{} in \calculus, we emulate GADTs using object types equipped with type members
which will serve the purpose of encoding type parameters and evidence for type equalities
together with a visitor function (analogous to Scott-encoding \cite[Sec. 5.2]{stump2009}) which will allow to emulate exhaustive pattern matching together with refinement of type equalities.
As our encoding does not depend on $\kCase$ forms, it remains applicable to other systems in the DOT family which may potentially support subtyping reconstruction, such as gDOT \cite{giarrusso2020}.

We denote an encoding of a term $e$ as $\E(e, \Theta)$ and an encoding of a type $\tau$ as $\T(\tau, \Theta)$.
Both definitions depend on a substitution $\Theta$ mapping \Xis{} type names $\alpha$ to type projections $p.A$ and \Xis{} variable names $x$ to \calculus{} expressions.
For example, we may have: $\Theta = [\alpha \mapsto x_{\alpha}.T, x \mapsto y]$.

In order to encode a complete \Xis{} program, we define a function to encode the signature $\Sigma$ and a term \texttt{lib} representing basic primitives like the unit and tuple types.
The signature $\Sigma$ defines what GADTs definitions a program may use:
it maps GADT constructors to their signatures, for example: $\Sigma(c_j) = \forall \vec{\beta_j}. \tau_j \to (\vec{\sigma_j}) T$.
We show a full definition in the appendix.
Overall, encoding of a \Xis{} program $e$ under a signature $\Sigma$ is defined as follows:

\begin{lstlisting}[mathescape=true,basicstyle=\small\ttfamily]
  EncodeProgram(e, $\Sigma$) =
    $\textbf{let}$ lib = ... $\textbf{in}$ $\textbf{let}$ env = EncodeSigma($\Sigma$) $\textbf{in}$ $\E$(e, $\varnothing$)
\end{lstlisting}

\noindent
\autoref{fig:expr-translation} shows how we would translate a \Xis{} version of @Expr@ from \autoref{fig:motiv-example} into \calculus.

\begin{figure}
  \begin{align*}
    \Let\ &\lib\ = \nu(\lib)[\lib.A]\{ \\
    & A = \top; \\
    & \ident{Expr} = \mu(x: \tTypeDec{X}{\bot}{\top} \wedge \tFldDec{\ident{pmatch}}{\forall(r: \tTypeDec{R}{\bot}{\top}) \forall(\ident{litIntCase}: \lib.\ident{IntLit} \rightarrow {r.R})\, {r.R}}); \\
    & \ident{IntLit} = \mu( \\
    & \quad x: \tTypeDec{X}{\Int}{\Int} \wedge \tFldDec{\ident{data}}{\lambda(\_)\, x.X} \\
    & \quad \quad\wedge \tFldDec{\ident{pmatch}}{\forall(r: \tTypeDec{R}{\bot}{\top}) \forall(\ident{litIntCase}: \lib.\ident{IntLit} \rightarrow {r.R})\, {r.R}} \\
    &); \\
    & \ident{mkIntLit}\colon \forall(x: \Int)\, \lib.\ident{IntLit}; \\
    & \ident{mkIntLit} = \lambda(x: \Int) \nu(y)[\lib.\ident{IntLit}]\{ \\
    & \quad X = \Int; \ident{data}: \forall(\_)\, y.X; \ident{data} = \lambda(\_)\, x; \\
    & \quad \ident{pmatch}\colon {\forall(r: \tTypeDec{R}{\bot}{\top}) \forall(\ident{litIntCase}: \lib.\ident{IntLit} \rightarrow {r.R})\, {r.R}} \\
    & \quad \ident{pmatch} = {\lambda(r: \tTypeDec{R}{\bot}{\top}) \lambda(\ident{litIntCase}: \lib.\ident{IntLit} \rightarrow {r.R})\  \ident{litIntCase}\ {y}} \\
    & \}; \\
    & \ident{eval}\colon \forall(x: \lib.\ident{Expr})\, x.X; \\
    & \ident{eval} = \lambda(x: \lib.\ident{Expr})\, \ttLet{y}{\lambda(y: \lib.\ident{IntLit})\ {y.\ident{data}}}{{x.\ident{pmatch} \ y}} \\
    &\!\!\!\!\!\!\!\!\! \} \ \In \ \cdots
  \end{align*}
  \caption{Translating \lstinline!Expr! from \Xis{} into \calculus{}}
  \label{fig:expr-translation}
\end{figure}

\subsection{Typing Preservation}
\label{sec:typing-preservation}

To prove typing preservation, we first define a notion of correspondence between environments in \Xis{} and \calculus{}.
We define \calculus{} type equality as follows:

\begin{restatable}[Type equality in \calculus{}]{definition}{typEqualityCDOT}
  We introduce a shorthand notation $\Gamma \vdash A =:= B$ which just stands for $\Gamma \vdash A <: B \wedge \Gamma \vdash B <: A$.
\end{restatable}

Then we formalize that our environment contains definitions of \texttt{lib} and \texttt{env} as expected by the described encoding.

\begin{restatable}[Signature encoding]{definition}{sigEncoding}
  \label{def:sigEnc}
  For any $\lambda_{2, G\mu}$ signature $\Sigma$ we denote by $\mathcal{S}(\Sigma)$ the types defined by its \calculus{} encoded counterpart.
  $\mathcal{S}(\Sigma)$ contains two variables:
  \begin{itemize}
    \item $\lib$ which is defined as in the encoding, giving us the builtin types,
    \item $\env$ which encodes the signature $\Sigma$: for every $\Sigma(c_i) = \forall \vec{\alpha}. \tau_A \to \tau_B$ we have:
      \begin{align*}
        &\env.c_i : \forall(ts: \{ B_{i,1}: \bot..\top; \ldots \}) \; \forall(v: \T(\tau_A, \theta_{\vec{\alpha}})) \; \T(\tau_B, \theta_{\vec{\alpha}}) \\
        &\qquad \text{where}\; \theta_{\vec{\alpha}} = [\alpha_1 \mapsto ts.B_{i,1}, \ldots, \alpha_m \mapsto ts.B_{i,m}]
      \end{align*}
  \end{itemize}
\end{restatable}

Now we can define a property specifying that all equations from the original environment $\Delta$ are also analogously emulated in the encoded environment $\Gamma$.

\begin{restatable}[Equational correspondence $\Psi(\Delta; \Theta; \Gamma)$]{definition}{eqCorr}
  A \calculus{} environment $\Gamma$ together with a substitution $\Theta$ satisfies equations of $\Delta$ if the following conditions are met:
  \begin{enumerate}
    \item For each $\alpha \in \Delta$, for $\Theta(\alpha) = p.A$, $\Gamma \vdash p : \{ A : \bot .. \top \} $.
    \item For each $\tau_1 \equiv \tau_2 \in \Delta$, $\Gamma \vdash \T(\tau_1, \Theta) =:= \T(\tau_2, \Theta)$.
    \item $\Gamma$ contains the bindings \texttt{lib} and \texttt{env} as defined in $\mathcal{S}(\Sigma)$.
  \end{enumerate}
\end{restatable}

We also introduce a similar notion for data environments: every variable in the original environment $\Pi$ also has its encoded counterpart, typeable to the encoded variant of its original type.
We gather both notions in a single property:

\begin{restatable}[Environment correspondence $\Phi(\Delta; \Pi; \Theta; \Gamma)$]{definition}{envCorr}
  A \calculus{} environment $\Gamma$ corresponds to $\lambda_{2, G\mu}$ contexts $\Delta$ and $\Pi$ under a substitution $\Theta$
  if $\Psi(\Delta; \Theta; \Gamma)$ and for each $\mathit{xf}$ such that $\Pi(\mathit{xf}) = \tau$, we have $\Gamma \vdash \Theta(\mathit{xf}) : \T(\tau, \Theta)$.
\end{restatable}

With all that, we can state the main theorem.

\begin{restatable}[Type Preservation]{thm}{encTypPres}
  \label{thm:typepreservation}
  For every $\Pi$, $\Delta$, $e$, $\tau$,
  any $\Theta$ such that $\mathit{fv}(e) \cup \mathit{fv}(\tau) \subseteq \mathit{dom}(\Theta)$,
  any $\Gamma$ such that $\Phi(\Delta; \Pi; \Theta; \Gamma)$,
  we have $\Gamma \vdash \E(e, \Theta) : \T(\tau, \Theta)$.
\end{restatable}

Based on this theorem and correctness of the \texttt{env} and \texttt{lib} encodings (shown in the appendix),
we can prove that for any well-typed \Xis{} program, the result of \texttt{EncodeProgram} is also a well-typed program in \calculus{}.

Apart from other more technical lemmas, an important lemma that we are using in the proof is one that specifies that under our modified notion of constraint reasoning, if an equation is derivable using this reasoning, then an analogous equation is also derivable in \calculus{} starting from analogous assumptions:

\begin{restatable}[Equation emulation]{lemma}{eqEmulation}
  \label{lemma:eqpreserved}
  For any $\Delta \vdash' \tau_1 \equiv \tau_2$ and any $\Theta$ and $\Gamma$, if $\Psi(\Delta; \Theta; \Gamma)$, 
  then $\Gamma \vdash \T(\tau_1, \Theta) =:= \T(\tau_2, \Theta)$.
\end{restatable}

\subsection{Preservation of Operational Semantics}

We conjecture that our encoding preserves the operational semantics of \Xis{} in the sense that
if $e$, a \Xis{} term, reduces to some $e'$, then $\ident{Enc}(e) -->^{*} t'$ such that $\ident{Enc}(e')$ corresponds to $t'$.
This is illustrated by the following diagram, where $\eqsim$ is an equivalence relation described below:
\newcommand{\arrowstar}{\vphantom{\to}^{*}}
\tikzset{
	startip/.tip={Glyph[glyph math command=arrowstar]},
	tostar/.style={->.startip}}
\begin{center}
\begin{tikzcd}
	e \arrow[d, "Enc"] \arrow[r, tostar] & e' \arrow[d, "Enc"] \\
	\ident{Enc}(e) \arrow[r, tostar] & t' \eqsim \ident{Enc}(e')
\end{tikzcd}
\end{center}

\noindent
The main technicality which would be involved in the proof of the above property is that reduction in \Xis{} relies on substitution as opposed a store like \calculus{}.
We envision the gap between the two can be bridged by defining an equivalence relation $t \eqsim s$ on \calculus{} terms.
Two terms should be considered equivalent if they are the same modulo $\alpha$-equivalence as well as order of let- and store-bindings and presence of unused bindings.
In particular, this equivalence makes it possible to deal with the fact that reduction in \Xis{} inherently may duplicate a value, while reduction in \calculus{} only duplicates references to a store-binding.

Aside from this problem, the proof should be rather straightforward: most forms from \Xis{} are translated to their direct \calculus{} equivalents,
and our encodings of type abstraction and tuples are standard \cite{amin2016}.
The most involved case of the proof will be correctness of the pattern matching translation, which relies on encoding $\kCase$ forms into calls to $\ident{pmatch}$ like in Scott encoding \cite[Sec. 5.2]{stump2009}.
In particular, the proof will have to demonstrate that replacing contradictory branches by infinite loops in the encoded term preserves the operational semantics;
intuitively, this is because a contradictory branch cannot be entered, which is precisely why \Xis{} allows typing any term as any type in such branches in the first place.

\newcommand\atyp[1][A]{\hat{#1}}
\newcommand\bctx[1][\Psi]{\hat{#1}}

\newcommand\alor{\mathbin{\widetilde{\lor}}}

\section{Implementing Subtyping Reconstruction}
How can we implement subtyping reconstruction?
We begin by inspecting the basic problem with writing an algorithm that can use subtyping relationships reconstructed by pattern matching.

\subsection{Towards Formal Decidable Subtyping}

\newcommand\Emirs{C$^\sharp$ minor\xspace}

Recall that \calculus{} term variables can introduce subtyping assumptions to the context
(see \emph{subtyping inversion} in \autoref{sec:inversion-subtyping}).
As a concrete example, consider $\Gamma = x: \tTypeDec{A}{\bot}{\top}, y: \tTypeDec{B}{\bot}{\top}, z: \tTypeDec{C}{x.A}{y.B}$.
We can derive $\Gamma |- x.A <: y.B$ thanks to the bounds of $z.C$.

The question is how to design an algorithmic subtyping judgment which is a subset of the declarative subtyping rules from \autoref{fig:subtyping-rules}.
More specifically, how do we decide what subtyping assumptions are introduced by bindings and when to use the inversion rules?
\citet{Kennedy05:gadts-oo} show \Emirs, a formalisation of C$^\sharp$ which allows pattern matching much like \calculus.
The declarative subtyping judgment in \Emirs has the form $\Delta |- S <: T$, where $\Delta$ is a list of subtyping assumptions. %
This judgment features inversion rules which allow %
recovering the premises of assumptions, again much like \calculus.
In their system, a branch of a pattern match may introduce subtyping assumptions;
specifically, if $Q$ is the type of the value being matched and $P$ is the type of values matched by the pattern,
then the subtyping assumptions introduced by the branch can be found by finding a class $K$ such that $Q <: K[\overline{T}]$ and $P <: K[\overline{S}]$,
and then relating $S_i$ to $T_i$ based on type parameter variance.
Conceptually, the problem is much simpler in \calculus:
the assumptions introduced by a binding can be found by simply listing all of its type members and then relating each member's upper and lower bounds.
\calculus{} also clearly points to the reason we can relate $S_i$ to $T_i$;
if the branch is entered, we have a value $v$ which can be typed as both $Q$ and $P$,
and this value's type member (representing the \emph{most precise} type argument to $K$) is related to both $S_i$ and $T_i$.

This leaves the question of deciding when to use the inversion rules.
\citet{Burak06:variance-gen-constraints-csharp} solve the problem for \Emirs:
they show an algorithmic subtyping judgment $\Psi |- S <: T$,
with $\Psi$ being a list of \emph{upper and lower bounds} on type variables of the form $X <: T$ and $T <: X$.
This judgment no longer needs inversion rules: the assumptions in $\Psi$ are as simple as possible and cannot be inverted.
$\Delta$ and $\Psi$ are related: any subtyping assumption $S <: T$ can be \emph{decomposed} into bounds by repeatedly applying appropriate inversion rules;
we denote the result of doing so as $\ident{Dec}(S <: T)$.
Any $\Delta$ can be converted to an equivalent $\Psi$ by first decomposing assumptions in $\Delta$
and then repeatedly simplifying the assumptions transitively introduced by the bounds, i.e. by reaching a fixpoint of the following sequence:
\begin{align*}
  \Psi_1 &= \bigcup \left\{ \ident{Dec}(S <: T) : S <: T \in \Delta \right\} \\
  \Psi_{i+1} &= \bigcup \left\{ \ident{Dec}(S <: T) : \{S <: X, X <: T\} \subset \Psi_i \right\} \cup \Psi_i
\end{align*}
\noindent
The algorithmic and declarative subtyping judgments are equivalent;
subtyping in \Emirs{} is decidable only if we restrict contravariant classes \cite{kennedy2007}.

Intuitively, we should be able to define an algorithmic subtyping judgment for \calculus{} with a similar approach,
though we leave the details of this process to future work.
This hypothetical subtyping would have the form $\Gamma ; \bctx |- S <: T$,
with $\bctx$ being a list of upper and lower bounds on \emph{abstract types} $\atyp$: type projections $p.A$ and singleton types $\single{p}$.
However, \calculus{} features intersection types and dependent function types (both absent from \Emirs), which make the problem more complex.
For instance, it is not possible in general to decompose $\atyp[A] \land \atyp[B] <: T$ into bounds.
These difficulties would be encountered when implementing subtyping reconstruction for any sufficiently advanced type system, such as the one in Scala or TypeScript.
Still, it is possible to have an algorithm which gives an incomplete approximation of the declarative subtyping.
For the particular case of \calculus{}, we are unlikely to do better.
We can encode System $F_{<:}$ in \calculus{} \cite{amin2014},
making it likely that \calculus{} subtyping is undecidable \cite{hu2020}.

\subsection{Outline of the Scala Implementation}

The first author of this paper was at the center of the effort to implement subtyping reconstruction for Scala 3.
Based on this experience, we outline key aspects of the implementation and relate them to the formalism developed in the present paper.

The sketch from the previous section gives us the two basic problems the implementation needs to solve:
finding subtyping assumptions introduced by patterns, and decomposing them into bounds so that they can be used for type checking.
In the implementation, the abstract types $\atyp$ for which bounds can be reconstructed include type projections $p.A$ and singleton types $\single{p}$,
as well as type parameters and pattern-bound types,\footnote{
  In Scala, types starting with a lower-case letter mentioned in a pattern are \emph{bound} by the pattern, c.f. \autoref{sec:uncovering-tps-rels}.
} both denoted as $X$.
As the assumptions following from a pattern are decomposed into bounds, we store the latter in a map from abstract types to pairs of types: upper and lower bounds of the abstract types.\footnote{
  The particular data structures of the Scala implementation were described by \citet{xu2021}.
}
This map will later be used for type checking the branch associated with the pattern, making it part of the context used to type check Scala code.
Before we describe how we find assumptions and decompose them, we should note that the implementation can be incomplete:
it is always sound to reconstruct less precise bounds than the ones permitted by the formalism.
Being able to do so is helpful when dealing with a language with a complex type system.
The Scala implementation at first only supported reconstructing bounds on function type parameters and finding assumptions following from class\footnote{
  When discussing Scala, we follow \citet{martres2022} and say ``class'' to mean either a \emph{proper} class or a trait.
} types; it was then progressively expanded to support bounds on class type parameters, type projections and singleton types, as well as assumptions following from structural, intersection and union types.%
\footnote{The support for type projections, singleton types, and structural types is currently under review and not yet merged.}

We now inspect the two problems in a simplified setting.
The bounds are reconstructed based on the \emph{scrutinee} type $Q$, i.e. the type of the value being matched,
and the \emph{pattern} type $P$, i.e. the type of values matched by the pattern.
We initially disregard Scala's variant refinement, as well as intersection, union, and higher-kinded types.
We only consider two forms of pattern types: either $P = K$, where $K$ is a class, or $P = K[\overline{X}]$, where $\overline{X}$ are all pattern-bound types, i.e. fresh abstract types.
To illustrate, only the first two patterns in the following example would be allowed:
\begin{lstlisting}
  def eval[T](expr: Expr[T]): T = expr match
    case e : IntLit        => e.value
    case e : First[b,c]    => eval[b](e.pair)._1
    case e : Second[Int,c] => // pattern type disallowed
\end{lstlisting}
Observe how \emph{naming} type arguments in the above example mirrors how in \calculus{}
we need to use type projections $x.A$ to reference type arguments of a pattern type representing a generic class.
To illustrate, we show a \calculus{} term that formally models the two first cases (assuming a simple generalisation of $\kCase$ forms).
In particular, note how we reference the type argument of the encoded version of @First@:
\begin{align*}
  &\ident{eval} = \tLambda {\ident{tl}} {\tTypeDec T \bot \top }\,
    \tLambda {\ident{expr}} {g.\ident{Expr} \land \{ A : \ident{tl}.T..\ident{tl}.T \}} \\
  &\quad\quad\quad \kCase \ident{expr}
    \begin{aligned}[t]
      &\kOf e : g.\ident{IntLit} => e.\ident{value} \\
      &\;\,|\;\,  e : g.\ident{First}  => \ident{eval}(\nu(s){\set {T = {e.B}}})(e.\ident{pair}).\ident{\_1}
    \end{aligned}
\end{align*}

Essentially, there are three procedures involved in reconstructing bounds.
First, \texttt{ASSUM(Q,P)} lists subtyping assumptions following a value of type $Q$ being also of type $P$.
Second, \texttt{DEC(S <: T)} decomposes an assumption into a list of bounds.
Third, \texttt{BOUNDS(Q,P) = ASSUM(Q,P).map(DEC)} reconstructs \emph{bounds} following from $Q$ and $P$.
These definitions should be understood as pseudocode rather than the precise implementation.

To find subtyping assumptions based on scrutinee and pattern types $Q$ and $P$,
we look for classes $K$ such that $Q <: K[\overline{T}]$ and $P <: K[\overline{S}]$;
such classes can be found by walking the inheritance hierarchy and in general there might be more than one such class.
The relationships between the type arguments of each such $K$ are what we are looking for:
if the $i$-th parameter of C is covariant or invariant, we must have $S_i <: T_i$, and likewise for contravariance.
For instance, if $Q = \ident{Expr}[X]$ and $P = \ident{IntLit}$,
we have $Q <: \ident{Expr}[X]$ and $P <: \ident{Expr}[\Int]$
and we return $\Int <: X, X <: \Int$.

Then, each assumption we find is decomposed into bounds and stored in the aforementioned map.
Decomposing an assumption $S <: T$ is implemented similarly to checking if $S$ is a subtype of $T$.
The main difference is that if we encounter a (sub)goal of the form $\atyp <: T$ or $T <: \atyp$
which cannot be concluded based on current bounds, we add the relationship to the map of reconstructed bounds.
For instance, when decomposing $\ident{List}[X] <: \ident{List}[\Int]$, we encounter the subgoal $X <: \Int$ and store the relationship.
Much like in the formal sketch, adding a new bound may result in new transitive subtyping assumptions:
as a simple example, if we extend $\ident{List}[X] <: Y$ with $Y <: \ident{List}[\Int]$, it now follows that $\ident{List}[X] <: \ident{List}[\Int]$.
To handle such cases, we first decompose these transitive assumptions and then return to decomposing the original assumption $S <: T$.

What we described so far is very close to the algorithm sketched by \citet{Kennedy05:gadts-oo} and presented there more formally.
We will now briefly sketch how we support more complex type system features.
The implementation supports pattern types containing subelements of the form $K[\overline{T}]$ (generic classes applied to arbitrary types, not only pattern-bound fresh abstract types)
by behaving as though all type arguments were fresh abstract types, then checking if the concrete arguments conform to the bounds of the corresponding abstract type.
Due to higher-kinded types (HKTs), we need to consider the \emph{injectivity} of type operators when decomposing assumptions.
Specifically, if we encounter a (sub)goal of the form $F[\overline{S}] <: T$ or $S <: F[\overline{T}]$,
we proceed only if $F$ is an alias of a class type, and therefore is known to be injective.
This correctly handles situations where one can conclude that $F[\overline{S}] <: F[\overline{T}]$ with $\overline{S}$ and $\overline{T}$ being unrelated,
for instance when we define \lstinline!type F[X] = Int!.
While we could support annotating abstract types with their injectivity \cite{stolarek2015}, so far we did not find the complexity that this would introduce worth the corresponding gains.
Scala's variant inheritance\footnote{
  Recall that in Scala, only final and case classes invariantly inherit from their parents (see \autoref{sec:ty-ctor-variance}).
} forces adjustments to listing subtyping assumptions \texttt{ASSUM(Q,P)}.
If $P$ is an (applied) class type using variant inheritance and $Q <: K[\overline{T}]$ as well as $P <: K[\overline{S}]$, we only take the relationships between invariant type arguments as assumptions.
The support for intersection and union types is the most complex extension and relies on seeing reconstructed bounds $\bctx$ as \emph{constraints} on abstract types;
we unfortunately lack the space for an appropriately detailed description.\footnote{
  The most important aspect is that we need to consider constraint \emph{disjunctions} (see, for instance, \citet{petrucciani2019})
  in \texttt{ASSUM(Q,P)} when either type is a union type and in \texttt{DEC(S <: T)} when $S$ is an intersection type or $T$ is a union type.
  Since calculating exact disjunctions is too costly, we approximate them such that the exact constraint entails the approximation.
}

While developing the implementation, we always used a translation from class types into structural types as a guide for what bounds to reconstruct.
Contrary to expectations, correctly handling various features of Scala's type system required quite sophisticated notions.
Based on our findings, we believe that the complexity of the implementation may be reduced if we instead based it directly on \calculus{} concepts.
While we leave developing this alternative implementation as future work, we envision it as follows.
To find subtyping assumptions based on pattern and scrutinee types $P,Q$,
we would first translate them to structural types $\widetilde{P},\widetilde{Q}$
and then return the relationships between upper and lower bounds of the type members of $\widetilde{P} \land \widetilde{Q}$.
For example, @Expr[X]@ and @IntLit@ would be translated to @{@ @type@ @A@ @=@ @X@ @}@ and @{@ @type@ @A@ @=@ @Int@ @}@ and we would return @Int@ $<:$ @X@, @X@ $<:$ @Int@, like the current implementation.
Likewise, to decompose an assumption $S <: T$ into bounds,
we would first translate $S$ and $T$ into structural types $\widetilde{S}, \widetilde{T}$
and then decompose $\widetilde{S} <: \widetilde{T}$ into bounds according to \calculus{} inversion rules.
For example, to decompose @List[X]@ $<:$ @List[Int]@, we would translate it to @{@ @type@ @A@ @<:@ @X@ @}@ $<:$ @{@ @type@ @A@ @<:@ @Int@ @}@, which can be inverted into @X@ $<:$ @Int@.
We believe that this algorithm would at least match the precision of the current implementation,
would require less special cases and would potentially naturally extend to intersections and unions without resolving to constraint-like reasoning.
To illustrate, consider that the current implementation needs special reasoning to handle a covariant @Expr@ and @IntLit@ using invariant refinement,
while the alternative implementation would simply translate @Expr@ to @{@ @type@ @A@ $<:$ @X@ @}@ and @IntLit@ to @{@ @type@ @A@ @=@ @Int@ @}@.

\section{Related Work}

The earliest implementation of pattern matching based on runtime subtype checks we are aware of
is that of Modula-3 \cite{boszormenyi2012programming} with its @TYPECASE@ statements,\footnote{
  \url{https://www.cs.purdue.edu/homes/hosking/m3/reference/typecase.html}
} allowing type-safe conditional refinement based on the runtime type of a scrutinee.

\newcommand\sysfc{System F$_C$\xspace}
Subtyping in DOT systems is believed to be undecidable, as we can encode $F_{<:}$ bounded polymorphism in DOT.
\citet{hu2020} show that this reasoning does not constitute a proof and demonstrate $D_{<:}$,
an undecidable fragment of DOT without self-references and intersection types.
$D_{<:}$ subtyping itself has two fragments which have subtyping algorithms.
However, undecidability of \calculus{} subtyping is unlikely to be a problem in practice.
Specific class hierarchies (allowed in Scala) on their own make subtyping undecidable \cite{kennedy2007}.
Similarly, Java generics are also Turing-complete, as one can reduce a Turing machine to a fragment of Java \cite{Grigore2017}.

\citet{parreaux2019} describe that DOT type members can be used to encode GADTs.
They observe that doing so solves the soundness problem of variant inheritance \cite{giarrusso2013}.

GADTs are well-known and widely implemented \cite{peytonjones2006}.
\citet{xi2003} may have been the first to study them in an ML-like setting with \Xis.
The original OCaml GADT implementation was based on a type-and-constraint system like \Xis~\cite{garrigue2011}.
\citet{sulzmann2007} present \sysfc, an intermediate representation in the Haskell compiler which supports GADTs (and other features) with type equality witnesses and coercions \cite{mitchell1984}.
The Haskell implementation was previously based on a type-and-constraint system, which can be encoded into \sysfc \cite{sulzmann2007}.
In \sysfc{} the $\mathrm{IntLit}$ constructor of $\mathrm{Expr}\, \alpha$ contains a witness $\mathsf{co} : \mathrm{\alpha \sim \Int}$.
Pattern matching on $\mathrm{Expr}\,\alpha$ binds this witness, which then can be used to type $0$ as $\Int$ with an explicit coercion: $0 \blacktriangleright \mathrm{sym}\,\mathsf{co}$.
\citeauthor{sulzmann2007} claim that explicit casts make \sysfc code more robust to transformations by ensuring that they do not violate witness scoping.
\calculus{} objects also act like witnesses of subtyping assumptions, but they only participate implicitly in subtyping.
It would be interesting to see a variant of \calculus{} where objects participate in subtyping through explicit coercions, like \sysfc{} witnesses.

\citet{scherer2013} study the combination of GADTs and variance in the OCaml setting.
They demonstrate that associating the constructors of a GADT such as @Expr@ with type equality constraints prevents the GADT from being covariant.
The problem does not arise in our setting, where constructors are, essentially, associated with subtyping constraints.

GADTs present a unique challenge for checking pattern exhaustivity:
based on a value's type, we can rule out that it was created using particular constructors.
For instance, a value of type @Expr<Int>@ cannot be an instance of @MkPair@.
Without special support for GADTs, checking pattern exhaustivity emits spurious inexhaustivity warnings,
for instance that a pattern for @MkPair@ is missing when matching on an @Expr<Int>@.
\citet{karachalias2015} show a GADT-aware algorithm implemented in GHC, the Haskell compiler.
\citet{abg2017} builds on their work to propose an analogous algorithm for Scala.

Recently, \citet{martres2022}
has shown a complete version of the class encoding we sketched in \autoref{sec:presentation-calculus}.
In order to preserve subtyping relationships in general,
this encoding relies on subtyping between recursive types using rules from an understudied branch of the DOT family \cite{rompf2016}.
Such rules are unnecessary to encode specific class hierarchies (like in our sketch);
\citeauthor{martres2022} conjectures that not relying on them is almost always possible if encoding a class type accounts for generic classes it extends.
For instance, an applied type like @MkPair[T1,T2]@ would be encoded as
$g.\ident{MkPair} \land \set{ B <: \widetilde{T_1};\; C <: \widetilde{T_2};\; A <: \widetilde{(T_1, T_2)} }$, where $\widetilde{T}$ is the encoded version of $T$.

Intensional type analysis allows performing runtime type analysis on values.
\citet{harper1995} develop $\lambda^{\mathrm{ML}}_i$ to allow efficient universal representations of values in polymorphic contexts.
Their calculus was later generalised by \citet{trifonov2000}, arriving at a \emph{fully reflexive} version: one that can analyse the runtime type of any value.
\calculus{} supports a limited version of runtime type analysis with its tagged objects and $\kCase$ forms.
Unlike true intensional type analysis, we cannot analyze arbitrary runtime type representations,
but on the other hand our system admits efficient implementation strategies
(like the runtime-class-based design it formalizes).

Bindings in DOT calculi, including \calculus, can introduce subtyping assumptions.
This was noticed early in the study of DOT systems, as it leads to the crucial \emph{bad bounds} problem \cite{rompf2016,amin2014}.
\citet{hu2020} frame this idea as \emph{subtyping reflection}: the notion that the context may contain evidence of subtyping, which can be \emph{reflected} into the subtyping judgment.
This name hints at the analogy to \emph{equality reflection} from extensional type theories,\footnote{See, for instance, \citet{nordstrom1990}.}
where the context can contain evidence of equality which can be reflected into the typing judgment.

\begin{acks}
  We thank Paolo Giarrusso for his help with the student project that became part of this work.
  We also thank the anonymous reviewers, whose feedback let us substantially improve this paper.
\end{acks}

\section*{Data-Availability Statement}
The mechanized soundness proofs of \calculus{} and our variant of \Xis{},
as well as some lemmas related to encoding \Xis{} into \calculus{},
are available online \cite{abgCDOTArtifact2022}.

\bibliography{bibliography,autogenerated}

\appendix
\section{Subsuming GADTs}

\subsection{Modifications to \texorpdfstring{\Xis{}}{the GADT system}}

In this section, we cover those details of our modified version of $\lambda_{2, G\mu}$ which are only relevant to the encoding.
The variant is exhaustively described in \cite{RadekThesis}.

\subsubsection*{Type annotations}
We will want to encode \Xis{} into \calculus{}.
In some places this encoding will need to state the types of some objects in the syntax of the target calculus.
To ensure that the encoding can be defined just by looking at the syntax of the input terms, we add the necessary annotations to the syntax of the calculus.
These annotations do not change semantics or expressivity.
In fact, a well-typed program without these annotations can be easily converted into a program with the annotations, just by appending the necessary types coming from the typing derivation.
We actually have mechanized two variants of the calculus: one with these annotations and another without them.

The only modifications are in the tuple syntax, where $\langle e_1, e_2 \rangle$ becomes $\langle e_1 : \tau_1, e_2 : \tau_2 \rangle$ and in the pattern matching syntax, explained in the next paragraph.
This modification does not affect the reduction rules (they can simply ignore the annotations).
In the typing rules, we simply verify that the type in the annotation matches the actual type stemming from the typing judgment.
For example, below we show the old and new versions of the \textbf{ty-tup} rule. The change is purely technical in that the type now not only appears in the judgment, but also at the syntax level.

\begin{prooftree}
  \AxiomC{$\Delta; \Pi \vdash e_1 : \tau_1$}
  \AxiomC{$\Delta; \Pi \vdash e_2 : \tau_2$}
  \RightLabel{(old \textbf{ty-tup})}
  \BinaryInfC{$\Delta; \Pi \vdash \langle e_1, e_2 \rangle : \tau_1 * \tau_2$}
\end{prooftree}
\begin{prooftree}
  \AxiomC{$\Delta; \Pi \vdash e_1 : \tau_1$}
  \AxiomC{$\Delta; \Pi \vdash e_2 : \tau_2$}
  \RightLabel{(new \textbf{ty-tup})}
  \BinaryInfC{$\Delta; \Pi \vdash \langle e_1 : \tau_1, e_2 : \tau_2 \rangle : \tau_1 * \tau_2$}
\end{prooftree}

\subsubsection*{Case forms}

As we mentioned before, our variant disallows nested patterns and requires exhaustive matches.
Without nesting, matching on tuples (which can be unpacked using the \textbf{fst} and \textbf{snd} operations) and unit is not really useful,
so we remove it from our variant.

Since our pattern matching requires all branches to be of the same constructor, we actually know statically which GADT is being matched by a particular match.
To simplify the encoding, we can also include this information at the syntax level.
Just like the type annotations, this can be trivially reconstructed from an unannotated variant coupled with a typing derivation.
The match form that used to be $\operatorname{\mathbf{case}} e \operatorname{\mathbf{of}} ms$
now becomes $\operatorname{\mathbf{matchgadt}} e \operatorname{\mathbf{as}} T \operatorname{\mathbf{returning}} \tau_{\mathit{ret}} \operatorname{\mathbf{with}} ms$,
where $T$ is the name of the matched GADT and $\tau_{\mathit{ret}}$ is the type that is being returned by the whole match.

Because we have removed the nesting of pattern matches and the matching of tuples and unit, the only rule left from the \textbf{pat-} rules is \textbf{pat-cons}.
When removing nesting we actually need to merge it with \textbf{pat-var}, thus we create the following new rule:

\begin{prooftree}
  \AxiomC{$\Sigma(c) = \forall \vec{\alpha}. \tau \to (\vec{\tau_1}) T$}
  \AxiomC{$\Delta_0, \vec{\alpha}, \vec{\tau_1} \equiv \vec{\tau_2} \vdash \tau : *$}
  \RightLabel{\textbf{(pat-cons')}}
  \BinaryInfC{$\Delta_0 \vdash c[\vec{\alpha}](x) \downarrow (\vec{\tau_2}) T \Rightarrow \vec{\alpha}, \vec{\tau_1} \equiv \vec{\tau_2}; x : \tau$}
\end{prooftree}

Due to our syntactic changes, the \textbf{ty-case} rule becomes:

\begin{prooftree}
  \AxiomC{$\Delta; \Gamma \vdash e : (\vec{\sigma}) T$}
  \AxiomC{$\Delta; \Gamma \vdash ms : (\vec{\sigma}) T \Rightarrow \tau_2$}
  \RightLabel{\textbf{(ty-case'')}}
  \BinaryInfC{$\Delta; \Gamma \vdash \operatorname{\mathbf{matchgadt}} e \operatorname{\mathbf{as}} T \operatorname{\mathbf{returning}} \tau_2  \operatorname{\mathbf{with}} ms : \tau_2$}
\end{prooftree}

The rule used in the mechanization of the calculus goes one step further and inlines \textbf{pat-cons'} in \textbf{ty-case'} to create a single big rule.
This helps with the mechanized proof, but would decrease the clarity on paper, so we will use the rules as shown above.
The rule used in the mechanization is however equivalent to this pair of rules.

\subsection{Encoding definition}
In this section, we present the translation from $\lambda_{2, G\mu}$ into \calculus.
The translation was originally presented by \citet{RadekThesis} and is presented here with some minor modifications (mostly notational).

Moreover, we add yet one more modification to the typing relation: we rely on a weaker $\vDash$ notion which is a strict weakening of the original one.
The paper introduced a relation $\vdash$ defined by a set of rules (see Figure 7) and proved that it is equivalent to the original $\vDash$.
We weaken the relation by removing some of the rules.
We will call this new relation $\vdash'$ and use it in place of $\vDash$.
We remove the \textit{ex falso}-like rules and the ones allowing to infer equalities from equality of type lambdas and arrows.
The remaining rules will be enumerated as part of proof of lemma \ref{lemma:eqpreserved}.

For conciseness, we will use some shorthand conventions when working with \calculus{} terms.
For example, instead of writing $\{a\} \wedge \{b\}$ we will sometimes write $\{a; b\}$.
Also to avoid repetition, we will use $\{A = T \}$ as a shorthand for $\{A: T..T\}$.
We will also use a more concise syntax for defining objects: instead of writing $\nu(s: \{ A: T..T; x: U \})[p.B] \{ A = T; x = u \}$ we write:
\begin{flalign*}
& \nu(s)[p.B]\{ A = T ; x : U ; x = u \}
\end{flalign*}

Whenever we use some variables that do not appear in the inputs of the encoding functions, we mean that these are some freshly generated identifiers that do not appear anywhere else in the program. We use the underscore for names which are not referred to anywhere.

Now we can show the encoding of types:

\begin{align*}
\T(\alpha, \Theta) &= \Theta(\alpha) \\
\T(\mathbf{1}, \Theta) &= \texttt{lib.Unit} \\
\T(\tau_1 * \tau_2, \Theta) &= 
\texttt{lib.Tuple} \wedge \{T_1 = \T(\tau_1, \Theta)\} \wedge \{T_1 = \T(\tau_2, \Theta)\} \\
\T(\tau_1 \to \tau_2, \Theta) &= 
\forall \left(\mathit{arg}: \T(\tau_1, \Theta) \right). \; \T(\tau_2, \Theta) \\
\T(\left(\tau_1, \dots, \tau_n \right) T, \Theta) &= 
\texttt{env.T} \wedge \{A_1 = \T(\tau_1, \Theta)\} \wedge \dots \wedge \{A_n = \T(\tau_n, \Theta)\} \\
\T(\forall \alpha. \tau, \Theta) &= 
\forall \left(A: \{T: \bot..\top \} \right). \; \T(\tau, \Theta[\alpha \mapsto A.T]) \\
\end{align*}

Then we introduce the term for encoding tuples and unit (the \texttt{lib}) and the construction of \texttt{env} which encodes the signature $\Sigma$.

\begin{lstlisting}[mathescape=true,label=stdlib, basicstyle=\small\ttfamily]
let lib = $\nu$[lib.Any](lib: {
  Any : $\bot$..$\top$
  Unit = {U: $\bot$..$\top$}
  unit : lib.Unit
  Tuple = $\mu$(s: {$T_1$: $\bot$..$\top$; $T_2$: $\bot$..$\top$; fst: s.$T_1$; snd: s.$T_2$})
  tuple : $\forall$(tl: {$T_1$: $\bot$..$\top$; $T_2$: $\bot$..$\top$})
	  $\forall$(x1: tl.$T_1$) $\forall$(x2: tl.$T_2$)
	    lib.Tuple $\wedge$ {$T_1$ = tl.$T_1$} $\wedge$ {$T_2$ = tl.$T_2$}
}) {
  Any = $\top$
  Unit = {U: $\bot$..$\top$}
  unit = $\nu$[lib.Unit](s: {U=$\top$})
  Tuple = $\mu$(s: {$T_1$: $\bot$..$\top$; $T_2$: $\bot$..$\top$; fst: s.$T_1$; snd: s.$T_2$})
  tuple = $\lambda$(tl: {$T_1$: $\bot$..$\top$; $T_2$: $\bot$..$\top$})
	  $\lambda$(x1: tl.$T_1$) $\lambda$(x2: tl.$T_2$)
	    $\nu$[lib.Tuple](s: {$T_1$ = tl.$T_1$; $T_2$ = tl.$T_2$; fst = x1; snd = x2})
}) in ...
\end{lstlisting}

We assume that the provided $\Sigma$ is described by the concrete syntax as described in Section~2.2 of the paper~\cite{xi2003},
and consists of a list of definitions ($\Sigma = \mathbb{T}_1, \ldots, \mathbb{T}_n$) where $\mathbb{T}_i$ has form the following form:
\begin{lstlisting}[mathescape=true,basicstyle=\small\ttfamily]
(type, ..., type) T = {$\vec{\beta_1}$}. ($\vec{\sigma_1}$) $c_1$ of $\tau_1$
                    | ...
                    | {$\vec{\beta_n}$}. ($\vec{\sigma_n}$) $c_n$ of $\tau_n$
\end{lstlisting}

The \texttt{type} symbol is repeated $m$-times and indicates the cardinality of a given type constructor $T$. Each $c_i$ denotes a name of one of its $n$ constructors. In each $j$th constructor, $\vec{\beta_j}$ denotes the list of its type parameter names, $\tau_j$ indicates the type of the constructor value-level parameter and $\vec{\sigma_j}$ indicates to what type each of the type constructor parameters is instantiated in a given case. Each of $\vec{\sigma_j}$ can reference the parameters $\vec{\beta_j}$.

We also use $\Sigma$ as a mapping of GADT constructor signatures. With a definition like the above, we will have $\Sigma(c_j) = \forall \vec{\beta_j}. \tau_j \to (\vec{\sigma_j}) T$.

For example, the definition of a GADT emulating equality would look like following:

\begin{lstlisting}[mathescape=true,basicstyle=\small\ttfamily]
(type, type) T = {$\beta$}. ($\beta$, $\beta$) refl of $\mathbf{1}$
\end{lstlisting}

The function \texttt{EncodeSigma} that we define below takes a list of type definitions ($\mathbb{T}_i$)
and returns an \texttt{env} object encoding each $\mathbb{T}_i$ as a base type and its associated constructors.

\begin{lstlisting}[mathescape=true,basicstyle=\small\ttfamily]
EncodeSigma($\mathbb{T}_1$, ..., $\mathbb{T}_n$) =
  $\nu$[lib.Any](env: {
    EncodeGADT($\mathbb{T}_1$)
    ...
    EncodeGADT($\mathbb{T}_n$) 
  })
\end{lstlisting}

To make the definitions more concise, we define a scheme for helper substitutions:

\begin{equation}
\label{eqn:ThetaShort}
\theta^{(x)}_{\vec{\beta_i}} = [\beta_{i,1} \mapsto x.B_{i,1}, \; \dots, \; \beta_{i,m_i} \mapsto x.B_{i,m_i}]
\end{equation}

With these tools, we can now define the \texttt{EncodeGADT} function, as shown on \autoref{fig:encodeGadtSig}.

To avoid name clashes we assume that all constructor and type names in $\Sigma$ are unique.

\newcommand\xf{\mathit{xf}}

And finally the encoding of terms.
Note that the $\xf$ ligature is used to denote that this rule applies both to $x$-variables and $f$-variables.

\begin{align*}
  \E(\xf, \; \Theta) &= \Theta(\xf) \\
  \E(c_i[\vec{\tau}](e), \; \Theta) &=
    \begin{aligned}[t]
      & \nLet ts = \nu[\texttt{lib.Any}](ts: \{B_{i,1} = \T(\tau_1, \Theta);\; \ldots;\; B_{i,m_i} = \T(\tau_{m_i}, \Theta)\}) \nIn \\
      & \nLet v = \E(e, \Theta) \nIn \\
      & \quad \texttt{env.}c_i \;\; ts \;\; v \\
    \end{aligned} \\
  \E(\left<\right>, \; \Theta) &= \texttt{lib.unit} \\
  \E(\left<e_1 : \tau_1, e_2 : \tau_2\right>, \; \Theta) &=
    \begin{aligned}[t]
      & \nLet v_1 = \E(e_1, \Theta) \nIn \\
      & \nLet v_2 = \E(e_2, \Theta) \nIn \\
      & \quad \texttt{lib.tuple} \; \nu[\texttt{lib.Any}](\_: \{T_1=\T(\tau_1, \Theta); T_2=\T(\tau_2, \Theta) \}) \; v_1 \; v_2 \\
    \end{aligned} \\
  \E(\mathbf{fst}(e), \; \Theta) &=
    \nLet v = \E(e, \Theta) \nIn v.\texttt{fst} \\
  \E(\mathbf{snd}(e), \; \Theta) &=
    \nLet v = \E(e, \Theta) \nIn v.\texttt{snd} \\
  \E(\lambda x : \tau_1. e, \; \Theta) &=
    \lambda (x_0: \T(\tau_1, \Theta)) \; \E(e, \Theta[x \mapsto x_0]) \\
  \E(e_1 (e_2), \; \Theta) &=
    \begin{aligned}[t]
      & \nLet v_1 = \E(e_1, \Theta) \nIn \\
      & \nLet v_2 = \E(e_2, \Theta) \nIn \\
      & \quad v_1 \; v_2
    \end{aligned} \\
  \E(\Lambda \alpha. e, \; \Theta) &=
    \lambda \left(x_{\alpha}: \{T: \bot..\top\} \right) \; \E(e, \Theta[\alpha \mapsto x_{\alpha}.T]) \\
  \E(e [\tau_1], \; \Theta) &=
    \letp{tl}{\nu[\texttt{lib.Any}](\_: \{ T = \T(\tau_1, \Theta) \})}{ \E(e, \Theta) \; tl } \\
  \E(\textbf{fix } f: \tau. e, \; \Theta) &=
    \begin{aligned}[t]
      & \nLet \mathit{hlpObj} = \nu[\texttt{lib.Any}](\mathit{self}: \{ \\
      & \quad \mathit{fix} = \lambda\left(\_ : \texttt{lib.Unit} \right) \; \E(e, \Theta\left[f \mapsto \mathit{self}.\mathit{fix} \;\; \texttt{lib.unit}\right]) \\
      & \}) \nIn \mathit{hlpObj}.\mathit{fix} \;\; \texttt{lib.unit} \\
    \end{aligned} \\
  \E(\nLet x = e_1 \nIn e_2 \textbf{ end}, \; \Theta) &=
     \letp{x}{\E(e_1, \Theta) }{\E(e_2, \Theta[x \mapsto x])} \\
\end{align*}%
\begin{center}
\begin{align*}
& \E( \textbf{matchgadt } e \textbf{ as } T \textbf{ returning } \tau_{ret} \textbf{ with } &&\\ 
& \qquad\qquad\qquad  | \; c_1[\beta_{1,1}, ..., \beta_{1, m_1}](x_1) \Rightarrow e_1 \; &&\\
& \qquad\qquad\qquad  \ldots &&\\
& \qquad\qquad\qquad  | \; c_n[\beta_{n,1}, ..., \beta_{n, m_n}](x_n) \Rightarrow e_n, \; \Theta) = &&\\
& \quad \nLet tl = \nu[\texttt{lib.Any}](tl: \{ R = \T(\tau_{ret}, \Theta) \}) \nIn &&\\
& \quad \nLet v = \E(e, \Theta) \nIn &&\\
& \quad \nLet \texttt{case}_{c_1} = \lambda\left(arg_1: \texttt{env.T}_{C_1} \wedge v.\text{type} \right) &&\\
& \quad \quad \nLet x'_1 = arg_1\texttt{.data in } &&\\
& \quad \quad \E(e_1, \Theta[x_1 \mapsto x'_1, \beta_{1,1} \mapsto arg_1.B_{1,1}, \; ..., \; \beta_{1,m_1} \mapsto arg_1.B_{1,m_1}]) \\
& \quad \nIn \\
& \quad ... \\
& \quad \nLet \texttt{case}_{c_n} = \lambda\left(arg_n: \texttt{env.T}_{C_n} \wedge v.\text{type} \right) &&\\
& \quad \quad \nLet x'_n = arg_n\texttt{.data in } \\
& \quad \quad \E(e_n, \Theta[x_n \mapsto x'_n, \beta_{n,1} \mapsto arg_n.B_{n,1}, \; ..., \; \beta_{n,m_n} \mapsto arg_n.B_{n,m_n}]) \\
& \quad \nIn \\
& \quad v\texttt{.pmatch} \;\; tl \;\; \texttt{case}_{c_1} \;\; ... \;\; \texttt{case}_{c_n} \\
\end{align*}
\end{center}

\begin{figure}[h]
  \caption{The definition of encoding of a single GADT from the signature $\Sigma$.}
  \label{fig:encodeGadtSig}
\begin{lstlisting}[mathescape=true,basicstyle=\small\ttfamily]
EncodeGADT(
    typecon (type$_1$, ..., type$_m$) T = {$\vec{\beta_1}$}. ($\vec{\sigma_1}$) $c_1$ of $\tau_1$
                                  | ...
                                  | {$\vec{\beta_n}$}. ($\vec{\sigma_n}$) $c_n$ of $\tau_n$)
) =
  T = $\mu$(s: {
    $A_1$: $\bot$..$\top$; ... $A_m$: $\bot$..$\top$
    pmatch: $\forall$(r: {R:$\bot$..$\top$})
            $\forall$(case$_{c_1}$: $\forall$(arg: env.T$_{c_1}$ $\wedge$ s.type) r.R)
            ...
            $\forall$(case$_{c_n}$: $\forall$(arg: env.T$_{c_n}$ $\wedge$ s.type) r.R)
            r.R
  })
  // and for each constructor from $c_1$ to $c_n$
  T$_{c_i}$ = $\mu$(s: env.T $\wedge$ {
    $B_{i,1}$ : $\bot$..$\top$, ..., $B_{i,m_i}$: $\bot$..$\top$
    $A_1$ = $\mathcal{T}$($\sigma_{i,1}$, $\theta^{(s)}_{\vec{\beta_i}}$), ..., $A_m$ = $\mathcal{T}$($\sigma_{i,m}$, $\theta^{(s)}_{\vec{\beta_i}}$)
    data: $\mathcal{T}$($\tau_i$, $\theta^{(s)}_{\vec{\beta_i}}$)
  })
  $c_i$ : $\forall$(ts: {$B_{i,1}$: $\bot$..$\top$; ...; $B_{i,m_i}$: $\bot$..$\top$})
       $\forall$(v: $\mathcal{T}$($\tau_i$, $\theta^{(ts)}_{\vec{\beta_i}}$))
         env.T $\wedge$ {$A_1$ = $\mathcal{T}$($\sigma_{i,1}$, $\theta^{(ts)}_{\vec{\beta_i}}$); ...; $A_m$ = $\mathcal{T}$($\sigma_{i,m}$, $\theta^{(ts)}_{\vec{\beta_i}}$)}
  $c_i$ = $\lambda$(ts: {$B_{i,1}$: $\bot$..$\top$; ...; $B_{i,m_i}$: $\bot$..$\top$}) $\lambda$(v: $\mathcal{T}$($\tau_i$, $\theta^{(ts)}_{\vec{\beta_i}}$))
    let s = $\nu$[env.T$_{c_i}$](s: {
      $B_{i,1}$ = ts.$B_{i,1}$, ..., $B_{i,m_i}$ = ts.$B_{i,m_i}$
      $A_1$ = $\mathcal{T}$($\sigma_{i,1}$, $\theta^{(s)}_{\vec{\beta_i}}$), ..., $A_m$ = $\mathcal{T}$($\sigma_{i,m}$, $\theta^{(s)}_{\vec{\beta_i}}$)
      data = v
      pmatch = $\lambda$(r: {R:$\bot$..$\top$})
               $\lambda$(case$_{c_1}$: $\forall$(arg: env.T$_{c_1}$ $\wedge$ s.type) r.R)
               ...
               $\lambda$(case$_{c_n}$: $\forall$(arg: env.T$_{c_n}$ $\wedge$ s.type) r.R)
                 let h=$\nu$[lib.Any](_: {z=s}) in case$_{c_i}$ h.z
    }) in s
\end{lstlisting}
\end{figure}
\subsection{Correctness proof for the GADT encoding}

Our proof will rely on many simple properties of \calculus{} derived from \calculus{} and of the $=:=$ relation. While they are rather trivial, a more thorough discussion of them can be found in section~4.5 of \cite{RadekThesis}. One notable example that we will use quite a lot is the following lemma:
\begin{lemma}[Lemma 6 from \cite{RadekThesis}]
  \label{lemma:6}
  For any types $A$, $B$, $C$ and $D$, if $\Gamma \vdash A =:= B$ and $\Gamma \vdash C =:= D$, then $\Gamma \vdash A \wedge C =:= B \wedge D$.
\end{lemma}
\begin{proof}
  By \textsc{Trans} and \textsc{And} subtyping rules\footnote{A mechanized proof can be found in the attached artifacts, in \texttt{translation/Helpers.v} under the name \texttt{eq\_and\_merge}.}.
\end{proof}

\begin{lemma}
  \label{lemma:sigmaPrime}
  The term for \texttt{lib} from the encoding and the ones generated by \texttt{EncodeGADT} are well-typed and satisfy the requirements stated in the \autoref{def:sigEnc}.
\end{lemma}
\begin{proof}
The typing proof of \textit{lib} has been mechanized\footnote{See \texttt{translation/Library.v} within attached artifacts.} and quite simple so we skip it here and refer to the mechanization for details.

\newcommand\thetaShorthand[1]{\theta^{(#1)}_{\vec{\beta_i}}}
\newcommand\thetaS{\thetaShorthand{s}}
\newcommand\thetaTS{\thetaShorthand{ts}}

The \texttt{env} contains necessary definitions to encode each GADT from $\Sigma$.
Let's consider a single GADT definition and show that the term returned by \texttt{EncodeGADT} admits the expected type.
We start by \textsc{\{\}-I}, then the type definitions are trivially checked by the \textsc{Def-Typ} rule.
The non-trivial part are the constructors $c_i$.
We start with the \textsc{Def-All} rule, followed by two rather standard applications of \textsc{All-I} rule.

Now, we need to show that the new object created in the body of the lambda is itself well-typed and satisfies the expected type.
As always, we start with \textsc{\{\}-I} and will assign to the result the type that directly stems from the definition --- we will generalize it to the desired result type and to show that it fits the ascribed tag a few lines later.
This ``primitive'' result type will be\footnote{The symbol $\thetaS$ is defined as before in \autoref{eqn:ThetaShort}.}:
\[
\mu(s: \{ B_{i,1} = ts.B_{i,1} \} \wedge \ldots \wedge \{A_1 = \T(\sigma_{i,1}, \thetaS )\} \wedge ... \wedge \{\data: v.type\} \wedge \{ \pmatch : ... \})
\]
(we omit the $\pmatch$ type as it is very long but it is exactly the same as inside of $\env.T$).

The types can be checked trivially with \textsc{Def-Typ}, the \texttt{data} field gets ascribed type $v.type$ by \textsc{Def-Path}. 
Finally, the hardest part is the $\pmatch$ lambda.
We use \textsc{Def-All} followed by \textsc{All-I}, in a rather standard way.
Inside the body of the lambda we construct a value $h$ of type $\{ z : s.type \}$
(checked by \textsc{\{\}-I}, \textsc{Def-Path} and \textsc{Rec-E}) and use it to apply $h.z$ to $\texttt{case}_{c_i}$.
To typecheck this application we need to show that $h.z$ fits the type $\env.T_{c_i} \wedge s.type$.
We have $h.z : s.type$ trivially by \textsc{Fld-E}.
We proceed to show $h.z : \env.T_{c_i}$.
We have $h.z : s.type$, and by \textsc{Sngl-Trans} so now all it remains to be shown is that $s : \env.T_{c_i}$.
But from the assumption of \textsc{\{\}-I} we know that $s$ admits the ``primitive'' type described earlier.
We use \textsc{Rec-E}, then split the big intersection type into parts using \textsc{Sub} with \textsc{And\textsubscript{1}-<:} and \textsc{And\textsubscript{2}-<:}.
After we generalize each type to fit $\env.T_{c_i}$ (shown in the next paragraph), we can rebuild the type back using \textsc{\&-I} and `close' it with \textsc{Rec-I}.

We generalize the type of \texttt{data} using \textsc{Fld-E}, \textsc{Sngl-Trans} again with $v.type$ and then get back with \textsc{Fld-I}.
$v$ itself has type $\T(\tau_i, \thetaTS)$ and we need $\T(\tau_i, \thetaS)$.
We can get from $\thetaTS$ to $\thetaS$, because they differ only in mapping each $\beta_{i,j}$ to $ts.B_{i,j}$ and $s.B_{i,j}$ respectively.
But by definition of $s$, $s.B_{i,j} =:= ts.B_{i,j}$, so we can simply perform the replacement by \autoref{lemma:equalSubstChange}.

We generalize the types of $B_{i,j}$ from $\{B_{i,j} = ts.B_{i,j} \}$ to $\{ B_{i,j}: \bot..\top \}$ by \textsc{Sub} with \textsc{Typ-<:-Typ}.

The types $A_j$ stay as-is. This makes us fit the requirements of $\env.T_{c_i}$ apart from one --- its first requirement is for the overall type to also fit $\env.T$.
But that is easy --- the $\pmatch$ function already has the right type and we simply generalize types $A_j$ in the same way as we did with $B_{i,j}$.
Then we collect the parts with \textsc{\&-I} and \textsc{Rec-I}, in a standard way.

All this allows us to derive both $s : \env.T$ and finally $s : \env.T_{c_i}$, thus giving us $h.z : \env.T_{c_i} \wedge s.type$.
Thus the application $\texttt{case}_{c_i}\; h.z$ is valid --- its result is $r.R$, like we want.

We have shown that the constructor constructs a valid object, with the primitive type described earlier.
All that remains is to generalize it.
Firstly, we need to remember that the newly created $\nu(\ldots)$ object is bound to a path thanks to the expression $\letp{s}{\nu[\env.T_{c_i}](\ldots)}{s}$.
We have that $s : \mu(s: \ldots)$ where $\ldots$ is the primitive type mentioned before.
We now need to show that $s : \T((\sigma_{i,1}, \ldots, \sigma_{i,m}) T, \thetaTS)$ and so we will be able to finish with \textsc{Let}.
First, we can unfold $\T$ a bit, so that what we actually have to show is $\env.T \wedge \{ A_1 = \T(\sigma_{i,1}, \thetaTS; ...; A_m = \T(\sigma_{i,m}, \thetaTS) \}$.
We use exactly the same reasoning that we did inside of \textsc{\&-I} to get that $s : \env.T$ and $s : \env.T_{c_i}$ to satisfy the additional requirement of \textsc{\{\}-I}.
Now we just need to show $s : \{ A_j = \T(\sigma_{i,j}, \thetaTS) \}$.
But the primitive type already gives us $s : \{A_j = \T(\sigma_{i,j}, \theta{s})\}$ (after we unpack it with \textsc{Rec-E} etc., obviously).
Now, again we can easily show that the two types are equal by \autoref{lemma:equalSubstChange}.
Again, $\thetaTS$ and $\thetaS$ map each $\beta_{i,k}$ to $ts.B_{i,k}$ and $s.B_{i,k}$ respectively,
and from the primitive type we also have $s : \{ B_{i,k} = ts.B_{i,k} \}$, giving us $s.B_{i,k} =:= ts.B_{i,k}$ by \textsc{<:-Sel} and \textsc{Sel-<:}.
\end{proof}

Before we can prove the preservation of types, we need to show that we can `emulate' all equalities of our modified entailment relation $\vdash'$ from $\lambda_{2, G\mu}$ in \calculus{}.
To do so we prove the following lemma:

\eqEmulation*
\begin{proof}
  The proof proceeds by induction on the derivation of the judgment $\vdash'$.
  We inspect each case.

  \begin{prooftree}
    \AxiomC{$\vec{\alpha} \vdash' \tau : *$}
    \UnaryInfC{$\vec{\alpha} \vdash' \tau \equiv \tau$}
  \end{prooftree}

  \noindent This case is trivially handled by the \textsc{Refl} rule.

  \begin{prooftree}
    \AxiomC{$\vec{\alpha}, \Delta \vdash' \tau_1 \equiv \tau_2$}
    \UnaryInfC{$\vec{\alpha}, \alpha \equiv \alpha, \Delta \vdash' \tau_1 \equiv \tau_2$}
  \end{prooftree}

  \noindent
  We have some $\Gamma$ and $\Theta$ where $\Psi(\vec{\alpha}, \alpha \equiv \alpha, \Delta; \Theta; \Gamma)$ holds and want to prove that $\Gamma \vdash \T(\tau_1, \Theta) =:= \T(\tau_2, \Theta)$ holds.
  We can just use the inductive hypothesis: we just need to show that $\Psi(\vec{\alpha}, \Delta; \Theta; \Gamma)$ holds.
  But that is a weaker requirement\footnote{Technically it is equivalent because $\alpha \equiv \alpha$ does not introduce any new information --- it holds in every environment.},
  so it holds trivially.

  \begin{prooftree}
    \AxiomC{$\vec{\alpha}, \Delta[\alpha \mapsto \tau] \vdash' \tau_1[\alpha \mapsto \tau] \equiv \tau_2[\alpha \mapsto \tau]$}
    \noLine
    \UnaryInfC{$\alpha \text{ has no free occurrences in } \tau$}
    \UnaryInfC{$\vec{\alpha}, \tau \equiv \alpha, \Delta \vdash' \tau_1 \equiv \tau_2$}
  \end{prooftree}
  \begin{prooftree}
    \AxiomC{$\vec{\alpha}, \Delta[\alpha \mapsto \tau] \vdash' \tau_1[\alpha \mapsto \tau] \equiv \tau_2[\alpha \mapsto \tau]$}
    \noLine
    \UnaryInfC{$\alpha \text{ has no free occurrences in } \tau$}
    \UnaryInfC{$\vec{\alpha}, \alpha \equiv \tau, \Delta \vdash' \tau_1 \equiv \tau_2$}
  \end{prooftree}

  The two rules are symmetric so we will describe the proof for the second one, the first one is completely analogous.

  We want to show that $\Gamma \vdash \T(\tau_1, \Theta) =:= \T(\tau_2, \Theta)$.
  We know that $\Psi(\vec{\alpha}, \alpha \equiv \tau, \Delta; \Theta; \Gamma)$.
  This gives us $\Gamma \vdash \T(\alpha, \Theta) =:= \T(\tau, \Theta)$, equivalently $\Gamma \vdash \Theta(\alpha) =:= \T(\tau, \Theta)$.

  From the inductive hypothesis we know that $\Gamma \vdash \T(\tau_1[\alpha \mapsto \tau], \Theta) =:= \T(\tau_2[\alpha \mapsto \tau], \Theta)$
  if only we can show that $\Psi(\vec{\alpha}, \Delta[\alpha \mapsto \tau]; \Theta; \Gamma)$.
  From the assumption $\Psi(\vec{\alpha}, \alpha \equiv \tau, \Delta; \Theta; \Gamma)$ we have that for any $\tau'_1 \equiv \tau'_2 \in \Delta$,
  we have $\Gamma \vdash \T(\tau'_1, \Theta) =:= \T(\tau'_2, \Theta)$.
  Now, we need to show the same thing for any $\tau''_1 \equiv \tau''_2 \in \Delta[\alpha \mapsto \tau]$.
  But for any $\tau''_1 \equiv \tau''_2 \in \Delta[\alpha \mapsto \tau]$,
  we know that there are some $\tau'_1 \equiv \tau'_2 \in \Delta$ such that $\tau''_1 = \tau'_1[\alpha \mapsto \tau]$ and $\tau''_2 = \tau'_2[\alpha \mapsto \tau]$.
  Thus we know that $\Gamma \vdash \T(\tau'_1, \Theta) =:= \T(\tau'_2, \Theta)$
  and just need to show $\Gamma \vdash \T(\tau'_1[\alpha \mapsto \tau], \Theta) =:= \T(\tau'_2[\alpha \mapsto \tau], \Theta)$.
  We will get $\Gamma \vdash \T(\tau'_1[\alpha \mapsto \tau], \Theta) =:= \T(\tau'_1, \Theta)$ by \autoref{lemma:substEqual}
  and analogously for $\tau_2$ and get the desired conclusion by transitivity of $=:=$.

  We have proven $\Psi(\vec{\alpha}, \Delta[\alpha \mapsto \tau]; \Theta; \Gamma)$,
  which gives us $\Gamma \vdash \T(\tau_1[\alpha \mapsto \tau], \Theta) =:= \T(\tau_2[\alpha \mapsto \tau], \Theta)$.
  Again, in the same way we use \autoref{lemma:substEqual} to show that $\Gamma \vdash \T(\tau_1[\alpha \mapsto \tau], \Theta) =:= \T(\tau_1, \Theta)$
  and also analogously for $\tau_2$ --- then the conclusion follows trivially from transitivity of $=:=$.

  In the last rule we handle the $T$ type constructor could stand for the GADT type constructors but also tuple and arrow types.
  However, due to our restriction the arrow rule is removed, so we stay with the following two rules: one for GADTs and one for tuples.

  \begin{prooftree}
      \AxiomC{$\vec{\alpha}, \vec{\tau_A} \equiv \vec{\tau_B} \vdash' \tau_1 \equiv \tau_2$}
      \UnaryInfC{$\vec{\alpha}, (\vec{\tau_A}) T \equiv (\vec{\tau_B}) T \vdash' \tau_1 \equiv \tau_2$}
    \end{prooftree}
  \begin{prooftree}
      \AxiomC{$\vec{\alpha}, \tau_{A_1} \equiv \tau_{B_1}, \tau_{A_2} \equiv \tau_{B_2} \vdash' \tau_1 \equiv \tau_2$}
      \UnaryInfC{$\vec{\alpha}, \tau_{A_1} * \tau_{A_2} \equiv \tau_{B_1} * \tau_{B_2} \vdash' \tau_1 \equiv \tau_2$}
  \end{prooftree}

  Let's analyze the tuple example first.

  To get the desired equation we can just apply the inductive hypothesis,
  so it suffices that we show $\Psi(\vec{\alpha}, \tau_{A_1} \equiv \tau_{B_1}, \tau_{A_2} \equiv \tau_{B_2}, \Theta, \Gamma)$.

  From our assumptions, we have $\Psi(\vec{\alpha}, \tau_{A_1} * \tau_{A_2} \equiv \tau_{B_1} * \tau_{B_2}; \Theta; \Gamma)$.
  Unpacking the definition of $\Psi$, we know that $\Gamma \vdash \T(\tau_{A_1} * \tau_{A_2}, \Theta) =:= \T(\tau_{B_1} * \tau_{B_2}, \Theta)$.
  Let's introduce shorthands for the tuple types:
  \[
    U_1 \triangleq \T(\tau_{A_1} * \tau_{A_2}, \Theta)
  \]
    \[
    U_2 \triangleq \T(\tau_{B_1} * \tau_{B_2}, \Theta)
  \]
  Unfolding $\T$, we get
  \[ U_1 = \texttt{lib.Tuple} \wedge \{ T_1 = \T(\tau_{A_1}, \Theta) \} \wedge \{ T_2 = \T(\tau_{A_2}, \Theta) \} \]
  \[ U_2 = \texttt{lib.Tuple} \wedge \{ T_1 = \T(\tau_{B_1}, \Theta) \} \wedge \{ T_2 = \T(\tau_{B_2}, \Theta) \} \]
  We will want to use \textsc{Typ-$<:$-Typ-Inv$_1$} and \textsc{Typ-$<:$-Typ-Inv$_2$} to get $\Gamma \vdash \T(\tau_{A_1}, \Theta) =:= \T(\tau_{B_1}, \Theta)$.
  However, we cannot do this straight away, as the $\texttt{lib.Tuple}$ is not acceptable in the $\searrow$ judgment.
  Instead, let's use the judgment $\Gamma \vdash \texttt{lib.Tuple} =:= \mathit{TupleDef}$
  where $\mathit{TupleDef}$ stands for the tuple type definition $\mu(\ldots)$ found in \texttt{lib}.
  By \autoref{lemma:6}, we can replace the reference with the expanded definition to get
  $U'_1 = \mathit{TupleDef}\wedge \{ T_1 = \T(\tau_{B_1}, \Theta) \} \wedge \{ T_2 = \T(\tau_{B_2}, \Theta) \}$
  with an equation $\Gamma \vdash U'_1 =:= U_1$, and analogously for $U'_2$.
  Now it suffices to show $\Gamma \vdash U'_1 =:= U'_2$ and through the helper equations and transitivity we will get the desired equation.
  Because $\mathit{TupleDef} = \mu(\ldots)$,
  we can now derive $U'_1 \searrow \{ T_1 = \T(\tau_{A_1}, \Theta) \}$ and analogously $U'_2 \searrow \{ T_1 = \T(\tau_{B_1}, \Theta) \}$,
  thus using the inversion rules we get $\Gamma \vdash \T(\tau_{A_1}, \Theta) =:= \T(\tau_{B_1}, \Theta)$.
  We have shown that indeed $\Gamma$ with $\Theta$ satisfies the equation $\tau_{A_1} \equiv \tau_{B_1}$\footnote{A mechanized version of this part of the proof (inverting the tuple equality) is attached with the paper. See \texttt{translation/DestructTupleLemma.v}.}.
  We show $\tau_{A_2} \equiv \tau_{B_2}$ completely analogously.

  The case for GADTs is also analogous --- we have an equality
  $\env.T \wedge \{A1 = \T(\tau_{A_1}, \Theta) \} \wedge \ldots =:= \env.T \wedge \{A1 = \T(\tau_{B_1}, \Theta) \} \wedge \ldots$
  and using the same technique we invert it to get $\T(\tau_{A_1}, \Theta) =:= \T(\tau_{B_1}, \Theta)$
  and the analogous for all other pairs of corresponding type parameters.
  This way, we show $\Psi(\vec{\alpha}, \vec{\tau_A} \equiv \vec{\tau_B}; \Theta; \Gamma)$, which (as above) gives us the desired equation by the inductive hypothesis.
\end{proof}

\encTypPres*
\begin{proof}
  We prove the theorem by induction on the derivation of the typing judgment.
  
  \begin{itemize}
    \item (\textbf{ty-unit}) We get $\Gamma \vdash \lib.\texttt{unit} : \lib.\texttt{Unit}$ by weakening and from the trivial conclusion that $\Gamma \vdash \lib : \{ \texttt{unit} : \lib.\texttt{Unit} \}$ because $\Gamma$ is defined to contain the bindings of $\mathcal{S}(\Sigma)$.
    
    \item (\textbf{ty-var}) From the assumption we know that $\Pi(\mathit{xf}) = \tau$, since $\E(\mathit{xf}, \Theta) = \Theta(\mathit{xf})$, we get the result from $\Phi(\Delta; \Pi; \Theta; \Gamma)$.
    
    \item (\textbf{ty-tup}) Let's introduce aliases for shortening the notation: 
    \[
    \mathit{tupleApp} \triangleq 	 
        \texttt{lib.tuple} \;
          \nu(\{T_1=\T(\tau_1, \Theta); T_2=\T(\tau_2, \Theta)\})
          \; x_1 \; x_2
    \]
    \[
    \mathit{tupleTyp} \triangleq
        \texttt{lib.Tuple} \wedge \{T_1=\T(\tau_1, \Theta) \} \wedge \{ T_2=\T(\tau_2, \Theta)\}
    \]

    We want to prove that 
    \[
    \Gamma \vdash 
	\letp{x_1}{\E(e_1, \Theta)}{
      \letp{x_2}{\E(e_2, \Theta)}{
       \mathit{tupleApp}
      }
    }
    : \mathit{tupleTyp}
    \]
  
   From the inductive hypotheses we have $\Gamma \vdash \E(e_1, \Theta) : \T(\tau_1, \Theta)$ and analogously for $e_2$ with $\tau_2$. Let's define
   \[ \Gamma' \triangleq \Gamma, x_1 : \T(\tau_1, \Theta), x_2 : \T(\tau_2, \Theta) \]
    We can use the $\textsc{Let}$ rule to get to $\Gamma' \vdash \mathit{tupleApp} : \mathit{tupleTyp}$.
  There we can just use the standard application rules and the definition of $\texttt{lib.tuple}$.
  
  \item (\textbf{ty-fst}) We need to prove that $\Gamma \vdash \letp{x}{\E(e, \Theta)}{x.\mathit{fst}} : \T(\tau_1, \Theta)$.
   From the inductive hypothesis, we have $\Gamma \vdash \E(e, \Theta) : \texttt{lib.Tuple} \wedge \{T_1=\T(\tau_1, \Theta) \} \wedge \{ T_2=\T(\tau_2, \Theta)\}$, we apply the \textsc{Let} rule and then we can derive\footnote{
     In short, from $x: \texttt{lib.Tuple}$ we get $x: \{ \mathit{fst} : x.T_1 \}$ and then from $x : \{T_1=\T(\tau_1, \Theta) \}$ we get $x.T_1 <: \T(\tau_1, \Theta)$ allowing us to derive the conclusion.
   
 } $\Gamma, x: \ldots \vdash x : \{ \mathit{fst} : \T(\tau_1, \Theta) \} $ to get the desired result.
  
  \item (\textbf{ty-snd}) Analogously as above.
  
  \item (\textbf{ty-lam}) Let's allocate a fresh $x_0 \notin \text{dom}(\Gamma)$. We can easily show that since $\Phi(\Delta; \Pi; \Theta; \Gamma)$, then also $\Phi(\Delta; \Pi, x : \tau_1; \Theta[x \mapsto x_0]; \Gamma, x_0 : \T(\tau_1, \Theta))$. This allows us to get from the inductive hypothesis that 
  \[
  \Gamma, x_0 : \T(\tau_1, \Theta) \vdash \E(e, \Theta[x \mapsto x_0]) : \T(\tau_2, (\Theta[x \mapsto x_0]))
  \]
  Then, with the \textsc{All-I} rule we prove that 
  \[
  \Gamma \vdash \lambda(x_0: \T(\tau_1, \Theta)) \, \E(e, \Theta[x \mapsto x_0]) : \forall(x_0 : \T(\tau_1, \Theta)) \, \T(\tau_2, \Theta)
  \]
  
  \item (\textbf{ty-app}) From the inductive hypotheses we have (we pick a fresh $x_0$ and $\alpha$-rename the type properly):
  \[
  \Gamma \vdash \E(e_1, \Theta) : \forall(x_0 : \T(\tau_1, \Theta)) \T(\tau_2, \Theta)
  \]
  and
  \[
  \Gamma \vdash \E(e_2, \Theta) : \T(\tau_1, \Theta)
  \]
  We can then use the \textsc{Let} rule twice to get from the goal of 
  \[
  \Gamma \vdash 
  \letp{x_1}{\E(e_1, \Theta)}{
    \letp{x_2}{\E(e_2, \Theta)}{
      x_1 \; x_2
    }
  } : \T(\tau_2, \Theta)
  \]
  to
  \[
  \Gamma, x_1 : \forall(x_0 : \T(\tau_1, \Theta)) \T(\tau_2, \Theta), x_2 : \T(\tau_1, \Theta) \vdash 
  x_1 \; x_2 : \T(\tau_2, \Theta)
  \]
  and finish with \textsc{All-E} by noticing that $\T(\tau_2, \Theta)[x_2/x_0] = \T(\tau_2, \Theta)$, because from construction $x_0 \notin \text{fv}(\T(\tau_2, \Theta))$.

  \item (\textbf{ty-tlam}) We need to show 
  \[
  \Gamma \vdash \E(\Lambda \alpha. e, \Theta) : \T(\forall \alpha. \tau, \Theta)
  \]
  Unfolding the encoding (assuming that $x_{\alpha}$ is some fresh name) that is:
  \[
  \Gamma \vdash \lambda(x_{\alpha}: \{T: \bot..\top\}) \, \E(e, \Theta[\alpha \mapsto x_{\alpha}.T]) : \forall(x_{\alpha}: \{T: \bot..\top\}) \, \T(\tau, \Theta[\alpha \mapsto x_{\alpha}.T])
  \]
  Since we have $\Phi(\Delta; \Pi; \Theta; \Gamma)$ then clearly we also have
  \[
  \Phi(\Delta, \alpha; \Pi; \Theta[\alpha \mapsto x_{\alpha}.T]; \Gamma, x_{\alpha} : \{T: \bot..\top\})
  \]
  Thus from the inductive hypothesis, we get
  \[
  \Gamma, x_{\alpha} : \{T: \bot..\top\} \vdash \E(e, \Theta[\alpha \mapsto x_{\alpha}.T]) : \T(\tau, \Theta[\alpha \mapsto x_{\alpha}.T])
  \]
  
  Clearly, $x_{\alpha} \notin \text{fv}(\{T: \bot..\top\})$. 
  Thus we can use the \textsc{All-I} rule to get what we needed:
  \[
  \Gamma \vdash \lambda(x_{\alpha}: \{T: \bot..\top\}) \, \E(e, \Theta[\alpha \mapsto x_{\alpha}.T]) : \forall(x_{\alpha}: \{T: \bot..\top\}) \, \T(\tau, \Theta[\alpha \mapsto x_{\alpha}.T])
  \]
  
  \item (\textbf{ty-tapp}) We want to show that
   \[
   \Gamma \vdash \E(e[\tau_1], \Theta) : \T(\tau[\alpha \mapsto \tau_1], \Theta)
   \]
   
    Unfolding the $\E$, that is 
  \[
  \Gamma \vdash 
  \letp{x_1}{\E(e, \Theta)}{
     \letp{tl}{\nu(\_: \{T = \T(\tau_1, \Theta)\})}{
	x_1 \; tl
     }
  } : \T(\tau[\alpha \mapsto \tau_1], \Theta)
 \]
 From the inductive hypothesis we know that, for some fresh $x_{\alpha}$:
 \[
 \Gamma \vdash \E(e, \Theta) : \forall(x_{\alpha} : \{T:\bot..\top\})  \T(\tau, \Theta[\alpha \mapsto x_{\alpha}.T])
 \]
 Thus by applying the \textsc{Let} rule twice\footnote{We also use \textsc{\{\}-I} and a trivial instantiation \textsc{Rec-E} to typecheck $tl$, but it's completely standard.}. Let's define
 \[
 \Gamma' \triangleq \Gamma, x_1 : \forall(x_{\alpha} : \{T:\bot..\top\})  \T(\tau, \Theta[\alpha \mapsto x_{\alpha}.T]), tl : \{T = \T(\tau_1, \Theta)\}
 \]
 We can then transform our goal into
 \[
 \Gamma' \vdash x_1 \; tl : \T(\tau[\alpha \mapsto \tau_1], \Theta)
 \]
 We can use \textsc{Sub} with \textsc{Typ-<:-Typ} to get $\Gamma' \vdash tl : \{T: \bot..\top\}$. Thus, we can use the \textsc{All-E} rule to get 
 \[
 \Gamma' \vdash x_1 \; tl : \T(\tau, \Theta[\alpha \mapsto x_{\alpha}.T])[tl/x_{\alpha}]
 \]
 Now, knowing that $\Gamma' \vdash tl.T =:= \T(\tau_1, \Theta)$, we need to use \autoref{lemma:substTranslation} (proven just below) to get 
 \[
 \Gamma' \vdash \T(\tau, \Theta[\alpha \mapsto x_{\alpha}.T])[tl/x_{\alpha}] =:= \T(\tau[\alpha \mapsto \tau_1], \Theta)
 \]
 With that we can conclude by using the \textsc{Sub} rule and replacing the type we have with the desired type.
 
  \item (\textbf{ty-fix}) We want to prove 
  \begin{align*}
  \Gamma \vdash \, & \letp{h}{
     \nu[lib.Any](\mathit{self}: \{\mathit{fix} = \lambda(x: \texttt{lib.Unit}) \; \E(e, \Theta[f \mapsto \mathit{self}.\mathit{fix} \; \texttt{lib.unit}]) \})
  }{\\ & \quad
    h.\mathit{fix} \; \texttt{lib.unit}
  }
  : \T(\tau, \Theta)
  \end{align*}

  Let's define
  \[
  \Gamma' \triangleq \Gamma, \mathit{self}: \mu(\mathit{self}: \{\mathit{fix}: \forall(x: \texttt{lib.Unit}) \T(\tau, \Theta) \})
  \]
  First, we show that $\Phi(\Delta; \Pi, f: \tau; \Gamma'; \Theta[f \mapsto \mathit{self}.\mathit{fix} \; \texttt{lib.unit}])$. This goes easily by \textsc{Rec-E}, \textsc{Fld-E} and \textsc{All-E}.
  
  This allows us to use the inductive hypothesis to get 
  \[
  \Gamma' \vdash \E(e, \Theta[f \mapsto \mathit{self}.\mathit{fix} \; \texttt{lib.unit}]) : \T(\tau, \Theta)
  \]
  This allows us to typecheck the $\nu$ in $h$ as $\mu(\mathit{self}: \{\mathit{fix}: \forall(x: \texttt{lib.Unit}) \T(\tau, \Theta) \})$ (we also need to show that $\mathit{self}$ types to \texttt{lib.Any}, but that is trivial since $\texttt{lib.Any} <: \top$).
  
  Then we proceed with the \textsc{Let} rule followed by \textsc{Rec-E}, which transforms our $h$ to have type $\{\mathit{fix}: \forall(x: \texttt{lib.Unit}) \T(\tau, \Theta) \}$ and then we can easily finish by \textsc{All-E}.
  
  \item (\textbf{ty-let}) Follows directly by the \textsc{Let} rule. Clearly $\Theta_x(x) \notin \text{fv}(\T(\tau_1, \Theta))$ as $\lambda_{2, G\mu}$ does not allow variable references in types.
  
  \item (\textbf{ty-eq}) From the inductive hypothesis we have $\Gamma \vdash \E(e, \Theta) : \T(\tau_1, \Theta)$. We also know that $\Delta \vDash \tau_1 \equiv \tau_2$, so by \autoref{lemma:eqpreserved}, we have $\Gamma \vdash \T(\tau_1, \Theta) =:= \T(\tau_2, \Theta)$, thus by \textsc{Sub} we get the desired result.
  
  \item (\textbf{ty-cons}) 
  To avoid getting lost in identifier names, let's rewrite the rule to the following:
  
  \begin{prooftree}
    \AxiomC{$\Sigma(c_i) = \forall \vec{\alpha}. \tau_A \to \tau_B $}
    \AxiomC{$\Delta \vdash \vec{\tau} : \vec{*}$}
    \AxiomC{$\Delta; \Pi \vdash e : \tau_A [ \vec{\alpha} \mapsto \vec{\tau}]$}
    \RightLabel{\textbf{(ty-cons)}}
    \TrinaryInfC{$\Delta; \Pi \vdash c_i[\vec{\tau}](e) : \tau_B [ \vec{\alpha} \mapsto \vec{\tau}]$}
  \end{prooftree}

  and let $\vec{\tau} \triangleq (\tau_1, ..., \tau_m)$, $\vec{\alpha} \triangleq (\alpha_1, ..., \alpha_m)$.
  
  We need to prove that 
  \begin{align*}
  \Gamma \vdash
  \, & \letp{ts}{\nu(ts: \{ B_{i,1} = \T(\tau_1, \Theta); ... \})
  }{
   \\ & 
    \letp{x}{\E(e, \Theta)}{
    \\ & 
      \letp{tmp}{\env.c_i \; ts}{tmp \; x}
    }
  } \\ & 
  : \T(\tau_B[\vec{\alpha} \mapsto \vec{\tau}], \Theta)
  \end{align*}
  
  From the assumption about $\Sigma(c_i)$ we know that $\env.c_i$ has the type:
  \[
  \forall(ts: \{ B_{i,1} = \T(\tau_1, \Theta); ... \}) \; \forall(v: \T(\tau_A, \theta^{(ts)}_{\vec{\alpha}})) \; \T(\tau_B, \theta^{(ts)}_{\vec{\alpha}})
  \]
   where $\theta^{(ts)}_{\vec{\alpha}}$ is defined as before.
  
  First we prove the typing of $ts$ by applying \textsc{\{\}-I} and then numerously applying \textsc{Def-Typ}.
  From the inductive hypothesis we get $\Gamma \vdash \E(e, \Theta) : \T(\tau_A[\vec{\alpha} \mapsto \vec{\tau}])$, so also $x$ will have that type.

  Now we will use \textsc{Let} and \textsc{All-E} twice to finish the derivation.
  
  First, we can type $tmp$ as $(\forall(v: \T(\tau_A, \theta^{(ts)}_{\vec{\alpha}})) \; \T(\tau_B, \theta^{(ts)}_{\vec{\alpha}}) )[ts \mapsto ts]$. %
  
  Because $\tau_A$ and $\tau_B$ should not have any free variables apart from $\vec{\alpha}$, we can type $tmp$ as $(\forall(v: \T(\tau_A, \Theta \cup \theta^{(ts)}_{\vec{\alpha}})) \; \T(\tau_B, \Theta \cup \theta^{(ts)}_{\vec{\alpha}}) )[ts \mapsto ts]$.
  
  Let $\Gamma'$ be the context $\Gamma$ extended with $ts$, $x$ and $tmp$ with types as shown above. We have $\Gamma' \vdash ts.{B_{i,j}} =:= \T(\tau_j, \Theta)$. Thus, we can use the \autoref{lemma:substTranslation} repeatedly to get
  \[
  \Gamma' \vdash \T(\tau_A, \Theta \cup \theta^{(ts)}_{\vec{\alpha}})[ts \mapsto ts] =:= \T(\tau_A[\vec{\alpha} \mapsto \vec{\tau}], \Theta)
  \]
  and the same for $\tau_B$. Thus, using the \textsc{Sub} rule, we show that indeed the argument $x$ fits the expected type in $tmp$ and the return type of the application $tmp \; x$ can indeed be coerced to what we wanted in the first place.
  
  \item (\textbf{ty-case}) We need to typecheck the following term:
  
  \begin{align*}
  & \letp{tl}{\nu(tl: \{R = \T(\tau_2, \Theta) \})}{ \\
  & \letp{v}{\E(e, \Theta) }{ \\ 
  & \letp{case_{c_1}}{
	(
	\lambda(arg_1: \env.T_{c_1} \wedge v.type) \\
	& \quad\quad\quad \letp{x_1}{arg_1.\data}{ \\
	& \quad\quad\quad  \E(e_1, \Theta[\beta_{1,1} \mapsto arg.B_{1,1}, ...])
	} \\
  & )
    }{ \\      
    & ...\\
    & v.\pmatch{} \; tl \; case_{c_1} \; ... \; case_{c_n}
    } 
  }
  }
  \end{align*}
  
  as having type $\T(\tau_2, \Theta)$.

  The last application is a series of multiple applications so technically,
  it should be replaced with a stack of let bindings with temporary names for intermediate step --- we skip that this time for clarity, as the application itself is pretty routine.
  
  From the inductive hypothesis we get that $\Gamma \vdash \E(e, \Theta) : \T(\tau_1, \Theta)$.
  
  From the restrictions placed on our variant of the calculus, we know that all pattern matching is only performed on GADTs,
  so $\tau_1 = (\delta_1, \dots, \delta_m) T$ for some types $\vec{\delta}$ and some GADT name~$T$.
  Thus, we have  $\Gamma \vdash \E(e, \Theta) : \env.T \wedge \{ A_1 = \T(\delta_1, \Theta) \} \wedge \ldots$.
  
  From that and through routine usage of \textsc{Let} we have $v : \env.T \wedge \{ A_1 = \T(\delta_1, \Theta) \} \wedge \ldots$.
  Thus, $v.pmatch$ has type $\forall(r: \{R: \bot..\top \}) \forall(case_{c_1}: \forall(arg: \env.T_{c_1} \wedge s.type) r.R) \ldots r.R$.
  So by rather standard application of \textsc{All-E} we will have 
  \[
  v.pmatch \; tl : \forall(case_{c_1}: \forall(arg: \env.T_{c_1} \wedge s.type) tl.R) \ldots tl.R
  \]
  Thus, by continuously using \textsc{All-E}, assuming that each $case_{c_i}$ has the expected type (we will show that in a moment),
  we get that the whole application types to $tl.R$.
  Since we have $tl.R =:= \T(\tau_2, \Theta)$ by the \textsc{Sel} rules, we get exactly what we needed.

  Now we just need to verify that for each branch, the generated lambda has the correct type.
  
  Let's consider the branch $i$.
  From the assumptions of \textbf{ty-case}, we have $\Delta; \Pi \vdash (p_i \Rightarrow e_i) : (\tau_1 \Rightarrow \tau_2)$
  which means in particular that $\Delta \vdash p \downarrow \tau_1 \Rightarrow (\Delta', \Pi')$
  and that $\Delta, \Delta'; \Pi, \Pi' \vdash e_i : \tau_2$.
  From the restrictions, we know that $p_i$ must have the form $c_i[\beta_{i, 1}, \ldots, \beta_{i, m_i}](x_i)$
  and so that the first judgment must have been derived by \textbf{pat-cons'}. %
  
  From that, we get $\Sigma(c_i) = \forall \vec{\beta_i}. \tau_C \to (\sigma_{i,1}, \ldots, \sigma_{i, m_i}) T$.
  Note that the $\sigma_{i,j}$ here are the same as in the definition of $\env.T_{c_i}$, this will be important in a moment.
  More importantly, because we were matching type $(\vec{\delta}) T$, we know that $\Delta' = \vec{\beta_i}, \vec{\sigma_i} \equiv \vec{\delta}$ and $\Pi' = x_i : \tau_C$.
  Thus the second judgment has form $\Delta, \vec{\beta_i}, \vec{\sigma_i} \equiv \vec{\delta}; \Pi, x_i : \tau_C \vdash e_i : \tau_2$.
  
  We want to show that
  \[
  \lambda(arg_i: \env.T_{c_i} \wedge v.type) \;\; \letp{x'_i}{arg_i.data}{ 
    \E(e_i, \Theta[x_i \mapsto x'_i, \beta_{i,1} \mapsto arg.B_{i,1}, \ldots])
  }\]
 types to $\T(\tau_2, \Theta)$. 
  We can apply \textsc{All-I} followed by \textsc{Let}, knowing that $arg_i.data : \T(\tau_C, \hat{\Theta})$ where 
  \[
  \hat{\Theta} \triangleq [\beta_{i,1} \mapsto arg_i.B_{i,1}, ...]
  \]
  to get to the goal\footnote{Technically in the meantime our environment accumulated some more identifiers, for example earlier case branch definitions are visible in the latter ones - we can ignore them though because independent branches do not rely on them and can be typechecked without them, and then add them using the weakening lemma.}
  \[
  \Gamma' \vdash \E(e_i, \Theta') : \T(\tau_2, \Theta)
  \]
  where \[
  \Gamma' \triangleq \Gamma, v : \env.T \wedge \{ A_1 = \T(\delta_1, \Theta) \} \wedge \ldots, arg_i : \env.T_{c_i} \wedge v.type, x'_i : \T(\tau_C, \hat{\Theta})
  \]
  \[
  \Theta' \triangleq \Theta[x_i \mapsto x'_i, \beta_{i,1} \mapsto arg.B_{i,1}, \ldots]
  \]
  
  Now, as long as we can ensure that $\Phi(\Delta, \vec{\beta_i}, \vec{\sigma_i} \equiv \vec{\delta}; \Pi, x_i : \tau_C; \Theta'; \Gamma')$ we can finish by applying the inductive hypothesis from the assumption 
  \[ 
  \Delta, \vec{\beta_i}, \vec{\sigma_i} \equiv \vec{\delta}; \Pi, x_i : \tau_C \vdash e_i : \tau_2
  \]
  giving us 
  \[
  \Gamma' \vdash \E(e_i, \Theta') : \T(\tau_2, \Theta')
  \]
    
  Trivially, $\T(\tau_2, \Theta') = \T(\tau_2, \Theta)$, because we can assume all the added identifiers $\beta_{i,j}$ are fresh in $\tau_2$.
 
  So now we just need to show $\Phi(\Delta, \vec{\beta_i}, \vec{\sigma_i} \equiv \vec{\delta}; \Pi, x_i : \tau_C; \Theta'; \Gamma')$ to finish the proof.
  
  We already know that $\Phi(\Delta; \Pi; \Theta; \Gamma)$. What remains to be shown is that the conditions of $\Phi$ hold for the added parts:
  \begin{itemize}
    \item For each $\beta_{i,j}$ we know that $\Theta'(\beta_{i,j}) = arg_i.B_{i,j}$ and since $arg_i : \env.T_{c_i}$, we can get\footnote{We do this based on the definition of $\env.T_{c_i}$ in $\Gamma'$; we need to use \textsc{Rec-E} and \textsc{And\textsubscript{i}-<:} to get there.} $arg_i : \{ B_{i,j}: \bot..\top \}$.
    
    \item We need to show that $\Gamma' \vdash \Theta'(x_i) : \T(\tau_C, \Theta')$ which is equivalent to $\Gamma' \vdash x'_i : \T(\tau_C, \Theta')$ but since $\T(\tau_C, \Theta') = \T(\tau_C, \hat{\Theta})$\footnote{Because $\text{fv}(\tau_C) = \{\beta_{i, 1}, ..., \beta_{i, m_i} \}$ and for each $j$, $\Theta'(\beta_{i,j}) = \hat{\Theta}(\beta_{i,j})$.}, it follows directly from \textsc{Var}.
    
    \item For each $\sigma_{i,j} \equiv \delta_j$, we need to show that $\Gamma' \vdash \T(\sigma_{i,j}, \Theta') =:= \T(\delta_j, \Theta')$. For the same reason as previously, $\T(\sigma_{i,j}, \Theta') = \T(\sigma_{i,j}, \hat{\Theta})$. Now, because $arg_i : \env.T_{c_i}$, we have $arg_i : \{ A_j = \T(\sigma_{i,j}, \hat{\Theta}) \}$. We also have already mentioned that $v : \{ A_j = \T(\delta_j, \Theta) \}$. With that, we can prove the desired equality: 
    \[
    \T(\sigma_{i,j}, \Theta') = \T(\sigma_{i,j}, \hat{\Theta}) =:= arg_i.A_j =:= v.A_j =:= \T(\delta_j, \Theta)
    \]
    The two outer equalities rely on the \textsc{Sel} rules and the inner equality is derived by the two \textsc{Sngl} rules, because we have $arg_i : v.type$.

  \end{itemize}
 
  \end{itemize}
\end{proof}

\begin{lemma}[Substitution through translation]
  \label{lemma:substTranslation}
  For every $\Delta$, $\alpha$, $\tau_1$ and $\tau_2$, if $\Delta \vdash \tau_1 : *$ and 
  \mbox{$\Delta, \alpha \vdash \tau_2 : *$},
  then for every $\Gamma$ and $\Theta$ such that $\Psi(\Delta; \Theta; \Gamma)$, if $\Gamma \vdash a.T =:= \T(\tau_1, \Theta)$ for some $a$, then for every fresh $x_{\alpha}$, we have $\Gamma \vdash \T(\tau_2, \Theta[\alpha \mapsto x_{\alpha}.T])[a/x_{\alpha}] =:= \T(\tau_2[\alpha \mapsto \tau_1], \Theta)$.
\end{lemma}
\begin{proof}
  The proof goes by induction on the structure of $\tau_2$.
  \begin{itemize}
    \item $\tau_2 = \mathbf{1}$: the equality stems trivially from the \textsc{Refl} rule, because $\T(\mathbf{1}, \ldots) = \texttt{lib.Unit}$, so both sides are syntactically equal.
    
    \item $\tau_2 = \tau_A * \tau_B$: from the inductive hypotheses we have $\Gamma \vdash \T(\tau_A, \Theta[\alpha \mapsto x_{\alpha}.T])[a/x_{\alpha}] =:= \T(\tau_A[\alpha \mapsto \tau_1], \Theta)$ and the same for $\tau_B$.
    
    We can transform both sides of the equation:
    \begin{align*}
    &\T(\tau_A * \tau_B, \Theta[\alpha \mapsto x_{\alpha}.T])[a/x_{\alpha}] = \\
    &\quad\quad (\texttt{lib.Tuple} \wedge \{ T_1 = \T(\tau_A, \Theta[\alpha \mapsto x_{\alpha}.T])[a/x_{\alpha}] \} \wedge \{ T_2 = \T(\tau_B, \Theta[\alpha \mapsto x_{\alpha}.T])[a/x_{\alpha}] \} )
    \\
    &\T((\tau_A * \tau_B)[\alpha \mapsto \tau_1], \Theta) = \\ 
    & \quad\quad \T(\tau_A[\alpha \mapsto \tau_1] * \tau_B[\alpha \mapsto \tau_1], \Theta) = \\
    & \quad\quad \texttt{lib.Tuple} \wedge \{T_1 = \T(\tau_A[\alpha \mapsto \tau_1], \Theta) \} \wedge \{T_2 = \T(\tau_B[\alpha \mapsto \tau_1], \Theta) \}
    \end{align*}
    
    By using the rules \textsc{Typ-<:-Typ}, from the inductive hypotheses we get 
    \[
    \Gamma \vdash \{T_1 = \T(\tau_A[\alpha \mapsto \tau_1], \Theta) \} =:= \{ T_1 = \T(\tau_A, \Theta[\alpha \mapsto x_{\alpha}.T])[a/x_{\alpha}] \}
    \]
    and the same for $\tau_B$, then by \autoref{lemma:6} (applied twice), we get the desired result.
    
    \item The case $\tau_2 = (\vec{\tau}) T$ goes analogously as for the tuples - we get inductive hypotheses for each type in $\vec{\tau}$, we can push the substitution down the structure and get the final equality by applying \autoref{lemma:6} enough times.
    
    \item $\tau_2 = \tau_A \to \tau_B$ goes also analogously as above, but instead we use the \textsc{All-<:-All} rule together with weakening.
    
    \item $\tau_2 = \forall \alpha'. \tau$: 
    We can unfold the lhs:
    \[
    \T(\forall \alpha'. \tau, \Theta[\alpha \mapsto x_{\alpha}.T])[a/x_{\alpha}] = \forall(x_{\alpha'} : \{T: \bot..\top\}) \T(\tau, \Theta[\alpha \mapsto x_{\alpha}.T, \alpha' \mapsto x_{\alpha'}.T])[a/x_{\alpha}]
    \]
     where $x_{\alpha'}$ is some fresh identifier.
    Let's consider two cases:
    \begin{itemize}
      \item $\alpha' = \alpha$:
      In this case the lhs further simplifies\footnote{
	We can get rid of the $[a/x_{\alpha}]$ substitution at the end, because if the $\mapsto x_{\alpha}.T$ disappeared, then from the freshness of $x_{\alpha}$ for $\Theta$, we know that this substitution will act as identity.} to 
      $\forall(x_{\alpha'} : \{T: \bot..\top\}) \T(\tau, \Theta[\alpha \mapsto x_{\alpha'}.T])$.
      
      The rhs can be unfolded\footnote{Again, we are free to remove the substitution $[\alpha \mapsto \tau_1]$ inside of \T, because the name is aliased by the same name in $\forall \alpha$.} to $\T(\forall \alpha. \tau, \Theta)$, and from the definition of \T, this is exactly the same as the transformed lhs.
      
      \item $\alpha' \not= \alpha$:
      
      In the rhs we can push the substitution deeper:
      \[
      \T((\forall \alpha'. \tau)[\alpha \mapsto \tau_1], \Theta) = 
      \T(\forall \alpha'. (\tau[\alpha \mapsto \tau_1]), \Theta) =
      \forall(x_{\alpha'} : \{T: \bot..\top\}) \T(\tau[\alpha \mapsto \tau_1], \Theta[\alpha' \mapsto x_{\alpha'}.T])
      \]
      
      We can reorder the substitution in the lhs to get:
      \[
      \forall(x_{\alpha'} : \{T: \bot..\top\}) \T(\tau, \Theta[\alpha' \mapsto x_{\alpha'}.T, \alpha \mapsto x_{\alpha}.T])[a/x_{\alpha}]
      \]
      
      After applying the \textsc{All-<:-All} rule to get the subtyping in both directions, it suffices to prove
      \[
      \Gamma, x_{\alpha'} : \{T: \bot..\top\} \vdash \T(\tau[\alpha \mapsto \tau_1], \Theta') =:= \T(\tau, \Theta'[\alpha \mapsto x_{\alpha}.T])[a/x_{\alpha}]
      \] 
      where
      \[
      \Theta' \triangleq \Theta[\alpha' \mapsto x_{\alpha'}.T]\]
      Clearly, $\Psi(\Delta, \alpha'; \Theta'; \Gamma, x_{\alpha'} : \{T: \bot..\top\})$ so we get the desired equality by the inductive hypothesis.

    \end{itemize}

    \item Finally, $\tau_2 = \alpha'$. Let's consider two cases:
    \begin{itemize}
      \item $\alpha' = \alpha$: The lhs is $\T(\alpha, \Theta[\alpha \mapsto x_{\alpha}.T])[a/x_{\alpha}] = (x_{\alpha}.T)[a/x_{\alpha}] = a.T$ and the rhs is \mbox{$\T(\alpha[\alpha \mapsto \tau_1], \Theta) = \T(\tau_1, \Theta)$}, so we just use the assumption $\Gamma \vdash a.T =:= \T(\tau_1, \Theta)$.
      
      \item $\alpha' \not= \alpha$: The lhs is $\T(\alpha', \Theta[\alpha \mapsto x_{\alpha}.T])[a/x_{\alpha}] = \Theta(\alpha')[a/x_{\alpha}] = \Theta(\alpha')$ (we can do the last transformation, because we know that $x_{\alpha}$ was fresh for $\Theta$). The rhs is $\T(\alpha'[\alpha \mapsto \tau_1], \Theta) = \Theta(\alpha')$, so both sides are equal: we get the type equality by \textsc{Refl}.
    \end{itemize}
  \end{itemize}
\end{proof}

\begin{lemma}
\label{lemma:equalSubstChange}
For any type $\tau$ and any substitutions $\Theta_1$ and $\Theta_2$, if $\text{fv}(\tau) \subseteq \text{dom}(\Theta_1) = \text{dom}(\Theta_2)$ and for every $x \in \text{dom}(\Theta_1)$, $\Gamma \vdash \Theta_1(x) =:= \Theta_2(x)$, then also $\Gamma \vdash \T(\tau, \Theta_1) =:= \T(\tau, \Theta_2)$.
\end{lemma}
\begin{proof}
  Easy induction on the structure of $\tau$. The base case is for variables and it goes directly from the assumption $\Gamma \vdash \Theta_1(x) =:= \Theta_2(x)$. All the other cases are trivial or can be handled by constructing equations on terms based on equations on the sub-terms derived from the inductive hypotheses, in the same way as in \autoref{lemma:substTranslation}.
\end{proof}

\begin{lemma}
  \label{lemma:substEqual}
  For any types $\tau_A$ and $\tau_B$, any variable $\alpha$, substitution $\Theta$ and environment $\Gamma$, 
  if~\mbox{$\Gamma \vdash \Theta(\alpha) =:= \T(\tau_A, \Theta)$}
  then $\Gamma \vdash \T(\tau_B, \Theta) =:= \T(\tau_B[\alpha \mapsto \tau_A], \Theta)$.
\end{lemma}
\begin{proof}
  Easy induction on the structure of $\tau_B$.
  
   The base case for the variable $\alpha$ goes directly from the assumption $\Gamma \vdash \Theta(\alpha) =:= \T(\tau_A, \Theta)$. For any $\alpha' \not= \alpha$, we have $\T(\alpha', \Theta) = \T(\alpha'[\alpha \mapsto \tau_A], \Theta)$, so it goes by \textsf{Refl}.
   
   All the other cases are trivial or can be handled by constructing equations on terms based on equations on the sub-terms derived from the inductive hypotheses, in the same way as in \autoref{lemma:substTranslation}.
   
   For example in the case of tuples, i.e. if $\tau_B = \tau_1 * \tau_2$, from the inductive hypotheses we have 
   \[
   \Gamma \vdash \T(\tau_1, \Theta) =:= \T(\tau_1[\alpha \mapsto \tau_A], \Theta)
   \]
   and same for $\tau_2$, then 
   \[ 
   \T(\tau_1 * \tau_2, \Theta) = \texttt{lib.Tuple} \wedge \{T_1 = \T(\tau_1, \Theta) \} \wedge \{T_2 = \T(\tau_2, \Theta) \}
   \]
    and 
   \[
   \T(\tau_1[\alpha \mapsto \tau_A], \Theta) = \texttt{lib.Tuple} \wedge \{T_1 = \T(\tau_1, \Theta)[\alpha \mapsto \tau_A] \} \wedge \{T_2 = \T(\tau_2, \Theta)[\alpha \mapsto \tau_A] \}
   \]
   and so we construct the necessary equation using \textsc{Typ-<:-Typ} and \autoref{lemma:6}.
   
   One more interesting case that may be worth showing is $\tau_B = \forall \alpha'. \tau$. Let's see two cases:
   \begin{itemize}
     \item if $\alpha = \alpha'$, then $(\forall \alpha. \tau)[\alpha \mapsto \tau_A] = \forall \alpha. \tau$ so the equality follows trivially from \textsc{Refl}.
     \item $\alpha \not=\alpha'$, then 
     \[
     \T(\forall \alpha'. \tau, \Theta) = \forall(x_{\alpha'}: \{ T: \bot..\top \}) \T(\tau, \Theta[\alpha' \mapsto x_{\alpha'}.T])
     \]
      and 
      \[
      \T((\forall \alpha'. \tau)[\alpha \mapsto \tau_A], \Theta) = \T(\forall \alpha'. (\tau[\alpha \mapsto \tau_A]), \Theta) = \forall(x_{\alpha'}: \{ T: \bot..\top \}) \T(\tau[\alpha \mapsto \tau_A], \Theta[\alpha' \mapsto x_{\alpha'}.T])
      \] 
     
     Of course, if $\Gamma \vdash \Theta(\alpha) =:= \T(\tau_A, \Theta)$, then by weakening also
    \[ \Gamma, x_{\alpha'}: \{ T: \bot..\top \} \vdash \Theta(\alpha) =:= \T(\tau_A, \Theta) \]
    and because $\alpha'$ is fresh for $\tau_A$, $\T(\tau_A, \Theta) = \T(\tau_A, \Theta[\alpha' \mapsto x_{\alpha'}.T])$ and $(\Theta[\alpha' \mapsto x_{\alpha'}.T])(\alpha) = \Theta(\alpha)$, so from the inductive hypothesis we get 
    \[
    \Gamma, x_{\alpha'}: \{ T: \bot..\top \} \vdash \T(\tau, \Theta[\alpha' \mapsto x_{\alpha'}.T]) =:= \T(\tau[\alpha \mapsto \tau_A], \Theta[\alpha' \mapsto x_{\alpha'}.T]) 
    \]
    
    Now we use \textsc{All-<:-All} to get the desired equation.
   \end{itemize}
\end{proof}

Using the \autoref{thm:typepreservation} together with \autoref{lemma:sigmaPrime} we finally can prove that the result of \texttt{EncodeProgram} is well-typed.

\section{Soundness Proof of \pdot{}}
\label{sec:soundness-proof-of-pdot}

\citeauthor{rapoportPathDOTFormalizing2019} introduces the seven-level stratified typing rules in the soundness of \pdot{}.
\begin{align*}
  {\tiny
  \text{General } (\vdash) \rightarrow
  \text{Tight } (\turnstileTight)\rightarrow
  {\text{{Introduction-$qp$ }} (\turnstileRepl)} \rightarrow
  \text{Introduction-$pq$ }(\turnstileInvertible) \rightarrow
  {\text{Elim-III }(\turnstilePrecThree)} \rightarrow
  {\text{Elim-II }(\turnstilePrecTwo)} \rightarrow
  \text{Elim-I }(\turnstilePrecOne)}
\end{align*}
These levels in \pdot{} are designed and organized to tackle two key problems in the metatheory: bad bounds and cycles in typing derivations.
We will explain the two problems along with \pdot{}'s solutions to them one by one.

\FloatBarrier
\begin{wide-rules}
\begin{multicols}{5}
\infax{\inert{\tForall x T U}}
\infrule
    {\record{T}}
    {\inert{\tRec x T}}
\infax{\record{\tTypeDec A T T}}
\infax{\record{\tFldDec a {\single q}}}
\infrule
    {{\inert T}}
    {\record{\tFldDec a T}}
\end{multicols}
  \caption{Inert Types}
  \label{fig:inertness}
\end{wide-rules}
\begin{wide-rules}

\begin{multicols}{3}
    
\infax[Top$_{\#}$]
  {\subTightDft T \top}

\infax[Bot$_{\#}$]
  {\subTightDft \bot T}

\infax[Refl$_{\#}$]
  {\subTightDft T T}

\infrule[Trans$_{\#}$]
  {\subTightDft S T
    \andalso
    \subTightDft T U}
  {\subTightDft S U}

\infax[And$_1$-$<:$$_{\#}$]
  {\subTightDft {\tAnd T U} T}

\infax[And$_2$-$<:$$_{\#}$]
  {\subTightDft {\tAnd T U} U}

\infrule[$<:$-And$_{\#}$]
  {\subTightDft S T
    \andalso
    \subTightDft S U}
  {\subTightDft S {\tAnd T U}}
  
\infrule[$<:$-Sel$_{\#}$]
  {{\typPrecThreeDft p {\tTypeDec A S S}}}
  {\subTightDft S {{p}.A}}

\infrule[Sel-$<:$$_{\#}$]
  {{\typPrecThreeDft p {\tTypeDec A S S}}}
  {\subTightDft {{p}.A} S}
  
\infrule[Sngl$_{pq}$-$<:$$_{\#}$]
  {\typPrecThreeDft p {\single q} \andalso \typeablePrecTwo q}
  {\subTightDft T {\repl p q T}}

\infrule[Sngl$_{qp}$-$<:$$_{\#}$]
  {\typPrecThreeDft p {\single q} \andalso \typeablePrecTwo q}
  {\subTightDft T {\repl q p T}}

\infrule[Fld-$<:$-Fld$_{\#}$]
  {\subTightDft T U}
  {\subTightDft {\tFldDec a T} {\tFldDec a U}}

\infrule[Typ-$<:$-Typ$_{\#}$]
  {\subTightDft {S_2} {S_1}
    \andalso
    \subTightDft {T_1} {T_2}}
  {\subTightDft {\tTypeDec A {S_1} {T_1}} {\tTypeDec A {S_2} {T_2}}}

\infrule[All-$<:$-All$_{\#}$]
  {\subTightDft {S_2} {S_1}
    \\
    \sub {\extendG x {S_2}} {T_1} {T_2}}
  {\subTightDft {\tForall x {S_1} {T_1}} {\tForall x {S_2} {T_2}}}

  \end{multicols}
    
  \caption{Tight subtyping rules of \pdot{}\cite{rapoportPathDOTFormalizing2019} and \calculus{}.}
\label{fig:tight-subtyping-rules}

\end{wide-rules}

\paragraph{Bad bounds.}
In DOT calculi, users can introduce arbitary subtyping relations using type members, and they can be unsound.
For example:
\[\lambda(x: \tTypeDec A \top \bot) t\]
When typing the body of the above function, the absurd subtyping judgment $\subDft \top \bot$ is derivable in the environment.
DOT calculi are designed in a way that bad bounds will not break soundness.
Specifically, when evaluating the program, the type member $x.A$ will be instantiated with a concrete type,
serving as the proof of the subtyping relation implied by its bounds.
Therefore, in the above example, the body of the function will never get executed
because it is impossible to supply such a $x$.
Tight typing are proposed to deal with the bad bounds problem in the soundness proof.
Its main difference of it from general typing is that it forbids type members from introducing custom subtyping relationships by modifying the \rn{$<:$-Sel} and \rn{Sel-$<:$} rules.
\begin{multicols}{2}
  \infrule[$<:$-Sel-\#]
  {\typPrecThreeDft {p} \new{\tTypeDec A T T}}
  {\subTightDft T {{p}.A}}

  \infrule[Sel-$<:$-\#]
  {\typPrecThreeDft {p} \new{ \tTypeDec A T T}}
  {\subTightDft {p.A} T}
\end{multicols}
The full formulation of tight subtyping rule is presented in \autoref{fig:tight-subtyping-rules}.
In such a way, in the tight typing level, subtyping relations can always be trusted.
But how can we transform general typing to the tight typing level?
\pdot{} introduces the notion of \emph{inertness}.
We can prove the following lemma, which states that general typing can be transformed to the tight level without loss of expressiveness in an inert environment.
\begin{lemma}[$\vdash$ to $\vdash_{\#}$] \label{thm:general-to-tight}
  If $\G$ is inert and $\typDft t T$, then $\typTightDft t T$.
\end{lemma}

\paragraph{Cycles in typing derivation.}
\pdot{} typing rules can result in cycles in typing derivation.
For example, we can apply \rn{Rec-I} and \rn{Rec-E} repeatedly in the derivation tree.
And we can similarly produce cycles with the \rn{Fld-I}, \rn{Fld-E}, \rn{Sngl$_{pq}$-<:} and \rn{Sngl$_{qp}$-<:} rules.
Such cycles hinders the inductive reasoning of typing judgments.
\pdot{} deals with the cycles by seperating each pair of the symmetric typing rules into different levels.
For instance, the \rn{Sngl$_{pq}$-<:} and \rn{Sngl$_{qp}$-<:} rules are inlined in
Introduction-$qp$ ($\turnstileRepl$)
Introduction-$pq$ ($\turnstileInvertible$) levels respectively.
Starting from the Introduction-$qp$ level, \pdot{} separates the reasoning about values and paths.
Therefore, we have two sets of typing rules, one for the values and another for the paths on these levels.
Transformation lemmas are proven for them respectively.

\end{document}